\documentclass[11pt, letterpaper]{article}

\usepackage{fullpage}

\usepackage[utf8]{inputenc} %
\usepackage[T1]{fontenc}    %
\usepackage[breaklinks]{hyperref}
\usepackage[svgnames]{xcolor}         %
\hypersetup{colorlinks={true},urlcolor={blue},linkcolor={DarkBlue},citecolor=[named]{DarkGreen}}
\usepackage{url}            %
\usepackage{booktabs}       %
\usepackage{amsfonts}       %
\usepackage{nicefrac}       %
\usepackage{microtype}      %
\usepackage{subcaption}
\usepackage{multirow}
\usepackage{enumitem} %
\usepackage{xspace}
\usepackage{amsthm, amsxtra}
\usepackage{bm}%
\usepackage{tikz}
\usepackage{algorithm}
\usepackage{algpseudocode}
\usepackage{thm-restate} %
\usepackage{tabularx}
\usepackage{mathtools}
\usepackage{thmtools}
\usepackage{lmodern} 
\usepackage{authblk}
\usepackage[authoryear, square]{natbib}

\setlist[itemize]{leftmargin=*,noitemsep,topsep=0pt}
\setlist[enumerate]{leftmargin=*,noitemsep,topsep=0pt}

\DeclareMathOperator*{\E}{\mathbb{E}}
\DeclareMathOperator*{\Var}{Var}

\newtheorem{problem}{Problem}
\newtheorem{theorem}{Theorem}[section]
\newtheorem{corollary}{Corollary}[section]
\newtheorem{lemma}[theorem]{Lemma}
\newtheorem{definition}{Definition}[section]

\DeclareMathOperator{\sign}{sgn}

\newcommand{\Rnd}{\ensuremath{\mathbf{SoS}}\xspace} %
\newcommand{\SdpRnd}{\ensuremath{\mathbf{SDP}}\xspace} %
\newcommand{\SdpRatio}{\ensuremath{\lambda}\xspace} %
\newcommand{\SdpRatioAvg}{\ensuremath{\varphi}\xspace} %

\newcommand{\Exp}[1]{\mathbb{E}\left[ #1 \right]}

\newcommand{\pr}{\Pr}
 
\newcommand{\Real}{\ensuremath{\mathbb{R}}\xspace} %

\newcommand{\OPT}{\ensuremath{\mathbf{Opt}}\xspace}
\newcommand{\Sol}{\ensuremath{\mathbf{Sol}}\xspace}

\newcommand{\vecv}{\bm{v}\xspace}
\newcommand{\vecx}{\bm{x}\xspace}

\newcommand{\vecr}{\bm{r}\xspace}
\newcommand{\vecone}{\bm{1}\xspace}

\newcommand{\algo}{\mathcal{A}}
\newcommand{\algoB}{\mathcal{B}}

\DeclareMathOperator*{\argmin}{arg\,min}

\newcommand{\abs}[1]{\lvert #1\rvert}

\renewcommand{\lg}{\log}

\DeclareRobustCommand{\OPT}{%
	\ifmmode
		\operatorname{\bf Opt}
	\else
		\textbf{Opt}\xspace
	\fi
}
\DeclareRobustCommand{\ALG}{%
	\ifmmode
		\operatorname{ALG}
	\else
		\text{ALG}\xspace
	\fi
}

\newcommand{\sbpara}[1]{{\smallskip\noindent\textbf{#1}}}

\newcommand{\inputtikz}[1]{%
  \includegraphics{figures/#1}%
}

\newcommand{\parzero}{\ensuremath{c_0}\xspace}
\newcommand{\parone}{\ensuremath{c_1}\xspace}
\newcommand{\partwo}{\ensuremath{c_2}\xspace}
\newcommand{\parthree}{\ensuremath{c_3}\xspace}

\newcommand{\vertices}{\ensuremath{V}\xspace}

\newcommand{\bigO}{\ensuremath{\mathcal{O}}\xspace}
\newcommand{\bigOtilde}{\ensuremath{\Tilde{\mathcal{O}}}\xspace}
\newcommand{\dks}{{\small D}\ensuremath{k}{\small S}\xspace}
\newcommand{\dksplus}{{\small D}\ensuremath{(k\!+\!1)}{\small S}\xspace} %
\newcommand{\dskc}{\text{{\small DS}-}\ensuremath{k}{\small R}\xspace}
\newcommand{\maxcut}{{\small M}\textsc{ax}\text{-}{\small C}\textsc{ut}\xspace}
\newcommand{\maxcutkc}{{\maxcut}\text{-}\ensuremath{k}{\small R}\xspace}

\newcommand{\ccmaxcut}{{\maxcut}\text{-}{\small CC}\xspace}
\newcommand{\maxbis}{{\small M}\textsc{ax}\text{-}{\small B}\textsc{isection}\xspace}

\newcommand{\kdense}{\ensuremath{k}\text{-}{{\small D}\textsc{ensify}}\xspace}
\newcommand{\sdp}{{\small SDP}\xspace}

\newcommand{\ccmaxuncut}{{\small M}\textsc{ax}\text{-}{\small U}\textsc{ncut}\text{-}{\small CC}\xspace} %
\newcommand{\ccvc}{{\small VC}\text{-}{\small CC}\xspace} %
\newcommand{\maxuncutkc}{{\small M}\textsc{ax}-{\small U}\textsc{ncut}-\ensuremath{k}{\small R}\xspace} %
\newcommand{\vckc}{{\small VC}-\ensuremath{k}{\small R}\xspace} %

\newcommand{\gpkc}{{\small GP}-\ensuremath{k}{\small R}\xspace} %
\newcommand{\maxgp}{{\small MAX}-{\small GP}\xspace} %

\newcommand{\spara}[1]{\smallskip\noindent\textbf{#1}}
\newcommand{\para}[1]{\noindent\textbf{#1}}

\newcommand{\SelectSet}{\ensuremath{C}\xspace} %
\newcommand{\SelectComSet}{\ensuremath{\overline{\SelectSet}}\xspace} %
\newcommand{\FixSelectSet}{\ensuremath{C}^{''}\xspace} %
\newcommand{\OldSet}{\ensuremath{U}\xspace} %
\newcommand{\TempSet}{\ensuremath{\OldSet'}\xspace}%
\newcommand{\FixTempSet}{\ensuremath{\OldSet^{''}}\xspace}%
\newcommand{\OldComSet}{\ensuremath{\overline{\OldSet}}\xspace} %
\newcommand{\TempComSet}{\ensuremath{\overline{\OldSet'}}\xspace}%

\newcommand{\symm}{\triangle}

\newcommand{\cutnode}[2]{\mathrm{cut}_{#1}\!\left(#2\right)}

\makeatletter
\newcommand{\cut}{\@ifnextchar\bgroup\cut@i{\mathrm{cut}}}
\newcommand{\cut@i}[1]{\ensuremath{\mathrm{cut}\!\left(#1\right)}}

\newcommand{\partition}{\@ifnextchar\bgroup\partition@i{\mathcal{G}}}
\newcommand{\partition@i}[1]{\ensuremath{\mathcal{G}\!\left(#1\right)}}

\makeatother

\newcommand{\Greedy}{\textsf{\small Greedy}\xspace}
\newcommand{\Peel}{\textsf{\small Peel}\xspace}
\newcommand{\SDPalgo}{\textsf{\small SDP}\xspace}
\newcommand{\Random}{\textsf{\small Rnd}\xspace}
\newcommand{\Local}{\textsf{\small Local}\xspace}
\newcommand{\Init}{\textsf{\small Init}\xspace}

\newcommand{\Blackbox}{\textsf{\small B:}\xspace}

\newcommand{\TypeOne}{\textsf{\small :I}\xspace}
\newcommand{\TypeTwo}{\textsf{\small :II}\xspace}

\newcommand{\denseGreedy}{\Greedy}
\newcommand{\denseSDPalgo}{\SDPalgo}
\newcommand{\denseSDPMerge}{\Blackbox\SDPalgo}
\newcommand{\densePeelMerge}{\Blackbox\Peel}
\newcommand{\denseSQD}{\textsf{\small SQD}\xspace}
\newcommand{\denseinit}{\Init}
\newcommand{\denserandom}{\Random}

\newcommand{\cutGreedy}{\Greedy}
\newcommand{\cutSDPalgo}{\SDPalgo}
\newcommand{\cutBlackSDP}{\Blackbox\SDPalgo}
\newcommand{\cutBlackGreedy}{\Blackbox\Greedy}
\newcommand{\cutBlackLocalOne}{\Blackbox\Local\TypeOne}
\newcommand{\cutBlackLocalTwo}{\Blackbox\Local\TypeTwo}

\newcommand{\balanced}{\textsf{\small SBM-Balanced}\xspace}
\newcommand{\dense}{\textsf{\small SBM-DenseSubg}\xspace}
\newcommand{\sparse}{\textsf{\small SBM-SparseSubg}\xspace}
\newcommand{\gb}{\textsf{\small Wiki-GB}\xspace}
\newcommand{\de}{\textsf{\small Wiki-DE}\xspace}
\newcommand{\es}{\textsf{\small Wiki-ES}\xspace}
\newcommand{\us}{\textsf{\small Wiki-US}\xspace}
\newcommand{\dblp}{\textsf{\small SNAP-DBLP}\xspace}
\newcommand{\youtube}{\textsf{\small SNAP-Youtube}\xspace}
\newcommand{\amazon}{\textsf{\small SNAP-Amazon}\xspace}

\usepackage{xcolor}
\newcommand{\rewrite}[1]{{#1}}

\title{\bf OptiRefine: Densest subgraphs and maximum cuts\\ with $k$ refinements}

\date{}

\author[1]{Sijing Tu}
\author[1,2]{Aleksa Stankovic} 
\author[3]{Stefan Neumann} 
\author[1]{Aristides Gionis}

\affil[1]{KTH Royal Institute of Technology}
\affil[2]{Qubos Systematic}
\affil[3]{TU Wien}

\begin{document}
\maketitle 

\begin{abstract}
Data-analysis tasks often involve an iterative process, which requires refining previous solutions. For instance, when analyzing social networks, we may obtain initial
	communities based on noisy metadata, and we want to improve them by adding
	influential nodes and removing non-important ones, without making too many
	changes.
However, classic optimization algorithms, which typically find solutions from scratch, potentially
return communities that are very dissimilar to the initial one. 
To mitigate these issues, we introduce the \emph{OptiRefine framework}.
The framework optimizes initial solutions by making a small number of \emph{refinements}, thereby ensuring that the new solution remains close to the
initial solution and simultaneously achieving a near-optimal solution for the optimization problem.
We apply the OptiRefine framework to two classic graph-optimization problems:
\emph{densest subgraph} and \emph{maximum cut}. For the \emph{densest-subgraph problem}, we
optimize a given subgraph's density by adding or removing $k$~nodes. We show
that this novel problem is a generalization of $k$-densest subgraph, and provide
constant-factor approximation algorithms for $k=\Omega(n)$~refinements.  
We also study a version of \emph{maximum cut} in which the goal is to improve a given cut.
We provide connections to the maximum cut with cardinality constraints and provide an
optimal approximation algorithm in most parameter regimes under the Unique Games
Conjecture for $k=\Omega(n)$~refinements.
We evaluate our theoretical methods and scalable heuristics on synthetic and
real-world data and show that they are highly effective in practice.
\end{abstract}

\maketitle

\section{Introduction}
\label{sec:intro}

Graphs are commonly used to model entities and their relationships in various application domains. 
Numerous methods have been developed to address diverse graph mining tasks,
including analyzing graph properties, finding patterns, discovering communities, and achieving other application-specific objectives. 
These graph mining tasks are generally formulated as optimization problems. 
Notable examples of such problems include graph-partitioning tasks, 
identifying densest subgraphs or determining the maximum cut of a graph.

In real-world applications, graph mining is often an iterative process that does not start from scratch. 
This process often begins with preliminary and sub\-optimal solutions, {based on
metadata or existing results. Subsequently the solution shall be refined and
improved.  Specifically, we seek to modify the initial solution~$U$ into a new solution~$U'$
while ensuring that the solution~$U'$ must be close enough to~$U$.}
More concretely, consider the following two examples:
\begin{enumerate}
    \smallskip
    \item 
    {Consider the task of detecting the community of users
	in an online social network who are interested in political commentary in a
	given country, say, Germany.
	Whether a user is related to political commentary in Germany can be
	obtained, for instance, based on metadata, such as the location data they
	are sharing and the interests they indicated when they joined the network.}
        Let us call this set of users $U$.
        Some of those users, however, may not be active in the social network
		{or changed their interests}.
		Instead, there might be other users, unknown to us {or from different countries, who are more active
		and also engage in political commentary regarding Germany.}
        We could then seek to identify a set of users $U'$, 
        which is close enough to $U$ (so that the community is still on the same topic), 
        and which has a high graph density 
        (so that it consists of a set of influential users who are tightly connected with each other).
        \smallskip
    \item Next, suppose that a network host aims to maximize the diversity of the
	social network by altering a few users' exposure to diverse news outlets
	(which can be realized, for instance, through recommendation). 
	As formulated by recent work \citep{matakos2020tell}, the exposure of each user can be modeled 
    by a value in the set $\{-1, 1\}$, e.g., representing two sides on a polarized topic. The
    network diversity can be measured by the size of the cut indicated by the partition
    $(U, \bar{U})$, where $U$ consists of the users whose exposure is $-1$
    and $\bar{U}$ consists of the users whose exposure is $1$. 
    {Now the goal of increasing the network diversity corresponds to changing
    the exposure of a small set of users in order to maximize the cut. To ensure
	that not too many users are affected, the network host must ensure that the
	new assignment is close to the original one.}
    \end{enumerate} 
    
\smallskip
The above examples present a
\emph{solution refinement process}. 
Specifically, an initial solution is provided, 
such as an online community in the first example or 
the exposure of all network users to news outlets in the second example.
The objective is to refine the solution
so as to maximize an objective function, 
i.e., either finding a densest subgraph or a maximum cut, 
while maintaining proximity to the initial solution.
The refinement process we consider in this paper appears in many graph-mining tasks.
It naturally generalizes classic optimization algorithms, such as finding the
densest subgraph or the maximum cut of a graph, by incorporating existing initial solutions. 

From an algorithmic standpoint, standard algorithms designed {for classical problems do not easily adapt to these new tasks that involve an \emph{initial solution} and aim for the best \emph{refinement}}. 
To address this, we introduce the \emph{OptiRefine framework}, which models this class of optimization problems. 
{We apply our framework to \emph{graph-partitioning} problems~\citep{DBLP:journals/jal/feigel01}, specifically focusing on the well-known \emph{densest subgraph} and \emph{max-cut} problems. We refer to our problems as \emph{graph-partition with $k$ refinements} (\gpkc), formally defined in Definition~\ref{prob:graph-partition-local}}.

{We apply our framework to two classic problems in data mining and
	combinatorial optimization:}
The \emph{densest-subgraph problem} is commonly used to identify tightly connected groups of entities, with applications such as finding trending stories~\citep{angel2012dense}, detecting bots and fraudulent content in social media~\citep{beutel2013copycatch}, and identifying correlated genes in gene-expression data~\citep{saha2010dense}, among others. 
The \emph{max-cut problem} has applications in graph-partitioning settings~\citep{ding2001min}, in discovering conflicting communities in social networks~\citep{bonchi2019discovering}, 
and in measuring the diversity of exposure to news outlets in online social networks~\citep{matakos2020tell}. 
The methods we introduce in this paper are particularly useful when a community has been identified and there is a need to find another community that is \emph{close} to the current one.

We identify a close relationship between (1)~the problems in our OptiRefine framework and (2)~the corresponding problems with cardinality constraints on the solution size.
For both the densest subgraph and max-cut problems {defined under the OptiRefine framework}, we obtain approximation algorithms that {nearly} match the best approximation ratios of algorithms that solve the cardinality-constrained versions of these problems.
We believe that our insights will be helpful in applying the OptiRefine framework to a broad range of problems and in developing practical algorithms with theoretical guarantees for a wide range of applications.

\subsection{Our results}
\label{intro:results}

\begin{figure}[t]
  \centering 
    \begin{tabular}{p{0.45\columnwidth}p{0.45\columnwidth}}
    \resizebox{0.45\columnwidth}{!}{%
      \inputtikz{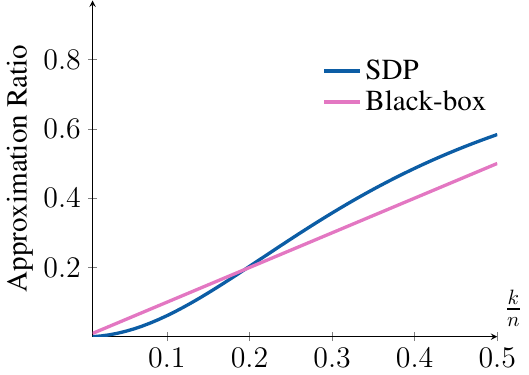}
    }&
    \resizebox{0.45\columnwidth}{!}{%
      \inputtikz{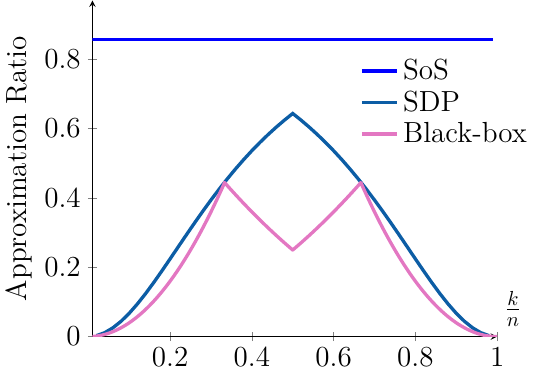}
    }\\[1ex]
    \hspace{-1.3em}
    \begin{minipage}{0.45\columnwidth}
    \centering {\small (a)~\kdense}
    \end{minipage} &
    \begin{minipage}{0.45\columnwidth}
    \centering {\small {(b)~\maxcutkc}}
    \end{minipage} \\
    \end{tabular}
  \caption{\small 
  {
	  The approximation ratios of the algorithms we present as a function
		  of~$k$.
	  The approximation ratio for \dskc equals to the approximation ratio of \kdense times $\frac{1-c}{1+c}$,
  where $c$ is such that $k \leq c\abs{\OldSet}$.
  For \kdense, we assume that the black-box solver is the algorithm for \dks by~\citet{asahiro2000greedily}.
    When $k=n/2$, the approximation ratio is at least $0.583$ with the \sdp algorithm and at least $0.5$ with the black-box solver.
	For \maxcutkc, we assume that the black-box solver is the algorithm for \maxcut by~\citet{goemans1994approximation}
    When $k=n/2$, the approximation ratio is at least $0.643$ with the \sdp algorithm, at least $0.250$ with the black-box solver, and $0.858$ with the SoS algorithm.}
  }
  \label{fig:plots-approx-ratio}
\end{figure}

First, we apply our OptiRefine framework to the classic densest subgraph and max-cut problems. 
For the densest-subgraph problem, we introduce the problem \emph{densest subgraph with $k$~refinements} (\dskc) in which the input consists of a weighted graph $G = (V, E, w)$, an initial subset of vertices $\OldSet \subseteq V$ and an integer $k \in \mathbb{N}$.  
The goal is to add or remove $k$~vertices from $\OldSet$ to obtain a {refined set} of vertices~$U'$ that maximizes the density of its induced subgraph, denoted by $G[U']$ (see Problem~\ref{prob:densest-subgraph-local} for the formal definition).
For the max-cut problem, we introduce the problem of \emph{max-cut with $k$~refinements} (\maxcutkc) in which we are given a weighted graph $G=(V, E, w)$, an initial subset of vertices $\OldSet\subseteq V$ and an integer $k\in\mathbb{N}$.
The goal is to find a set of vertices $U' \subseteq V$ such that $U$ and $U'$ differ by $k$~vertices and the cut $(U', V\setminus U')$ is maximized (see Problem~\ref{prob:max-cut-local} for the formal definition).
For both \dskc and \maxcutkc, 
we give reductions showing that they are closely connected to their corresponding classic problems, 
densest subgraph and maximum cut, respectively.
In particular, we show that \dskc can be solved using solvers for the classic densest
$k$-subgraph problem~\citep{dblp:conf/coco/feigeseltser} and the \kdense
problem~\citep{matakos2022strengthening}.
Furthermore, \maxcutkc can be solved using a solver for the classic max-cut
problem~\citep{goemans1994approximation} with only constant factor loss in
approximation ratio. 
The approximation ratios using the black-box solvers are shown in Figure~\ref{fig:plots-approx-ratio} with label \emph{Black-box}. 
Moreover, we show that the reduction also works in the other way, and we present the hardness results for \dskc and \maxcutkc. 

Second, we apply our OptiRefine framework to the more general \emph{maximum
graph-partitioning problem} (\maxgp)~\citep{DBLP:journals/jal/feigel01,han2002improved,dblp:journals/rsa/halperinz02} (see Problem~\ref{prob:max-graph-partition} for formal definition). 
{Through parameter settings~\citep{han2002improved}, \maxgp encompasses various problems} including \emph{densest-$k$ subgraph} (\dks), \emph{max-cut with cardinality constraints} (\ccmaxcut), \emph{max-uncut with cardinality constraints} (\ccmaxuncut), and \emph{vertex cover with cardinality constraints} (\ccvc). 
\maxgp can be solved approximately using semidefinite program-based (\sdp-based) approaches.

We define the \maxgp under the OptiRefine framework as 
\emph{graph partition with $k$ refinements} (\gpkc) 
(see Problem~\ref{prob:graph-partition-local} for the formal definition).
{To solve \gpkc, we adapt the \sdp-based approaches for \maxgp by incorporating the initial node statuses (inside or outside the initial subgraph) and designing appropriate constraints in the relaxed \sdp program.
We also propose a systematic yet simple method to determine the search space for suitable parameters, as introduced in Lemma~\ref{lem:cut-bound-z}. 
Additionally, we introduce two pairs of parameters in the \maxcutkc problem to handle the case distinctions arising in the rounding step.
These schemes help us to obtain better approximation ratios. 
} 

For these problems, we establish constant-factor approximation ratios when $k=\Omega(n)$~refinements are allowed, where $n$~is the number of vertices in the graph. 
Our {approximation ratios} closely match the \sdp-based approximation ratios for \dks and \ccmaxcut, which our problems generalize. 
{Our approach} extends the results of~\citet{matakos2022strengthening} for \kdense {(formally defined in Problem~\ref{prob:k-densify})} by providing approximation algorithms for broader graph classes (see Section~\ref{sec:related} for details).
The approximation ratios achieved by our \sdp-based approaches for \dskc and \maxcutkc are
shown in Figure~\ref{fig:plots-approx-ratio} with label~\sdp.

Third, for \maxcutkc, we provide approximation algorithms that almost match the best-known algorithms for max-cut with cardinality constraints. 
In particular, we use the sum-of-squares hierarchy (SoS)~\citep{DBLP:journals/siamjo/Lasserre02} to show that for $k=\Omega(n)$, \maxcutkc admits the same approximation ratio as the state-of-the-art algorithm for cardinality-constrained max-cut~\citep{DBLP:conf/soda/RaghavendraT12}. 
This approximation ratio is optimal for almost all parameter choices $k=\Omega(n)$, assuming the Unique Games Conjecture~\citep{DBLP:conf/stoc/Khot02a}.
We present the approximation ratio for \maxcutkc using SoS in Figure~\ref{fig:plots-approx-ratio}(b).

From a practical perspective, we implement our theoretical algorithms with provable guarantees, including \sdp-based algorithms, black-box solver-based algorithms.
For comparison, we also implement the greedy heuristics.
We evaluate these algorithms and heuristics on both synthetic and real-world datasets. 
For \dskc, our algorithms with black-box solvers are scalable and increase the densities of the given subgraphs.
For \maxcutkc, our \sdp-based approaches significantly outperform the methods by~\citet{matakos2020tell}.

\vspace{1mm}
\para{Structure of the paper.} 
This paper is structured as follows.
In Section~\ref{sec:densest-subgraph}, we define \dskc, and we illustrate its connections to \kdense and \dks. 
Moreover, we present a black-box solution for \dskc by applying a solver for \dks. 
Section~\ref{sec:max-cut} presents \maxcutkc, and we illustrate its connections to \ccmaxcut. 
Moreover, we consider a black-box solution for \maxcutkc by applying a solver for \maxcut. 
In Section~\ref{sec:general-framework}, we use a general problem, graph
partitioning with $k$ refinements (\gpkc), that captures both \dskc and
\maxcutkc. We then propose an \sdp-based algorithm and
obtain constant approximation results assuming $k \in \Omega(n)$. 
In Section~\ref{sec:max-cut:sos}, we describe a sum-of-squares algorithm to optimally solve \maxcutkc under certain regimes of $k \in \Omega(n)$.
Section~\ref{sec:experiments} presents our experiments, in which we evaluate our
algorithms for \dskc using multiple datasets. Moreover, in the appendix, we present more
experimental evaluations for \dskc and \maxcutkc, as well as all
omitted proofs from the main text.

\subsection{Notation}
\label{preliminary}
Throughout the paper, we let $G=(V,E,w)$ be an un\-directed graph 
with non-negative edge weights. We set $n=|V|$ and $m=|E|$.
For a subset of vertices $U\subseteq V$, 
we write $E[U]$ to denote the set of edges with both endpoints in $U$, 
and $G[U]$ to denote the subgraph $(U,E[U])$.
The \emph{density} of a subgraph $G[\OldSet]$ is defined as $d(\OldSet) = \frac{{w(E[\OldSet])}}{|\OldSet|}$, 
where ${w(E[\OldSet])} = \sum_{(i,j) \in E[\OldSet]} w(i,j)$. 
We write $\OldComSet$ to denote the complement of $U$, 
i.e., $\OldComSet = V\setminus U$. 
Given a partition 
$(U,\OldComSet)$ of $V$, we write $E[U,\OldComSet]$ to denote the set of edges that
have one endpoint in $U$ and one in $\OldComSet$.
We let $\cut{\OldSet} = \sum_{(i,j)\in E[\OldSet,\OldComSet]} w(i,j)$ denote {the sum of edge weights over the cut} $(\OldSet,\OldComSet)$.
For two sets $A$ and $B$ we denote their symmetric difference by 
$A\symm B = (A \setminus B) \cup (B \setminus A)$. 
The operator $\log$ stands for $\log_2$.

\section{Related work}
\label{sec:related}
The densest-subgraph problem has received considerable attention in the literature. 
Here, we discuss only the results that are most relevant to our work. 
We refer to the recent surveys by \citet{lanciano2023survey} and \citet{luo2023survey} for more discussion.

The problem of the densest subgraph with $k$~refinements (\dskc), as introduced in this paper, is a generalization of \emph{densest $k$-subgraph} (\dks)~\citep{DBLP:conf/focs/KortsarzP93}.
Despite significant attention in the theory community, there is still a large gap between the best approximation algorithms and hardness results for 
\dks~\citep{%
DBLP:journals/dam/AsahiroHI02,%
DBLP:conf/stoc/Barman15,%
DBLP:conf/stoc/BhaskaraCCFV10,%
DBLP:conf/stoc/Feige02,%
DBLP:journals/algorithmica/FeigePK01,%
DBLP:conf/stoc/LeeG22,%
DBLP:journals/algorithms/Manurangsi18}.
{
\citet{DBLP:journals/dam/AsahiroHI02} propose a greedy-peeling algorithm with an approximation ratio of $k/n$.
\citet{DBLP:journals/jal/feigel01} introduce an approximation algorithm based on
an \sdp relaxation and hyperplane rounding, achieving a better approximation ratio than $\frac{k}{n}$.
\citet{han2002improved} build on this method by incorporating partition information into the hyperplane rounding, resulting in an improved approximation ratio.
Notably, \citet{dblp:conf/coco/feigeseltser} demonstrate that the integrality
gap for standard \sdp relaxation is $\bigO(n^{1/3} \log n)$ when $k =
\Omega(n^{1/3})$, excluding the possibility of utilizing the standard \sdp-based approach to achieve a better approximation ratio for all settings of $k$.
The best current polynomial-time algorithm for \dks provides an approximation
guarantee of $\Omega(n^{-1/4-\varepsilon})$ for any constant $\varepsilon
>0$~\citep{DBLP:conf/stoc/BhaskaraCCFV10}, utilizing a more combinatorial
approach compared to the \sdp-based methods.
Other strategies include those that leverage the structure of the adjacency matrix, as seen in \citep{papailiopoulos2014finding}, and scalable methods for solving the relaxation of \dks, as seen in \citep{lu2024densest}.
}

The unconstrained version of the densest-subgraph problem has been used in practice to find communities in graphs. 
\citet{goldberg1984finding} has shown that this problem can be solved in polynomial
time and~\citet{charikar00approximation} provides a $2$-approximation algorithm.  
Besides the case of static graphs, efficient methods have been designed for  dynamic~\citep{bhattacharya2015space,sawlani2020near} and streaming settings~\citep{mcgregoar2015densest}.

Many variants of the densest-subgraph problem have been studied. For instance, \citet{sozio2010community} consider finding a densely connected subgraph that contains a set of query nodes.
\citet{dai2022anchored, ye2024efficient} consider finding a dense subgraph whose vertices are close to a set of reference nodes and contain a small set of anchored nodes.
\citet{tsourakakis2013denser} observe that, in practice, the densest subgraph is typically large and 
aim to find smaller, denser subgraphs by modifying the objective~function. 
Another line of work aims at finding subgraphs that are dense for multiple graph snap\-shots \citep{charikar2018finding,jethava2015finding,semertzidis2019finding}.

\citet{matakos2022strengthening} introduce the \kdense problem, where the goal is to add
$k$~vertices to a given subgraph to increase its density. 
They obtain $O(\sigma)$-approximation algorithms for graphs that have a $\sigma$-quasi-elimination~ordering. 
This class of graphs includes, for instance, chordal graphs. However, general graphs may not have a small $\sigma$-quasi-elimination~ordering.\footnote{Their algorithm requires a predefined $\sigma$.}
In contrast, our algorithms from Section~\ref{sec:densest-subgraph} provide constant-factor approximations for general graphs, but we need the assumption~$k=\Omega(n)$. 
As our hardness results show, this assumption is necessary to obtain $O(1)$-approximation algorithms, assuming the 
Small Set Expansion Conjecture~\citep{dblp:conf/stoc/raghavendras10}.

The problem of max-cut with $k$~refinements (\maxcutkc) is a generalization of \emph{max-cut with cardinality constraints} ({\ccmaxcut})~\citep{DBLP:journals/algorithmica/FriezeJ97,DBLP:journals/jcss/PapadimitriouY91}, in which the size of the cut is restricted to be equal to $k$, for a given $k \in [n]$.
The most studied version of {\ccmaxcut} is known as \maxbis, in which $k=n/2$
\citep{%
DBLP:journals/talg/AustrinBG16,%
DBLP:journals/jal/feigel01,%
DBLP:journals/algorithmica/FriezeJ97,%
DBLP:conf/soda/Manurangsi19,%
DBLP:conf/soda/RaghavendraT12%
}. 
The best polynomial-time approximation algorithm for {\ccmaxcut} achieves an approximation ratio of approximately $0.858$~\citep{DBLP:conf/soda/RaghavendraT12}.
This ratio is improved to approximately $0.8776$ for \maxbis 
\citep{DBLP:journals/talg/AustrinBG16}.
With respect to hardness of approximation for {\ccmaxcut}, the best result is due to~\citet{DBLP:conf/approx/AustrinS19} who give hardness of approximation with $k=\tau n$  as a function of~$\tau$, assuming the Unique Games Conjecture~\citep{DBLP:conf/stoc/Khot02a}.
A matching approximation algorithm for $\tau \in (0,0.365) \cup (0.635,1)$ is given by \citet{DBLP:conf/soda/RaghavendraT12}.
Our approximation algorithm for \maxcutkc obtains the same approximation ratio. As we show that \maxcutkc generalizes \ccmaxcut, in the parameter setting above, our algorithms are also optimal under the Unique Games Conjecture~\citep{DBLP:conf/stoc/Khot02a}.

The problem of graph partitioning with $k$ refinements (\gpkc) is a
generalization of \emph{maximum graph partitioning}
(\maxgp)~\citep{DBLP:journals/jal/feigel01, han2002improved,
	dblp:journals/rsa/halperinz02}, which generalizes several graph partitioning
	problems.
Besides the two problems \dks, \ccmaxcut that we generalize in this paper, two
other problems, namely, \emph{max-uncut with cardinality constraints} (\ccmaxuncut),
	  and \emph{vertex cover with cardinality constraints} (\ccvc) are also presented
	  as applications of \maxgp. 
All of the problems can be solved by \sdp-based approaches, and constant
approximation results are obtained assuming $k \in \Omega(n)$, where the
concrete approximation ratios depend on the value of~$k$. 

{As for picking the correct value of~$k$ in practice, we are not aware of
previous work covering this question for DkS or for Max-Cut-CC, and thus also
not for our problems. We suggest the following two strategies for DSkR: (1)~Our
greedy algorithms can be adapted such that they keep on including/removing nodes
as long as the density of the subgraph is increased. Once the density would
decrease by a node addition/deletion, they stop. This makes the algorithm
completely parameter-free, since $k$ is determined solely by the number of nodes
whose addition/removal increases the subgraph’s density. Beyond our greedy
algorithms, this approach can be mimicked by running our SDP-based or black-box
methods for multiple values of k and picking the one with the highest density.
(2)~Another approach could be to only add vertices to the subgraph as long as at
least a $\theta$-fraction of their degree is to vertices that were contained in
the original subgraph, for some threshold $\theta$. This ensures that only
vertices are added which are closely connected to vertices from the original
subgraph. Similar approaches can also be used for Max-Cut-kR (by replacing the
density objective with the cut objective).}

Finally, we remark that independently and concurrently with our project, \citet{fellows2023solution} and \citet{ grobler2023solution} propose reconfiguration frameworks that have similarities with our OptiRefine framework.
In addition, \citet{dalirrooyfard2024graph} study the $r$-move-$k$-partition problem, where the goal is to \emph{minimize} the multi\-way cut among $k$ partitions by moving $r$ nodes' positions. 
We note that this line of research focuses on developing fixed-parameter
algorithms for different sets of problems. Hence, their results apply in the
setting of small values of~$k$, whereas we concentrate on large values of~$k$.
Therefore, the concrete results of these papers are incomparable with ours, but
highlight the importance of studying refinement problems.

\section{Densest subgraph with $k$~refinements}
\label{sec:densest-subgraph}

In this section, we formally define the problem of finding the densest subgraph with
$k$~refinements (\dskc) and other relevant problems. 
We then study the relationship between these problems.

\begin{problem}[Densest subgraph with $k$~refinements (\dskc)]
\label{prob:densest-subgraph-local}
Given an undirected graph $G = (V, E, w)$, a subset $\OldSet \subseteq V$, and an integer $k \in \mathbb{N}$, \dskc seeks to find a subset of vertices $\SelectSet \subseteq V$ 
with $|\SelectSet| = k$ such that the density 
$d(\OldSet \symm \SelectSet)$ is maximized.
\end{problem}

In the definition of Problem~\ref{prob:densest-subgraph-local}, 
the operation $\symm$ denotes the symmetric difference of sets and 
$d(\OldSet\symm \SelectSet)$ is the density of the subgraph 
$G[\OldSet\symm \SelectSet]$, which is obtained by 
removing the vertices in $\OldSet \cap \SelectSet$ and adding the 
vertices in $\SelectSet \setminus \OldSet$. 
It is important to note that when $G[\OldSet]$ is already the 
densest subgraph, $G[\OldSet\symm \SelectSet]$ may have a lower 
density than $G[\OldSet]$ due to the constraint $|\SelectSet| = k$.

{
We note that it may also be interesting to study the problem with the inequality constraint $\abs{\SelectSet} \leq k$. 
However, from a theoretical point of view, the results of~\citet{khuller2009finding} and the proof of Lemma~\ref{lem:densest-subgraph-generalization} imply that (up to constant factors) the inequality-constraint version of the problem is at least as hard to approximate as the equality-constraint version. 
Therefore, we study the equality-constraint version, because it simplifies the analysis.
From a more practical point of view, we can solve the inequality-constraint version by running the equality-constraint version for multiple values of~$k$.
}

Below, in Theorem~\ref{thm:densest-subgraph-sdp} we show that we can obtain a
constant factor approximation algorithm for \dskc. In this section, we focus on
its computational complexity and its relationship with other classic problems.
First, we introduce the \kdense problem, %
which serves as an intermediate step in our analysis.
\begin{problem}[\kdense~\citep{matakos2022strengthening}]
\label{prob:k-densify}
	Given an undirected graph $G=(V, E, w)$, a set $U \subseteq V$ and an integer $k \in \mathbb{N}$, \kdense seeks to find a set of vertices $C \subseteq V$ with $\abs{C} = k$ such that ${w(E[\OldSet \symm \SelectSet])}$ is maximized.~\footnote{
We note that in the definition given by \citet{matakos2022strengthening}, the objective is
${w(E[\OldSet \cup \SelectSet])}$. However, they implicitly require that $k\leq n - \abs{U}$. 
{
Observe that in this case, the $k$ nodes are always selected from $V \setminus U$ in order to maximize the objective function.}
Hence, the maximum of ${w(E[\OldSet \cup \SelectSet])}$ is equal to ${w(E[\OldSet \symm \SelectSet])}$.
In Problem~\ref{prob:k-densify}, we adopt the objective function ${w(E[\OldSet \symm \SelectSet])}$ to make it consistent with \dskc.
}
\end{problem}

Note that the \kdense problem has the same cardinality constraint as \dskc, i.e., $|C|=k$. 
However, the objective function of \kdense maximizes ${w(E[\OldSet \symm \SelectSet])}$, 
{i.e., it maximizes the weight of the edges in the subgraph.}
That is unlike \dskc which maximizes the \emph{density} $d(\OldSet\symm \SelectSet)$.
Thus, while in \dskc we can add \emph{and remove} vertices from the given subgraph, in
\kdense we can only add vertices to the subgraph.
{
This is because removing vertices from the subgraph always causes a decrease of the sum of the edge weights of the subgraph. However, in some cases removing the vertices from the subgraph increases the density of the subgraph.
In Figure~\ref{fig:kdense-dskc-example}, we give a concrete example showing that
\dskc may find subgraphs which are more community-like and have higher density.  
}

\begin{figure}
    \centering
    \begin{tabular}{ccc}
        \includegraphics[width=0.3\textwidth]{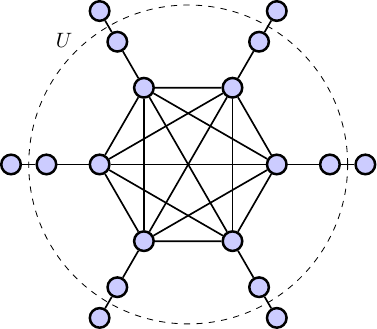} &
        \includegraphics[width=0.3\textwidth]{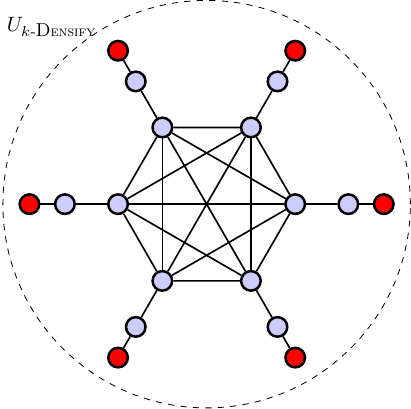} &
        \includegraphics[width=0.3\textwidth]{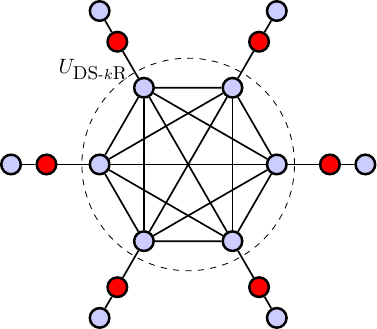} \\
        \begin{minipage}[t]{0.3\textwidth}
        \caption*{(a) {The input subgraph $G[\OldSet]$. The outer nodes
			of $\OldSet$ are connected to exactly one node in $V \setminus \OldSet$.}}
        \end{minipage} &
        \begin{minipage}[t]{0.3\textwidth}
        \caption*{(b) {The optimal solution $\OldSet_{\kdense}$ for \kdense. The density of the subgraph is $1.5$.}}
        \end{minipage} &
        \begin{minipage}[t]{0.3\textwidth}
        \caption*{(c) {The optimal solution $\OldSet_{\dskc}$ for \dskc. The density of the subgraph is $2.5$.}}
        \end{minipage}
    \end{tabular}
    \caption{{An example illustrating the difference between the optimal
		solutions for \kdense and \dskc. Here, we set $k=6$. 
		We illustrate the subgraphs of the initial solution and of the optimal
		solutions using dotted circles.
		The nodes in red are those selected by the optimal solutions of the \kdense and \dskc problems, respectively.
		Clearly, the solution found by \dskc is more community-like than the
		solution returned by \kdense, and also yields a higher density.}
	}
    \label{fig:kdense-dskc-example}
\end{figure}

Finally, we introduce the densest $k$-subgraph (\dks).
\begin{problem}[Densest $k$-subgraph (\dks)]
\label{prob:densest-k-subgraph}
Given an undirected graph $G=(V,E,w)$ and an integer
$k\in\mathbb{N}$, \dks seeks to find a set of vertices $\OldSet\subseteq V$
with $|\OldSet|=k$ such that the density $d(\OldSet)$ is maximized.
\end{problem}

\subsection{Relationships of the problems}
\label{sec:relationships-densest}

We show the connection of the three problems above.
In the subsequent analysis, the approximation ratios that we derive will
typically depend on $k$.  To emphasize the dependence on $k$, we will denote the
approximation ratios using the notation $\alpha_k$.

\spara{Relationship of \dskc and \dks.}
First, we show that \dskc \emph{generalizes}
the classic \emph{densest $k$-subgraph} (\dks) problem.  

\begin{lemma}
\label{lem:densest-subgraph-generalization}
	Let $k\in\mathbb{N}$.
	If there exists an $\alpha_k$-approximation algorithm for 
	\dskc with running time $T(n,m,k)$, then there exists an
	$\alpha_k$-approximation algorithm for \dks running in
	time $\bigO(T(n,m,k))$.
\end{lemma} 

Lemma~\ref{lem:densest-subgraph-generalization} allows us to obtain
hardness-of-approximation results for \dskc based on the hardness of
\dks~\citep{DBLP:conf/stoc/LeeG22}.

\begin{corollary}[Hardness of approximation for \dskc]
\label{cor:densest-subgraph-hardness}
	Any hardness of factor $\beta_k$ for approximating \dks implies
	hardness of approximating \dskc within a factor $\beta_k$. 
	In particular, if $k=\tau n$ for $\tau\in (0,1)$, 
	then \dskc is hard to approximate within 
	$\SelectSet \tau \log(1/\tau)$, 
	where $\SelectSet \in \mathbb{R}_+$ is some fixed constant,
	assuming the Small Set Expansion Conjecture~\citep{dblp:conf/stoc/raghavendras10}. 
\end{corollary}

Corollary~\ref{cor:densest-subgraph-hardness} implies that for $k=o(n)$ 
(and, therefore, $\tau=o(1)$), it is not possible to obtain constant-factor approximation
algorithms for \dskc. Additionally, for $k=\Omega(n)$, the problem is hard to
approximate within a constant~factor.

\spara{Relationship of \kdense and \dskc.}
We show that \dskc can be effectively solved by applying an algorithm for
\kdense on the same input instance, with little impact on the approximation
ratio.
Specifically, Lemma~\ref{lemma:k-densify-local-changes-connection} shows that running an
$\alpha_k$-approximation algorithm for \kdense on an instance of \dskc only
loses a factor of factor~$\frac{1-c}{1+c}$ in approximation, where $c \in (0,1)$
is a constant such that~$k \leq c\abs{U}$. We note that in practice, it is
realistic that $c$ is small since we only want to make a small fraction of
changes to~$U$.  This implies that $\frac{1-c}{1+c}$ is large; for
instance, $\frac{1-c}{1+c} \geq 0.9$ if $c=5\%$.
Notice that the condition $k \leq n - \abs{U}$ is an implicit requirement from the definition of the \kdense problem. 

\begin{lemma}
    \label{lemma:k-densify-local-changes-connection}
	Let $(G=(V, E, w), U, k)$ be an instance for \dskc.
    Let us assume that $k \leq c \abs{U}$ for some $c \in (0, 1)$, and $k \leq n - \abs{U}$. 
	Applying an $\alpha_k$-approximation algorithm for \kdense on $G$ 
    provides a $\left(\frac{1-c}{1+c}\, \alpha_k\right)$-approximate solution for \dskc. 
\end{lemma}

\spara{Relationship of \kdense and \dks.}
We show that an algorithm for the Densest $(k+1)$-Subgraph problem (\dksplus)
can be effectively utilized to solve \kdense, with only small losses in the approximation ratio.

\begin{lemma}
\label{lemma:k-densify-dks}
	Suppose there exists an $\alpha_{k+1}$-approximation algorithm for \dksplus,
	then there exists a $\frac{k-1}{k+1}\alpha_{k+1}$-approximation algorithm
	for \kdense.
\end{lemma}

\begin{proof}[Proof sketch]
On a high level, the reduction from \kdense to \dksplus works as follows (see
Appendix~\ref{appendix:ds-proofs:kdensifydks} for details). Given an instance
$(G=(V,E,w),\OldSet,k)$ for \kdense, we build a graph~$G'$ in which we
\emph{contract} all vertices in $\OldSet$ into a special vertex~$i^*$. Then we
solve \dksplus on the graph~$G'$.  We then consider two cases:
either~$i^*$ is included in the (approximate) solution~$U'$ of \dksplus, or not.  If
$i^*$ is part of the solution, we have added exactly $k$~vertices
to~$\OldSet$. Otherwise, we add all vertices from $U'$ to $U$, except the vertex
from $U'$ of lowest degree in $G[U\cup U']$; this implies that we also added
$k$~vertices to~$U$. Then we can show
that the $k$~additional vertices chosen by the \dksplus-algorithm form a
solution for \kdense subgraph with the desired approximation guarantee.
\end{proof}

\subsection{Black-box reduction}
\label{sec:densest-subgraph:black-box}

By combining the results from the previous section, we can now show that any approximation
algorithm for \dks implies (under some assumptions on $|\OldSet|$ and~$k$) an
approximation algorithm for \dskc.  This result can be viewed as the reverse
direction of the reduction in Lemma~\ref{lem:densest-subgraph-generalization},
and is obtained by applying
Lemmas~\ref{lemma:k-densify-local-changes-connection}
and~\ref{lemma:k-densify-dks}.  Our formal result is as follows.
\begin{theorem}
\label{thm:densest-subgraph-reduction}
Suppose there exists an $\alpha_k$-approximation algorithm for the \dks problem with running time $T(m,n,k)$.
If $k \leq c\abs{\OldSet}$, for some $c \in (0, 1)$, 
there exists a $\left(\frac{1-c}{(1+c)}\, \frac{k-1}{k+1} \, \alpha_{k+1}\right)$-approximation algorithm for 
the \dskc problem with running time $\bigO(T(m,n,k+1) + m)$.
\end{theorem}

By utilizing the approximation algorithm for \dks by \citet{DBLP:journals/jal/feigel01} in
Theorem~\ref{thm:densest-subgraph-reduction}, we obtain
Corollary~\ref{cor:densest-subgraph-reduction}.  We
include the assumption $k = \Omega(n)$ in the corollary, as it is a requirement
of~\citet{DBLP:journals/jal/feigel01}.

\begin{corollary}
\label{cor:densest-subgraph-reduction}
	If $\abs{\OldSet}=\Omega(n)$, $k=\Omega(n)$, and 
	$k \leq c\abs{\OldSet}$ for a constant~$c \in (0, 1)$,
    there exists an $\bigO(1)$-approximation algorithm
	for~\dskc.
\end{corollary}

In our experiments, we will implement the algorithm from
Corollary~\ref{cor:densest-subgraph-reduction} and denote it
as \emph{black-box algorithm}, since the algorithm is using a black-box solver
for \dks.

\section{Max-cut with $k$ refinements}
\label{sec:max-cut}

In this section, we study max-cut with $k$ refinements (\maxcutkc). 
We start by providing formal definitions for \maxcutkc and for cardinality-constrained max-cut (\ccmaxcut). 
We then present a hardness result for \maxcutkc through \ccmaxcut.  
In the reverse direction, we show that any approximation algorithm for unconstrained \maxcut 
provides an approximation algorithm for \ccmaxcut, in a black-box reduction. 
Later, in Sections~\ref{sec:general-framework} and~\ref{sec:max-cut:sos}, we will
also provide constant-factor approximation algorithms for \maxcutkc.

{We start by formally defining \maxcutkc and give an example in Figure~\ref{fig:maxcutkr-example}.}

\begin{problem}[Max-cut with $k$ refinements (\maxcutkc)]
\label{prob:max-cut-local}
	Given an undirected graph $G = (V, E, w)$, a subset $\OldSet \subseteq V$, and an integer $k \in \mathbb{N}$, 
	\maxcutkc seeks to find a set of
	vertices $\SelectSet\subseteq V$ such that $|\SelectSet|=k$ and 
	$\cut{\OldSet \symm \SelectSet}$ is maximized.
\end{problem}

As with the \dskc problem, observe that if $\cut{\OldSet}$ is the maximum cut,
$\cut{\OldSet\symm \SelectSet}$ may have a lower value than $\cut{\OldSet}$ 
due to the equality constraint $|\SelectSet| = k$.  
{
\begin{figure}
    \centering
    \begin{subfigure}[c]{0.3\textwidth}
        \centering
        \includegraphics[width=\textwidth]{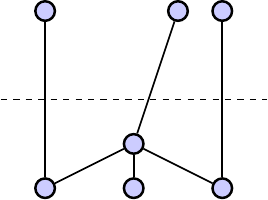}
        \caption{{The initial partition of the graph. The cut value is $3$.}}
        \label{fig:cut-example-a}
    \end{subfigure}
    \hfill
    \begin{subfigure}[c]{0.3\textwidth}
        \centering
        \includegraphics[width=\textwidth]{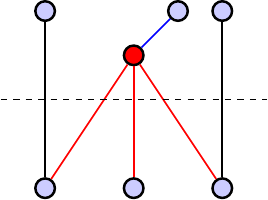}
        \caption{{The optimal solution when $k=1$. The cut value is $5$.}}
        \label{fig:cut-example-b}
    \end{subfigure}
    \hfill
    \begin{subfigure}[c]{0.3\textwidth}
        \centering
        \includegraphics[width=\textwidth]{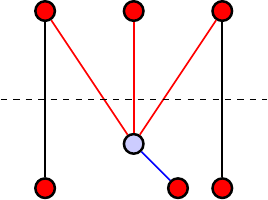}
        \caption{{The optimal solution when $k = 6$. The cut value is $5$.}}
        \label{fig:cut-example-c}
    \end{subfigure}
    \caption{{An illustration of \maxcutkc on a graph with $7$ nodes. The dotted line indicates the partition of the graph. The nodes with red color are the nodes selected by the optimal solution of the \maxcutkc problem. The red and blue edges indicate that the edge is moved into or out of the cut set, respectively.}}
    \label{fig:maxcutkr-example}
\end{figure}
}

We also introduce the problem of \emph{maximum cut with cardinality constraints} ({\ccmaxcut}), 
which we use for the analysis of the hardness of \maxcutkc.
\begin{problem}[Max-cut with cardinality constraints (\ccmaxcut)]
Given an undirected graph $G = (V, E, w)$, and an integer $k \in \mathbb{N}$, 
	\ccmaxcut seeks to find a set of vertices $\OldSet\subseteq V$ of size $|\OldSet|=k$ 
	such that $\cut{\OldSet}$ is maximized. 
\end{problem}

\spara{Relationship of \maxcutkc and \ccmaxcut.}
First, we show that \maxcutkc is a generalization of \ccmaxcut in Lemma~\ref{lem:max-cut-generalization}.

\begin{lemma}
\label{lem:max-cut-generalization}
	Let $k\in\mathbb{N}$.
	If there exists an $\alpha_k$-approximation algorithm for
	\maxcutkc with running time $T(n,m,k)$, then there exists an
	$\alpha_k$-approximation algorithm for 
	\ccmaxcut with running time~$\bigO(T(n,m,k))$.
\end{lemma}

Combining Lemma~\ref{lem:max-cut-generalization} and the results of~\citet{DBLP:conf/approx/AustrinS19}, we obtain the following hardness result for \maxcutkc.

\begin{corollary}
\label{cor:max-cut-hardness}
	Any hardness of factor $\beta_k$ for approximating {\ccmaxcut} 
    implies hardness of approximating {\maxcutkc} within factor $\beta_k$. 
	In particular, suppose that $k=\tau n$ for $\tau\in (0,1)$. 
	For $\tau\leq 0.365$ or $\tau\geq 0.635$, we have that
	$\beta_k\approx0.8582$ and for $\tau\in(0.365,0.635)$ we have
   	$\beta_k\in[0.8582, 0.8786]$.
\end{corollary}

\subsection{Black-box reduction}
\label{sec:max-cut:black-box}

Finally, we show that algorithms for the
(unconstrained) {\maxcut} problem can be used to solve {\maxcutkc}.
{
\begin{theorem}
\label{thm:max-cut-black-box}
	Consider an $\alpha$-approximation algorithm for \maxcut with running time
	$T(m,n)$. 
	There exists an algorithm for \maxcutkc with running time $\bigO(T(m,n) + m)$.
	The approximation ratio when $k \leq n/2$ is
    $\alpha \left(1 -\Theta\!\left(\frac{1}{n}\right)\right)
    \min\left\{4\frac{k^2}{n^2}, \left(1-\frac{k}{n}\right)^2\right\}$.
	For $k \geq n/2$, the approximation ratio is 
    $\alpha \left(1 -\Theta\!\left(\frac{1}{n}\right)\right)
    \min\left\{4\frac{(n-k)^2}{n^2}, \left(\frac{k}{n}\right)^2\right\}$.
\end{theorem}
}

In our reduction, we run an algorithm 
(e.g., the {\sdp} relaxation by~\citet{goemans1994approximation}) for {\maxcut} without
cardinality constraint, i.e., the cut may contain an arbitrary number of
vertices on each side. Suppose this algorithm returns a cut $(\TempSet,\TempComSet)$.
Then we make refinements to $\TempSet$ to bring it ``closer'' to $\OldSet$,
i.e., we greedily move vertices until $\TempSet=\OldSet\symm \SelectSet$ 
for some set $\SelectSet$ with $|\SelectSet|=k$.  
Using the assumptions from the theorem, we obtain an $\bigO(1)$-approximation
algorithm assuming that $k=\Omega(n)$.
The pseudo\-code for this method is presented in the appendix.

\section{Graph partitioning with $k$ refinements}
\label{sec:general-framework}

The \emph{maximum graph-partitioning problem} (\maxgp) characterizes \dks and \ccmaxcut, as well as \emph{max-uncut with cardinality constraints} (\ccmaxuncut), and \emph{vertex cover with cardinality constraints} (\ccvc) with specific parameter settings. 
Previous papers~\citep{DBLP:journals/jal/feigel01,han2002improved,dblp:journals/rsa/halperinz02}
provided \sdp-based approximation algorithms for \maxgp.
Let us define \maxgp the same way as \citep{han2002improved}, by introducing parameters $\parzero, \parone, \partwo, \parthree$ into the \maxgp objective function, specified in Table~\ref{tab:case_distinctions}.
Our analysis in the remainder of this section holds for these specific settings. 

\begin{table}[t]
\centering
\caption{
Partition measures and corresponding values of 
$\parzero, \parone, \partwo, \parthree$ for the problems we study in this paper.
}
\begin{tabular}{lllrrrr}
\toprule
{Problem} \gpkc & {Problem} \maxgp & {Measure} $\partition$ & \parzero & \parone & \partwo & \parthree\\  \midrule \smallskip
\kdense & \dks & ${w(E[\cdot])}$ & $\frac{1}{4}$ & $\frac{1}{4}$ & $\frac{1}{4}$ & $\frac{1}{4}$ \\ \smallskip %
\maxcutkc & \ccmaxcut &$\cut(\cdot)$ & $\frac{1}{2}$ & $0$ & $0$ & $-\frac{1}{2}$\\  \smallskip %
\maxuncutkc & \ccmaxuncut &${w(E[V])} - \cut(\cdot)$ & $\frac{1}{2}$ & $0$ & $0$ & $\frac{1}{2}$\\  %
\vckc & \ccvc &${w(E[\cdot])} + \cut(\cdot)$ & $\frac{3}{4}$ & $\frac{1}{4}$ & $\frac{1}{4}$ & $-\frac{1}{4}$\\
\bottomrule
\end{tabular}
\label{tab:case_distinctions}
\end{table}

\begin{problem}[Maximum graph-partitioning (\maxgp)]
	\label{prob:max-graph-partition}
Given an undirected graph $G = (V, E, w)$ where $\abs{V}=n$ and an integer $k \in \mathbb{N}$, 
the predefined parameters $\parzero$, $\parone$, $\partwo$, $\parthree$ and the partition measure $\partition{\cdot}$,
\maxgp seeks to find a set of vertices $\OldSet \subseteq V$ where $|\OldSet| = k$, such that the partition measure $\partition{\OldSet}$ is maximized.

In particular, 
    let $\vecx\in\{-1,1\}^n$ be an indicator vector such that $x_i=1$ if and only if $i\in \OldSet$. 
    \maxgp is formulated as follows:
     \begin{equation}
    	\label{ip:max-graph-partition}
    \begin{aligned}
    \max_{\vecx} \quad 	& 
    \sum_{i<j} w(i,j) \, (\parzero + \parone x_i + \partwo x_j + \parthree x_i x_j),\\
    \quad \text{such that }	& \sum_i x_i = -n + 2k  \text{ and }\\ 
	& \vecx \in \{-1,1\}^n.
    \end{aligned}
    \end{equation}
The predefined parameters and partition measures are presented in Table~\ref{tab:case_distinctions}. 
\end{problem}

In this section, we generalize \maxgp to take into account refinements of
existing solutions. We call our generalized problem \emph{graph partitioning with $k$ refinements} (\gpkc) and define it formally below.
The \kdense and \maxcutkc from above are instances of \gpkc by specifying the partition measure $\partition{\cdot}$.

\begin{problem}[Graph partitioning with $k$ refinements (\gpkc)]
	\label{prob:graph-partition-local}
Given an undirected graph $G = (V, E, w)$ where $\abs{V} = n$, a subset $\OldSet \subseteq V$, and an integer $k \in \mathbb{N}$,
the predefined parameters $\parzero$, $\parone$, $\partwo$, $\parthree$ and the partition measure $\partition{\cdot}$,
\gpkc seeks to find a subset of vertices $\SelectSet \subseteq V$ with $\abs{\SelectSet} = k$, such that the partition measure $\partition{\OldSet \symm \SelectSet}$ is maximized.

In particular, let $\vecx^0 \in \{-1, 1\}^n$ be an indicator vector such that $x^0_i = 1$ if and only if $i \in \OldSet$ and 
    let $\vecx\in\{-1,1\}^n$ be an indicator vector such that $x_i=1$ if and only if $i\in \OldSet\symm \SelectSet$. 
    \gpkc is formulated as follows:
     \begin{equation}
    	\label{ip:densest-subgraph}
    \begin{aligned}
    \max_{\vecx} \quad 	& 
    \sum_{i<j} w(i,j) \, (\parzero + \parone x_i + \partwo x_j + \parthree x_i x_j),\\
    \quad \text{such that }	& 
    \sum_i x^0_i x_i = n-2k \text{ and } \\ & 
    \vecx \in \{-1,1\}^n.
    \end{aligned}
    \end{equation}
The predefined parameters and partition measures are presented in Table~\ref{tab:case_distinctions}. 
\end{problem}

Notice that when $\OldSet=\emptyset$ (and $\vecx^0=-\vecone$), the formulation of \gpkc becomes the same as  \maxgp. 
Taking into consideration the initial solution ($\OldSet \neq \emptyset$) changes the formulation of the constraint part. 
Hence, we mainly illustrate the constraint.  
We use the indicator vector $\vecx^0\in\{-1,1\}^{V}$ to encode the set $\OldSet$, i.e., we set $x^0_i = 1$ if $i\in \OldSet$ and $x^0_i = -1$ if $i\in V\setminus \OldSet$.  
We then set our solution~$\SelectSet$ to all $i$ such that
$x_i \neq x^0_i$, i.e., $\SelectSet=\{i\colon x_i\neq x^0_i\}$.
Observe that the constraint of the integer program implies that $\SelectSet$ contains exactly
$k$ vertices: we have $x^0_i x_i = 1$ if and only if $x^0_i = x_i$ and $x^0_i
x_i = -1$ if and only if $x^0_i \neq x_i$. 
Hence, the sum of all $x^0_i x_i$ is equal to $n-2\ell$, 
where $\ell=|\{ i \colon x_i \neq x^0_i\}|$ is the
number of entries in which $\vecx$ and $\vecx^0$ differ. Since in the constraint we
set $\sum_i x^0_ix_i=n-2k$, we get $|\SelectSet|=k$.

The objective functions of \gpkc are essentially the same as \maxgp, as well as the parameter settings.
Next, we illustrate the settings of \kdense and \maxcutkc. For \maxuncutkc and \vckc, we refer readers to \citet{han2002improved} for more details.
For \kdense, the term $(\frac{1}{4}+\frac{1}{4}x_i+\frac{1}{4}x_j+\frac{1}{4}x_ix_j)$ equals $1$ if $x_i=x_j=1$ and equals $0$ otherwise. 
Hence, the objective aims to maximize the
number of edges in $G[\OldSet\symm \SelectSet]$, i.e., we maximize ${w(E[\OldSet \symm \SelectSet])}$.
For \maxcutkc, the terms $(\frac{1}{2}-\frac{1}{2}x_ix_j) = 1$ if and only if $x_i \neq x_j$, 
i.e., when $i$ and $j$ are on different sides of the cut, otherwise $\frac{1}{2}-\frac{1}{2}x_ix_j = 0$. 
Thus, the objective function sums over all edges with
one endpoint in $\OldSet\symm \SelectSet$ 
and one endpoint in $\overline{\OldSet}\symm \SelectSet$, which is identical to $\cut{\OldSet \symm \SelectSet}$.

\subsection{SDP-based algorithm}
\label{sec:general:sdp}

Next, we present our \sdp-based algorithm. 
We will show that the \sdp-based algorithm, as applied by
\citet{DBLP:journals/algorithmica/FriezeJ97} to solve \ccmaxcut and by \citet{DBLP:journals/jal/feigel01} to solve \dks, can be adapted to solve our problem \gpkc.

\spara{The semidefinite program.}
We state the {\sdp} relaxation of the integer
program~\eqref{ip:densest-subgraph} as follows, where $\vecv_i \cdot \vecv_j$
denotes the inner product of $\vecv_i$ and $\vecv_j$:
\begin{equation} \label{sdp:densest-subgraph}
\begin{aligned}
	\max_{\vecv_0, \ldots, \vecv_n\in\mathbb{R}^n}~  & 
         \sum_{i < j} w(i,j) \, (\parzero + \parone \vecv_0 \cdot \vecv_i + \partwo \vecv_0 \cdot \vecv_j + \parthree \vecv_i \cdot \vecv_j),\\
    \text{such that }~ 	&
         \sum_i x_i^0 \vecv_i \cdot \vecv_0 = n - 2k, \\ & 
         \sum_{i,j} x_i^0 x_j^0 \vecv_i \cdot \vecv_j = (2k-n)^2, 
             					\text{ and } \\ & 
             					\vecv_i \in \mathcal{S}_{n}, \text{ for } i = 0, \ldots, n.
\end{aligned}
\end{equation}

In the \sdp, we replace each entry $x_i$ with the inner product $\vecv_i\cdot\vecv_0$ and each $x_ix_j$ with the inner product $\vecv_i \cdot \vecv_j$. 
All $\vecv_i$ are vectors from the unit sphere $\mathcal{S}_n$.
For technical reasons, we add the constraint $\smash{\sum_{i,j} x_i^0 x_j^0 \vecv_i \cdot \vecv_j} = (2k-n)^2$, which is a relaxation of the constraint $\smash{(\sum_{i,j} x_i^0 x_{i})^2} = (2k-n)^2$.
This constraint will help us to bound the variance of the random-rounding procedure.

\spara{The algorithm.}
Our algorithm for \gpkc consists of the following three steps:
\begin{enumerate}
\item We solve the {\sdp}~\eqref{sdp:densest-subgraph} to obtain a solution~$\vecv_0, \ldots, \vecv_n$. 
\item We obtain an indicator vector~$\overline{\vecx}\in\{-1,1\}^n$ (which we will later use to obtain~$\SelectSet$) as follows. 
We sample a unit vector $\vecr$ from $\mathcal{S}_n$ uniformly at random and obtain the indicator vector~$\overline{\vecx} \in\{-1,1\}^n$ by setting $\overline{x}_i=1$ if
$(\vecr \cdot \vecv_0)(\vecr \cdot \vecv_i) \geq 0$ and $\overline{x}_i=-1$ otherwise.
We call this step \emph{hyperplane rounding}.
\item As before, we set $\SelectSet=\{i \colon \overline{x}_i \neq x^0_i\}$ and define a temporary set $\TempSet = \{i \colon \overline{x}_i = 1\}$. 
Note that it is possible that $\abs{\SelectSet} \neq k$. 
In this case, we greedily modify $C$ to ensure that $\TempSet$ is a valid solution of \gpkc. We call this last step \emph{fixing $C$}.
\end{enumerate}

To boost the probability of success, 
the second and third steps of the algorithm are repeated
$\bigO(1/\varepsilon \lg(1/\varepsilon))$~times
and we return the solution with the largest objective function value.
The pseudocode of our concrete algorithms for \kdense and \maxcutkc are
presented in Appendices~\ref{appendix:densest-subgraph:sdp:pseudocode}
and~\ref{appendix:maxcutkc:sdp:pseudocode}, resp.

\spara{Our main results}. Before analyzing the \sdp-based algorithm, let us state our main results obtained through the \sdp-based algorithm for \dskc and \maxcutkc. 
For these problems, we provide a complete analysis of the constant factor approximation guarantees for $k=\Omega(n)$.
The specific value of~$k$ determines the constant factor approximation ratios, and we present them in Figure~\ref{fig:plots-approx-ratio}.
Note that, due to our hardness results from Corollary~\ref{cor:densest-subgraph-hardness}, this assumption is inevitable for \dskc if we wish to obtain $\bigO(1)$-approximation algorithms (assuming the Small Set Expansion Hypothesis~\citep{dblp:conf/stoc/raghavendras10}).
Moreover, the assumption $k=\Omega(n)$ is necessary for the analysis of the approximation factor, similar to previous work~\citep{DBLP:journals/jal/feigel01,DBLP:journals/algorithmica/FriezeJ97,han2002improved}.

\begin{theorem}
\label{thm:densest-subgraph-sdp}
Let $|\OldSet|=\Omega(n)$ and $k=\Omega(n)$ where $k \leq \min\{c|\OldSet|, n -
|\OldSet|\}$ for some constant $c\in(0,1)$. There exists an {\sdp}-based
randomized algorithm that runs in time $\bigOtilde(n^{3.5})$ and outputs a
$\bigO(1)$-approximate solution for \dskc with high probability.
\end{theorem}

\begin{theorem}
\label{thm:max-cut-sdp}
	If $|\OldSet|=\Omega(n)$ and $k=\Omega(n)$,
	there exists an {\sdp}-based randomized algorithm that runs in time $\bigOtilde(n^{3.5})$ and returns a
	$\bigO(1)$-approximate solution for \maxcutkc with high probability.
\end{theorem}

We note that the running time is dominated by solving the semidefinite program with $\Theta(n)$ variables and $\Theta(n)$ constraints, and the time complexity is $\bigOtilde(n^{3.5})$~\citep{jiang2020faster}.

\subsection{Heuristic}
\label{sec:densest-subgraph:heuristic}
We note that the time complexity of the \sdp-based algorithms is dominated by
solving the \sdp. By the property of our formulation, the number of constraints
is $\bigO(n)$.  The time complexity of solving the SDP using the
state-of-the-art method is $\bigO(n^{3.5} polylog(1/\epsilon))$~\citep{jiang2020faster}, where $\epsilon$ is an
accuracy parameter.

As this time complexity is undesirable in practice, we also provide an efficient
heuristic for \gpkc.  Our heuristic computes a set~$\SelectSet$, 
by starting with $\SelectSet=\emptyset$
and greedily \emph{adding} $k$~vertices to~$\SelectSet$. 
More concretely, while $|\SelectSet|<k$, we greedily add vertices to $\SelectSet$ 
by computing the value of the objective function $\partition_u = \partition{\OldSet\symm (\SelectSet\cup\{u\}}$,
for each vertex $u\in V\setminus \SelectSet$,
and adding to $\SelectSet$ the vertex~$u^*$ that increases the objective
function the most.
The heuristic terminates when $|\SelectSet|=k$.

\section{Sum-of-squares algorithm}
\label{sec:max-cut:sos}

In this section, we present an optimal approximation algorithm for \maxcutkc under some regimes of $k=\Omega(n)$, assuming the Unique Games Conjecture~\citep{DBLP:conf/stoc/Khot02a}. 
Our main result is stated in Theorem~\ref{theorem:sos-main-theorem}.
Due to Corollary \ref{cor:max-cut-hardness}, the approximation algorithm from
our theorem is optimal for 
$\tau \in (0,0.365) \cup (0.635,1)$, assuming the Unique Games Conjecture~\citep{DBLP:conf/stoc/Khot02a}.
\begin{theorem}
\label{theorem:sos-main-theorem}
Let $\tau\in (0,1)$ be a constant, and consider \maxcutkc with $k=\tau \cdot n$ refinements. Then for every $\varepsilon>0$, there is a randomized polynomial-time algorithm that runs in time $n^{O(1/poly(\varepsilon,\tau))}$, and with high probability, outputs a solution that approximates the optimal solution within a factor of $\alpha_{*} - \varepsilon$, where the value of $\alpha_{*}$ is approximately $0.858$. 
\end{theorem}

Our algorithm relies on the Lasserre hierarchy formulation~\citep{DBLP:journals/siamjo/Lasserre02} of the problem along with the algorithmic ideas introduced by \citet{DBLP:conf/soda/RaghavendraT12} for solving \ccmaxcut. 
We show that their method can be applied in our setting when an initial partition is given.
In addition, we show that their approximation ratio still applies in our setting.
Due to the complexity of the entire algorithm, here we only give an overview of
our main ideas; we present the detailed algorithm and its analysis in Appendix~\ref{sec:sos_detailed}.

Let $\overline{\vecx}$ be the partition obtained by
applying the randomized algorithm introduced by \citet{DBLP:conf/soda/RaghavendraT12}.
Specifically, it is obtained by the rounding scheme on an uncorrelated Lasserre SDP solution, which we illustrate in the appendix. 
Let $\overline{\Rnd}$ be the cut value of the partition induced by $\overline{\vecx}$, and $\OPT$ be the optimal solution of \maxcutkc. 
We state the main result that we obtain from applying the method of~\citet{DBLP:conf/soda/RaghavendraT12} in Lemma~\ref{lem:bound-variance}.
We note that the specific algorithm and its analysis are an extension of this
method, with minor (though not trivial) changes. 

\begin{restatable}{lemma}{boundvariance}
\label{lem:bound-variance}
Let $\tau\in (0,1)$ be a constant, and $k=\tau \cdot n$. 
Let $\delta >0$ be a small constant.
There is a randomized polynomial-time algorithm that runs in time $n^{O(1/poly(\delta,\tau))}$, and with
high probability returns a solution $\overline{\vecx}$  such that 
$\overline{\Rnd} \geq 0.858 \cdot \OPT (1-\delta)^2 $  and such that 
$\left|\sum_{i \in V} x_i^0 \overline{x}_i - (n-2k) \right| \leq \delta^{1/48} |V|$.
\end{restatable}

We note that compared to our analysis from Section~\ref{sec:general-framework},
the main differences are as follows: The \sdp-relaxation in
Section~\ref{sec:general-framework} also obtains a solution with objective
function value approximately $0.858 \cdot \OPT$ after the hyperplane rounding.
However, after the hyperplane rounding, we might have that $\abs{\SelectSet}$ is
very far away from the desired size of $\abs{\SelectSet}=k$. This makes the
greedy fixing step (Step~3 in the algorithm) expensive and causes a relatively
larger loss in the approximation ratio. This is where the strength of the Lasserre
hierarchy comes in: Lemma~\ref{lem:bound-variance} guarantees that
$\left|\sum_{i \in V} x_i^0 \overline{x}_i - (n-2k) \right| \leq \delta^{1/48} |V|$, 
which implies that $\abs{C}$ is very close to the desired size of $\abs{C}=k$.
Thus, we have very little loss in the approximation ratio when greedily ensuring
the size of~$C$ and obtain a better result overall.

Next, we use Lemma~\ref{lem:bound-variance} to directly prove our main theorem.

\begin{proof}[Proof of Theorem~\ref{theorem:sos-main-theorem}]
Observe that $\overline{\vecx}$ does not necessarily satisfy the constraint on the number of refinements.
We aim to obtain a partition $\vecx$ such that $\sum_{i \in V} x_i^0 x_i = n-2k$. 

Let $s = \sum_{i \in V} x_i^0 \overline{x}_i - (n - 2k)$, and $C$ be the set of vertices that changed their partitions. Observe that $s = (n - 2\abs{C}) - (n - 2k) = 2(k - \abs{C})$.
This implies that we need to move $\frac{s}{2}$
vertices\footnote{We allow negative values of $s$  to mean that we are moving vertices from $C$ to $\overline{C}$.}
from $\overline{C}$ to $C$, where $\overline{C}$ is the complement set of $C$,
i.e., $\overline{C} = V \setminus C$, and
\begin{equation*}
 |s| \leq \delta^{1/48} |V|.
\end{equation*}

Recall that $\overline{\Rnd}$ denotes the cut value induced by $\overline{\vecx}$. For $i \in V$ 
let us denote with $\zeta(i)$ the value that node $i$ contributes to the cut. We have that 
\begin{equation} \label{eq:s_eq}
	\sum_{i\in C} \zeta(i) \leq \sum_{i \in V} \zeta(i) \leq 2 \cdot \overline{\Rnd},
\end{equation}
and also
\begin{equation}\label{eq:sbar_eq}
	\sum_{i\in \overline{C}} \zeta(i) \leq \sum_{i \in V} \zeta(i) \leq 2 \cdot \overline{\Rnd}.
\end{equation}
In case we need to move $|s|$ vertices from $C$ to $\overline{C}$, we move vertices $i \in C$ with the smallest values of $\zeta(i)$. 
Let us use $\left \{x_i\right\}_{i \in V}$ to denote the final assignment obtained by moving $\frac{s}{2}$ vertices from $\overline{C}$  to
$C$, and let us use $\Rnd$ to denote the value of the cut under this assignment.
In this case due to \eqref{eq:s_eq} we have that 
\begin{equation*}
	\Rnd \geq \overline{\Rnd}\cdot (1-|s|/|C|),
\end{equation*}
Similarly, if we need to move $|s|$  vertices from $\overline{C}$ to $C$, we move vertices $i \in \overline{C}$ with the
smallest values of $\zeta(i)$. 
In this case due to \eqref{eq:sbar_eq} we have that 
\begin{equation*}
	\Rnd \geq \overline{\Rnd}\cdot (1-|s|/|\overline{C}|),
\end{equation*}
By the previous discussion, we have that 
\begin{equation*}
	\Rnd \geq \overline{\Rnd}\cdot \left(1-|s|/\min(|C|,|\overline{C}|)\right).
\end{equation*}
Now, since $k = \tau n$ for constant $\tau \in (0, 1)$, we can choose $\delta$
sufficiently small so that $\abs{V} - (2 \delta^{1/96} +  \delta^{1/48}) \abs{V}
\geq k \geq (2 \delta^{1/96} +  \delta^{1/48}) \abs{V}$.
For such chosen $\delta$ we have 
\begin{equation*}
	\Rnd \geq \overline{\Rnd}\cdot \left(1-\frac{\delta^{1/48}n  }{ 2\delta^{1/96} n + \delta^{1/48}n - \delta^{1/48} n}\right) \geq 
	\overline{\Rnd}\cdot (1-\delta^{1/96}).
\end{equation*}
Since $\overline{\Rnd} \geq 0.858 \cdot \OPT (1-\delta)^2$, we have that
\begin{equation*}
	\Rnd \geq (1-\delta^{1/96}) \cdot (1-\delta)^2  \cdot 0.858 \cdot \OPT.
\end{equation*}
This shows that we have an $(1-3\delta^{1/96})\cdot 0.858$-approximation
algorithm. By setting $\delta=(\varepsilon/3)^{96}$, we indeed obtain our
desired $(1-\varepsilon) 0.858$-approximation algorithm.
\end{proof}

\section{Experimental evaluation}
\label{sec:experiments}
We evaluate our approximation algorithms and heuristics on a collection of real-world and synthetic datasets. 
Our implementation is publicly available.\footnote{Our implementation is accessible through \url{https://github.com/SijingTu/2023-OPTRefinement-code/}}
In this section, we only report our findings for \dskc, 
while the results for \maxcutkc are in Appendix~\ref{sec:exp:cut}.
We aim to answer the following questions: 
\begin{description}
\item[{\rm{\bf RQ1}:}] Do our algorithms increase the density of the ground-truth subgraph?
\item[{\rm{\bf RQ2}:}] Do our algorithms restore a good solution after removing some nodes from the ground-truth subgraph?
\item[{\rm{\bf RQ3}:}] What are the differences between the nodes identified by different algorithms?
\item[{{\rm{\bf RQ4}:}}] {Are our algorithms scalable?}
\item[{\rm{\bf RQ5}:}] Which algorithm has the best performance?
\end{description}

\spara{Datasets.}
We use three types of datasets: 
(1)~Wikipedia politician page networks where the nodes represent 
politicians, who are labeled with their political parties, 
and the edges represent whether one politician appears in the Wikipedia page of
another politician.
(2)~Networks labeled with ground-truth communities collected from
SNAP~\citep{snapnets}.
(3)~Synthetic networks generated using the stochastic block model (SBM). 

As for the \emph{ground-truth subgraphs}, let $G[U]$ be the ground-truth
subgraph induced by the set of vertices $U$; we select $U$ in different ways
depending on the datasets.  For the real-world datasets, we select $U$ according
to the nodes' labels: specifically, for Wikipedia networks, we select all the
nodes from one political party as $U$, and for SNAP networks, we select the
largest community as $U$.  As for the synthetic networks generated by SBM, we
always set one of the planted communities as $U$.  

The statistics of the datasets and the properties of the ground-truth subgraphs
are {presented in both Table~\ref{tab:not-move-out-10} and Table~\ref{tab:appendix:move-out-10}}. 
More detailed descriptions,
including how we set the parameters to generate our random graphs, can be found
in Appendix~\ref{sec:add-exp:data}.

\begin{table*}[t]
	\centering
	\caption{
	{\small 
	Network statistics and average relative increase {in} density with respect to the ground-truth subgraphs.
	Here, $n$ and $m$ {represent} the number of nodes and edges in the graph;
  $n_0$ and $\rho_0$ {denote} the number of nodes and density of the ground-truth subgraph;
  $n^*$ and $\rho^*$ {indicate} the number of nodes and density of the densest subgraph.
    An algorithm not terminating within 2 hours is denoted by \_\,;
	{for \denseSQD, it also indicates that the algorithm does not output a result, }
	as no $\sigma$-quasi-elimination order exists. 
	We set $k = 10\% \, n_0$.}}
	\label{tab:not-move-out-10}
	\resizebox{\textwidth}{!}{%
		\begin{tabular}{@{}rrrrrrrrrrrrrrrrrrrrr}
    \toprule
  \multirow{2}{*}{\textsf{dataset}} & 
  \multicolumn{6}{c}{Network Statistics} &
  \multicolumn{3}{c}{Algorithms} & 
  \multicolumn{2}{c}{Blackbox Methods}&
  \multicolumn{2}{c}{Baselines} \\
  \cmidrule(lr){2-7}
  \cmidrule(lr){8-10} 
  \cmidrule(lr){11-12} 
  \cmidrule(lr){13-14}
& $n$ & $m$ & $n_0$ & $\rho_0$  & $n^*$ & $\rho^*$ & \denseGreedy & \denseSDPalgo & \denseSQD & \densePeelMerge & \denseSDPMerge & \denserandom \\
\midrule
\balanced & 1\,000 & 74\,940 & 250 & 37.39 & 1\,000 & 74.94 & \rewrite{\textbf{0.007}} & \rewrite{\textbf{0.004}} & \rewrite{0.003} & \rewrite{\emph{0.003}} & \rewrite{\emph{0.003}} & \rewrite{-0.043}\\
\dense & 1\,000    & 81\,531 & 250 & 99.68 & 250 & 99.68 & \textbf{-0.057} & \textbf{-0.057} & \_ & \textbf{-0.057} & \textbf{-0.057} & \rewrite{-0.078}\\
\sparse & 1\,000   & 81\,192 & 250 & 24.93 & 250 & 100.10 & \rewrite{0.059} & \rewrite{\textbf{0.065}} & \rewrite{0.036} & \rewrite{\textbf{0.065} }& \rewrite{\textbf{0.065}} & \rewrite{-0.021}\\
\midrule
\es & 205 & 372 & 128 & 1.56 & 22 & 3.09 & \textbf{0.202} & 0.193 & 0.170 & \emph{0.193} & \emph{0.193} & \rewrite{-0.047}\\
\de & 768 & 3\,059 & 445 & 3.28 & 129 & 7.80 & 0.270 & \textbf{0.277} & 0.254 & \textbf{0.277} & \textbf{0.277} & \rewrite{-0.060}\\
\gb & 2\,168 & 18\,617 & 1\,160 & 6.55 & 234 & 17.85 & \textbf{0.321} & \textbf{0.321} & 0.290 & \textbf{0.321} & \textbf{0.321} & \rewrite{-0.072}\\ 
\us & 3\,912 & 18\,359 & 2\,014 & 2.97 & 303 & 11.37 & \textbf{0.595} & \textbf{0.595} & 0.567 & \rewrite{0.594} & \textbf{0.595} & \rewrite{-0.049}\\
\midrule
\dblp    & 317\,080    & 1\,049\,866  & 7\,556 & 2.40  & 115 & 56.57  & 0.436 & \_ & \emph{0.849} & \rewrite{\textbf{1.151}} & \_ & -0.084\\
\amazon  & 334\,863    & 925\,872     & 328 & 2.07          & 97 & 4.80    & \emph{0.075} & \_ & \_ & \rewrite{\textbf{0.079}} & \_ & -0.089\\
\youtube & 1\,134\,890 & 2\,987\,624  & 2\,217 & 4.09     & 1\,959 & 45.60  & 
\textbf{1.586} & \_ & \_ & \rewrite{\textbf{1.586}} & \_ & -0.088\\
\bottomrule
\end{tabular}
	}
\end{table*}

\begin{table*}[t]
	\centering
	\caption{
	\small{Network statistics and average relative increase {in} density with respect to the ground-truth subgraph with $k= 10\%\, n_0$ random nodes moved out. 
	Here, $n$ and $m$ {represent} the number of nodes and edges in the graph;
  $n_0$ and $\rho_0$ {denote} the number of nodes and density of the ground-truth subgraph;
  $n^*$ and $\rho^*$ {indicate} the number of nodes and density of the densest subgraph.
	An underscore (\_) {denotes} that an algorithm does not finish within 2 hours;
	{for \denseSQD, it also indicates that the algorithm does not output a result, }
	as no $\sigma$-quasi-elimination order exists. 
	We set $\sigma = 5$.
	{\denseinit denotes the solution containing the $k$ nodes removed from the ground-truth subgraph.}}
	}
	\label{tab:appendix:move-out-10}
	\resizebox{\textwidth}{!}{
\begin{tabular}{@{}rrrrrrrrrrrrrrrrrrrrr}
    \toprule
  \multirow{2}{*}{\textsf{dataset}} & 
  \multicolumn{6}{c}{Network Statistics} &
  \multicolumn{3}{c}{Algorithms} & 
  \multicolumn{2}{c}{Blackbox Methods}&
  \multicolumn{2}{c}{Baselines} \\
  \cmidrule(lr){2-7}
  \cmidrule(lr){8-10} 
  \cmidrule(lr){11-12} 
  \cmidrule(lr){13-14}
& $n$ & $m$ & $n_0$ & $\rho_0$  & $n^*$ & $\rho^*$ & \denseGreedy & \denseSDPalgo & \denseSQD & \densePeelMerge & \denseSDPMerge & \denserandom & \denseinit \\
\midrule
\balanced & 1\,000 & 74\,940 & 250 & 37.39 & 1\,000 & 74.94 & \rewrite{\textbf{0.110}} & \rewrite{\textbf{0.110}} & \_ & \rewrite{0.102} & \rewrite{\textbf{0.110}} & -0.045 & \rewrite{\textbf{0.110}}\\
\dense & 1\,000    & 81\,531 & 250 & 99.68 & 250 & 99.68 
& \textbf{0.112} & \textbf{0.112} & \textbf{0.112} & \textbf{0.112} & \textbf{0.112} & \rewrite{-0.078} & \textbf{0.112}\\
\sparse & 1\,000   & 81\,192 & 250 & 24.93 & 250 & 100.10 & \rewrite{\textbf{0.112}} & \rewrite{0.109} & \rewrite{\_} & \rewrite{0.082} & \rewrite{\emph{0.111}} & \rewrite{-0.015} & 0.111\\
\midrule
\es & 205 & 372 & 128 & 1.56 & 22 & 3.09 & \rewrite{\textbf{0.424}} & \rewrite{0.420} & \rewrite{0.389} & \rewrite{0.408} & \rewrite{\emph{0.420}} & \rewrite{-0.069} & \rewrite{0.200}\\
\de & 768 & 3\,059 & 445 & 3.28 & 129 & 7.80 & \rewrite{0.378} & \rewrite{\textbf{0.380}} & \rewrite{0.358} & \rewrite{0.379} & \rewrite{\textbf{0.380}} & \rewrite{-0.084} & \rewrite{0.098}\\
\gb & 2\,168 & 18\,617 & 1\,160 & 6.55 & 234 & 17.85 & \rewrite{\textbf{0.421}} & \rewrite{\textbf{0.421}} & \rewrite{0.370} & \rewrite{\textbf{0.421}} & \rewrite{\textbf{0.421}} & \rewrite{-0.085} & \rewrite{0.084}\\
\us & 3\,912 & 18\,359 & 2\,014 & 2.97 & 303 & 11.37 & \rewrite{\textbf{0.775}} & \rewrite{0.774} & \rewrite{0.735} & \rewrite{\textbf{0.775}} & \rewrite{\textbf{0.775}} & \rewrite{-0.035} & \rewrite{0.125}\\
\midrule
\dblp       & 317\,080   & 1\,049\,866 & 7\,556 & 2.40  & 115 & 56.57 & \rewrite{0.542} & \_ & \rewrite{\emph{1.046}} & \rewrite{\textbf{1.417}} & \_ & \rewrite{-0.092} & \rewrite{0.109}\\
\amazon  & 334\,863    & 925\,872 & 328 & 2.07  & 97 & 4.80 & \rewrite{\textbf{0.160}} & \_ & \_ & \rewrite{\emph{0.134}} & \_ & \rewrite{-0.097} & \rewrite{0.098}\\
\youtube& 1\,134\,890 & 2\,987\,624 & 2\,217 & 4.09 & 1\,959 & 45.60  & \rewrite{\emph{1.893}} & \_ & \_ & \rewrite{\textbf{1.894}} & \_ & \rewrite{-0.096} & \rewrite{0.135}\\
\bottomrule
\end{tabular}

	}
\end{table*}

\spara{Algorithms.} 
We evaluate six algorithms for solving \dskc.
We use \denseSDPalgo from Theorem~\ref{thm:densest-subgraph-sdp},
\denseGreedy from Section~\ref{sec:densest-subgraph:heuristic}, 
and two methods based on the black-box reduction from Theorem~\ref{thm:densest-subgraph-reduction}. 
In the black-box reduction, we use two solvers for \dks:  
the peeling algorithm of~\citet{asahiro2000greedily}
and the {\sdp} algorithm of~\citet{DBLP:journals/jal/feigel01},
which we denote as \densePeelMerge and \denseSDPMerge, respectively.
We also compare our methods with the \denseSQD ($\sigma$\emph{-quasi-densify}) 
algorithm~\citep{matakos2022strengthening}, which solves \kdense, where we set $\sigma =5$ 
as in the original paper.
We also consider a baseline that picks $k$ random nodes~(\denserandom);
for \denserandom, we repeat the procedure $5$ times and compute the average results. 

\spara{Evaluation.} 
We evaluate all methods with respect to the relative increase of density,
where we compare to the density of the initial subgraph~$U$.
Formally, assuming that a method $\algo$ returns a subgraph $U_{\!\algo}$,
the \emph{relative increase of density} with respect to the density of $U$ is ${(d(U_{\!\algo}) - d(U))}/{d(U)}$. 
We say that method $\algo$ performs $X$ times better than method $\algoB$, 
if the relative increase of density 
from the subgraph returned by $\algo$ is $X$ times larger than $\algoB$.  

\spara{Initialization using ground-truth subgraphs.}
{We begin by examining} whether our algorithms indeed increase the density of a given subgraph. 
{For this purpose}, we initialize~$U$ as a ground-truth subgraph from our dataset and apply the algorithms directly to the instance $(G, U, k)$. 
The relative increases {in} densities {compared to} $d(U)$ are presented in Table~\ref{tab:not-move-out-10} and Figure~\ref{fig:density-direct-ratio}.

{As shown in Table~\ref{tab:not-move-out-10} and Figure~\ref{fig:density-direct-ratio}, 
all algorithms increase the density of the ground-truth subgraph in most datasets. } 
Specifically, in Table~\ref{fig:density-direct-ratio}, on the datasets with larger ground-truth subgraphs, 
i.e., \us, \dblp {and \youtube}, the relative increase of densities surpasses $40\%$. 
The only outlier in which the algorithms decrease the density is \dense; this is because the ground truth 
subgraph is already the densest subgraph.
These findings indicate that our algorithms are capable of identifying denser subgraphs over the input
instances, addressing {\bf RQ1}.

\begin{figure}[t]
    \centering
 	\inputtikz{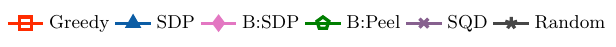} \\
	\begin{tabular}{cccc}
		\resizebox{0.23\columnwidth}{!}{%
			\inputtikz{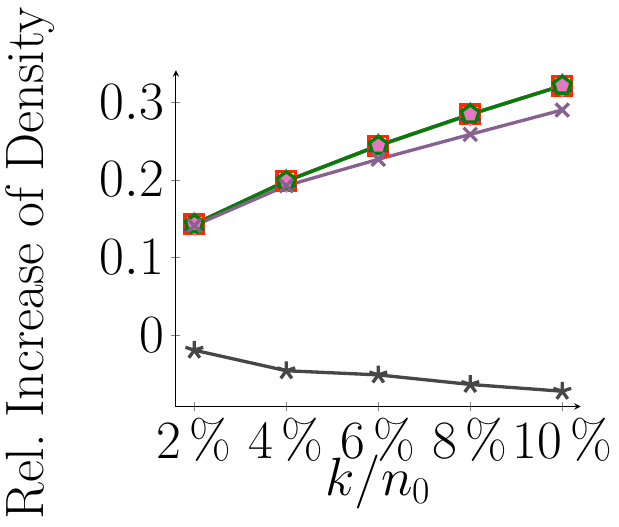}
		}&
		\resizebox{0.23\columnwidth}{!}{%
			\inputtikz{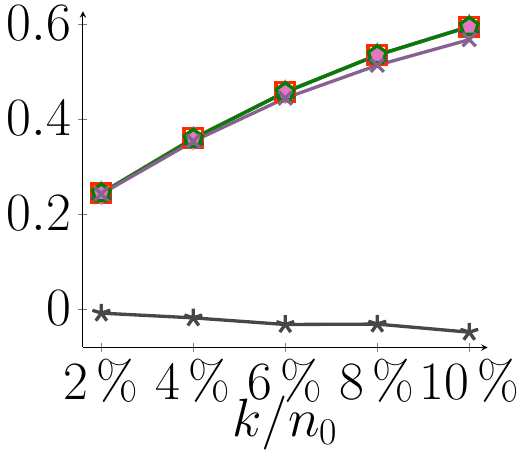}
		}&
		\resizebox{0.23\columnwidth}{!}{%
			\inputtikz{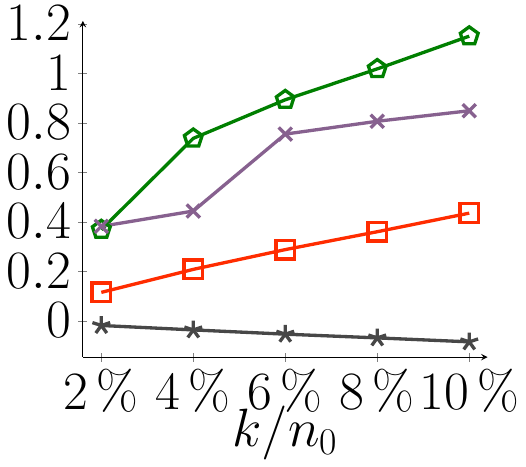}
		}&
		\resizebox{0.23\columnwidth}{!}{%
			\inputtikz{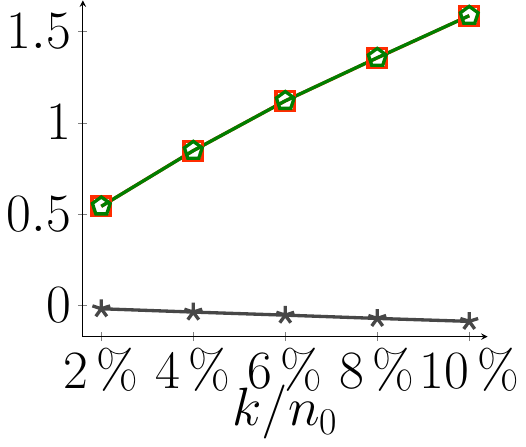}
		}
		\\
		{(a)~{\gb}} &
		{(b)~{\us}} &
		{(c)~{\dblp}} &
		{(d)~{\youtube}} \\
	\end{tabular}
	\caption{Relative increase of density for varying values of~$k$.
	We initialized $U$ as a ground-truth subgraph.
    No results of \SDPalgo and \denseSDPMerge are reported on \dblp {and \youtube} due to the prohibitive time complexity.}
	\label{fig:density-direct-ratio}
\end{figure}

\begin{figure}[t]
\centering 	
 \inputtikz{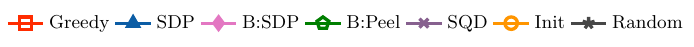} \\
	\begin{tabular}{cccc}
		\resizebox{0.23\columnwidth}{!}{%
		    \inputtikz{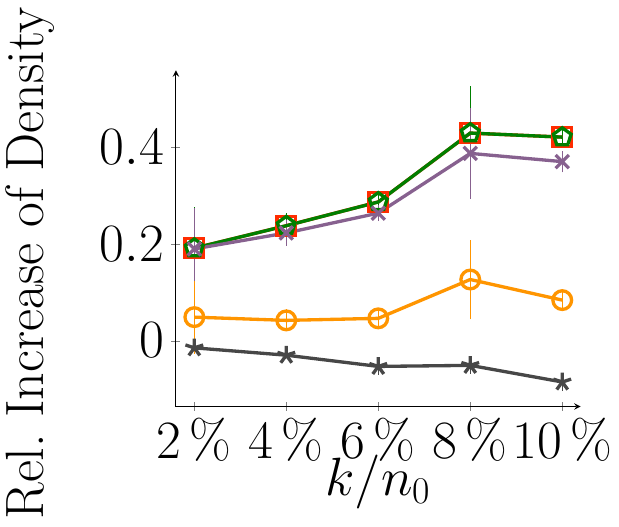}
		}&
		\resizebox{0.23\columnwidth}{!}{%
			\inputtikz{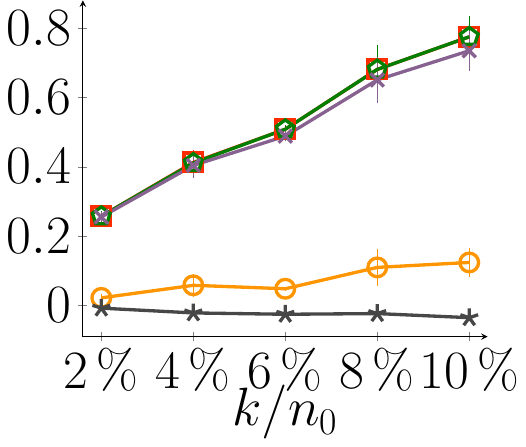}
		}&
		\resizebox{0.23\columnwidth}{!}{%
			\inputtikz{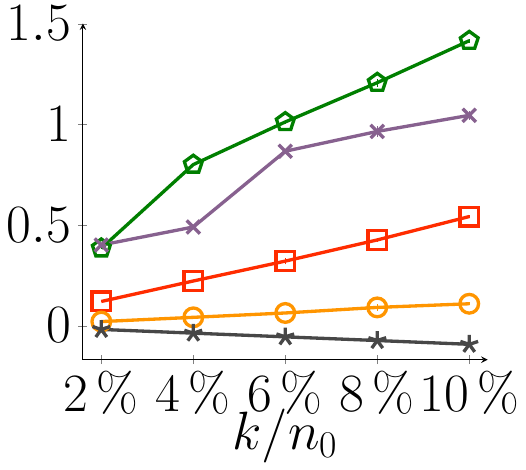}
		}&
		\resizebox{0.23\columnwidth}{!}{%
			\inputtikz{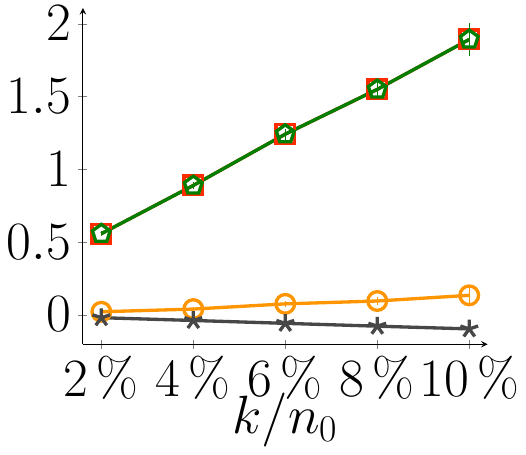}
		}
		\\
		{(a)~{\gb}} &
		{(b)~{\us}} &	
		{(c)~{\dblp}} &
		{(d)~{\youtube}} \\
	\end{tabular}
	\caption{Relative increase of density for varying values of~$k$.
	   We initialized $U$ as a {ground-truth subgraph with $k$~nodes removed uniformly at random}.
    No results of \SDPalgo and \denseSDPMerge are reported on \dblp {and \youtube} due to the prohibitive time complexity.}
	\label{fig:density-out-select-ratio}
\end{figure}

\spara{Initialization via ground-truth subgraphs with $k$ nodes removed.}  
Next, we study whether our algorithms can uncover subgraphs of
	densities equal to or greater than those of the ground-truth subgraphs
	after they were slightly perturbed.
{For this purpose}, we initialize $\hat{U}$ as a ground-truth subgraph $U$ with $k$ random nodes removed. 
We then apply the algorithms to the instance $(G, \hat{U}, k)$. 
We let \denseinit denote the solution containing the $k$ nodes that are removed to obtain $\hat{U}$. 
Our results are presented in {Table~\ref{tab:appendix:move-out-10}} and Figure~\ref{fig:density-out-select-ratio}, which show the average relative increase in density with respect to $d(\hat{U})$ across 5 initializations.
{In Figure~\ref{fig:density-out-select-ratio} we also show} the corresponding standard deviations.

{Table~\ref{tab:appendix:move-out-10}} and Figure~\ref{fig:density-out-select-ratio} show that, across various datasets, 
all algorithms consistently find subgraphs that are denser than \denseinit for all values of $k$. 
Additionally, the difference in density between the algorithms and \denseinit increases 
as more nodes are removed (i.e., larger $k$). 
Our observation suggests that our algorithms are as good, if not better, 
at finding subgraphs with comparable densities, and this  
observation addresses {\bf RQ2}.

Next, study whether our algorithms can reconstruct ground-truth subgraphs.
We again start with the ground-truth subgraphs and randomly remove $k$ nodes. 
Then, we run our algorithms to see if they can identify the removed $k$ nodes 
or if they select different nodes that were not part of the ground-truth subgraph.
The results are reported in Figure~\ref{fig:similarity-country}, which illustrates the 
Jaccard similarity among the nodes selected by the algorithms and those identified by \denseinit. 
The Jaccard similarity between two sets $A$ and $B$ is defined as $J(A, B) = \frac{|A \cap B|}{|A \cup B|}$.

Figure~\ref{fig:similarity-country} shows distinct behaviors of the algorithms for 
both the real-world and the synthetic datasets. 
We start by analyzing the results on the Wikipedia datasets. 
We notice that these datasets contain sparse ($\rho_0$ is around half of $\rho^*$) 
and large (around half of the nodes) ground-truth subgraphs, and many nodes within the ground-truth subgraph 
maintain more connections to the nodes outside of this subgraph.
This structure leads all the algorithms to select different nodes than \denseinit.
Moreover, the nodes selected by \denseSDPalgo, \denseSDPMerge, and \densePeelMerge 
have a high degree of similarity in all three datasets, whereas \denseGreedy demonstrates a distinct pattern in \de. 
This distinction can be attributed to the ``peeling'' routine inherent in the prior three algorithms, 
which iteratively \emph{removes} nodes with the lowest degree starting from the
graph from the intermediate step.
On the other hand, \denseGreedy iteratively \emph{adds} nodes with the highest
degrees.
The different procedures thus result in very different node selections in some datasets.  

Next, we analyze the results on the synthetic datasets.
Recall that in \sparse, despite the ground-truth subgraph being sparse, 
each node within this subgraph maintains more connections internally than with
nodes outside (in expectation). 
Under this construction, \denseGreedy (which does not contain the ``peeling'' routine) and 
\densePeelMerge (which highly relies on the ``peeling'' routine) exhibit
significant performance differences. 
\denseGreedy exhibits high similarity to \denseinit, indicating that it finds the same set of $k$ removed nodes. 
However, \densePeelMerge finds a different set of nodes compared to \denseinit. 
Combined with its superior performance over \denseGreedy on the \dblp dataset in 
Figure~\ref{fig:density-direct-ratio} and Figure~\ref{fig:density-out-select-ratio}, 
this behavior suggests that \densePeelMerge tends to find globally denser structures 
when the ground-truth subgraph is much sparser compared to the densest subgraph.
On the other hand, \denseGreedy effectively uncovers locally dense structures, 
albeit at the expense of overlooking globally denser structures that are less interconnected 
within the local subgraph.

We now look into \dense. All algorithms show high similarity to \denseinit,
which can be explained as the ground-truth subgraph is the densest subgraph and is much denser than other subgraphs.
Hence, any sufficiently good algorithm should recover this ground-truth subgraph.
Lastly, we look into \balanced, where only \densePeelMerge shows a different
selection compared to \denseinit, further implying its ``non-local'' property. 

We remark that the performance of \denseSDPalgo and \denseSDPMerge is between
\densePeelMerge and \denseGreedy. 
Specifically, they find the removed nodes for \balanced, 
and find nodes that belong to a denser subgraph for \sparse. 
As \denseSQD's choice is generally different from all the other algorithms, 
we present a more detailed analysis of \denseSQD in Appendix~\ref{sec:add-exp:sqd}. 

As anticipated, the algorithms exhibit varied node selection behaviors across different datasets. 
\denseGreedy favors nodes with stronger connections to the initial subgraph.
In contrast, \densePeelMerge is inclined to identify nodes with globally higher connectivity. 
\denseSDPMerge and \denseSDPalgo offer a balance between these frameworks, addressing {\bf RQ3}.

\begin{figure}[t!]
	\centering 
    \begin{tabular}{ccc}
		\includegraphics[width=0.24\columnwidth]{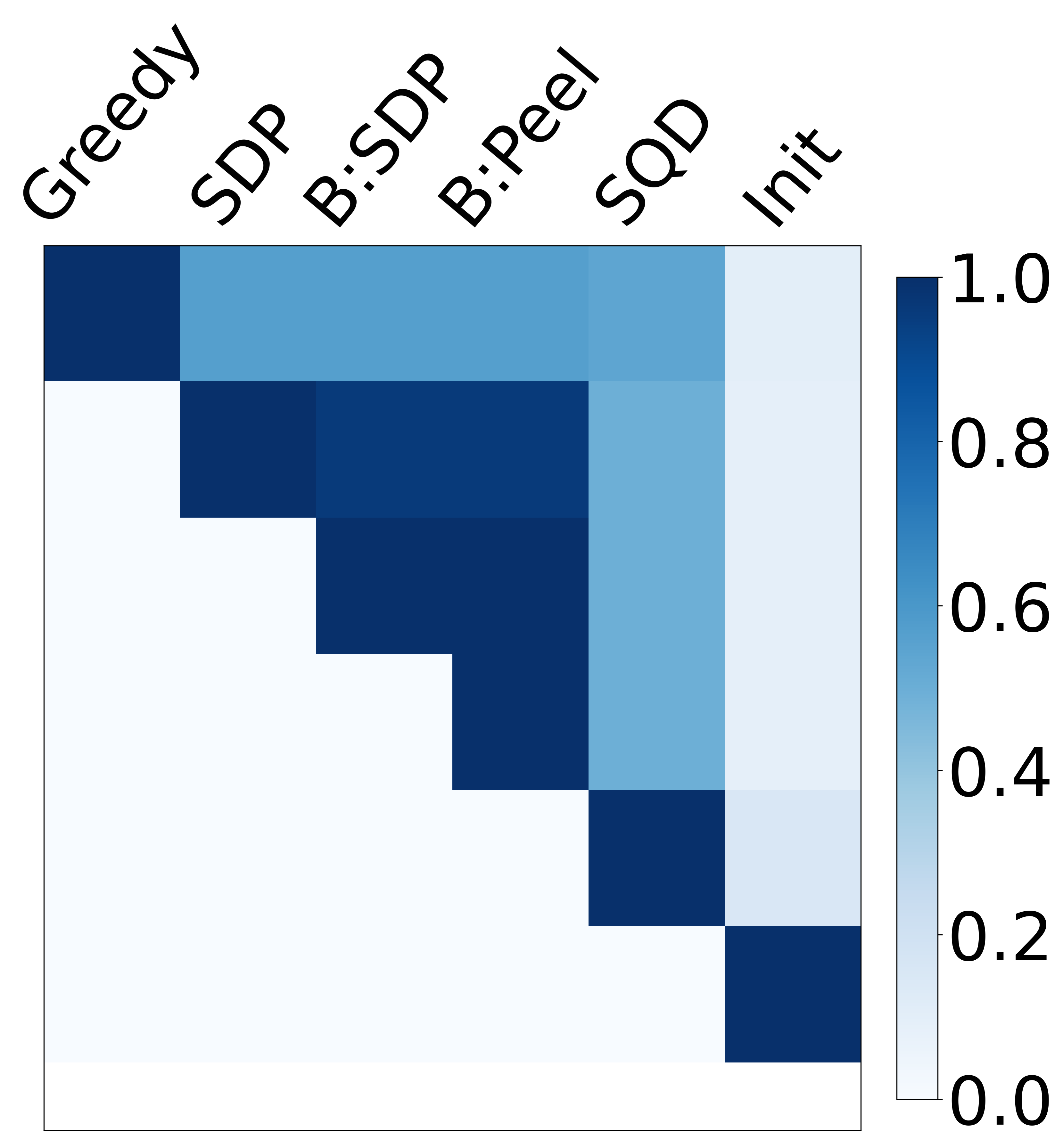}&
		\includegraphics[width=0.24\columnwidth]{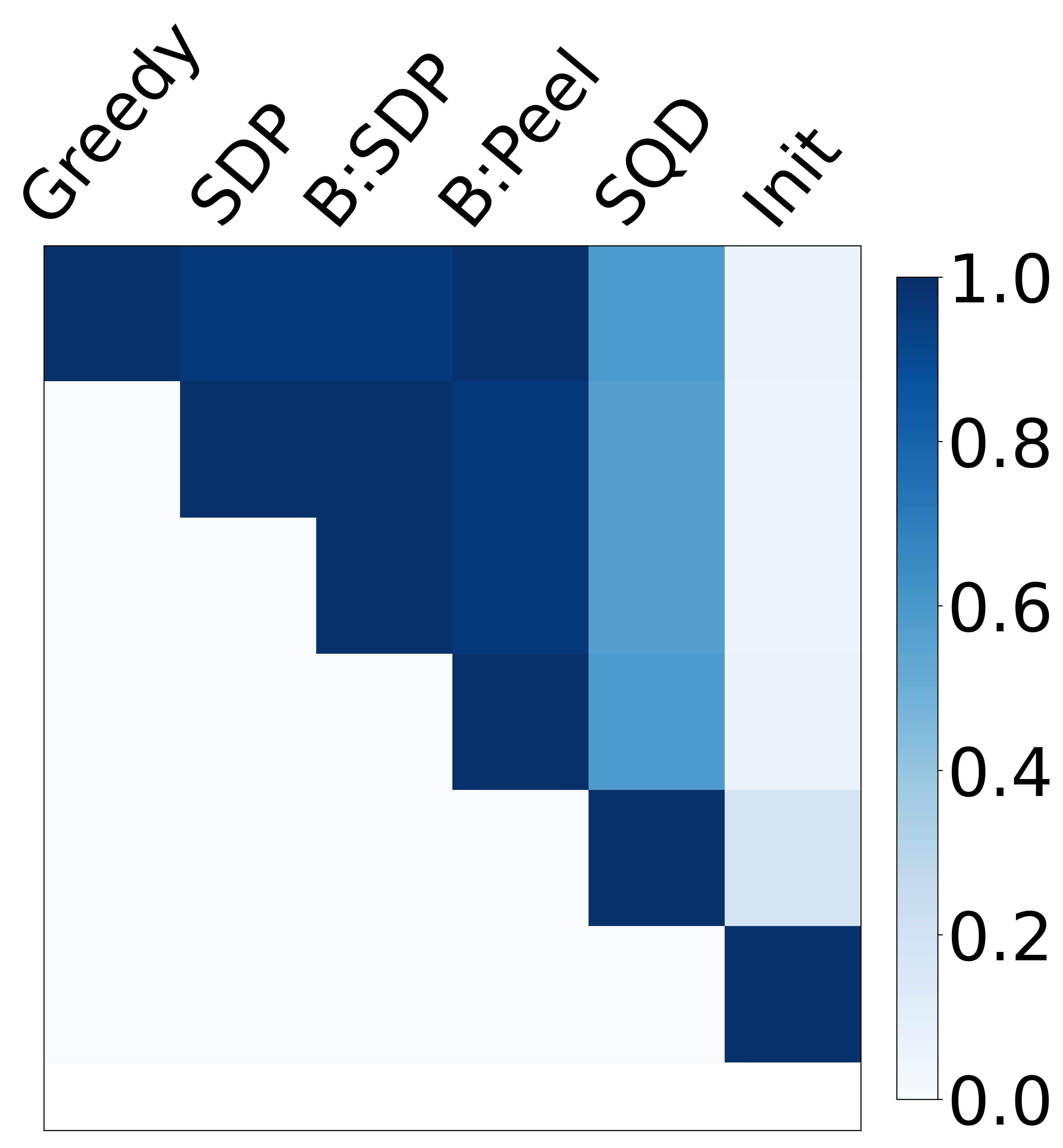}&
		\includegraphics[width=0.24\columnwidth]{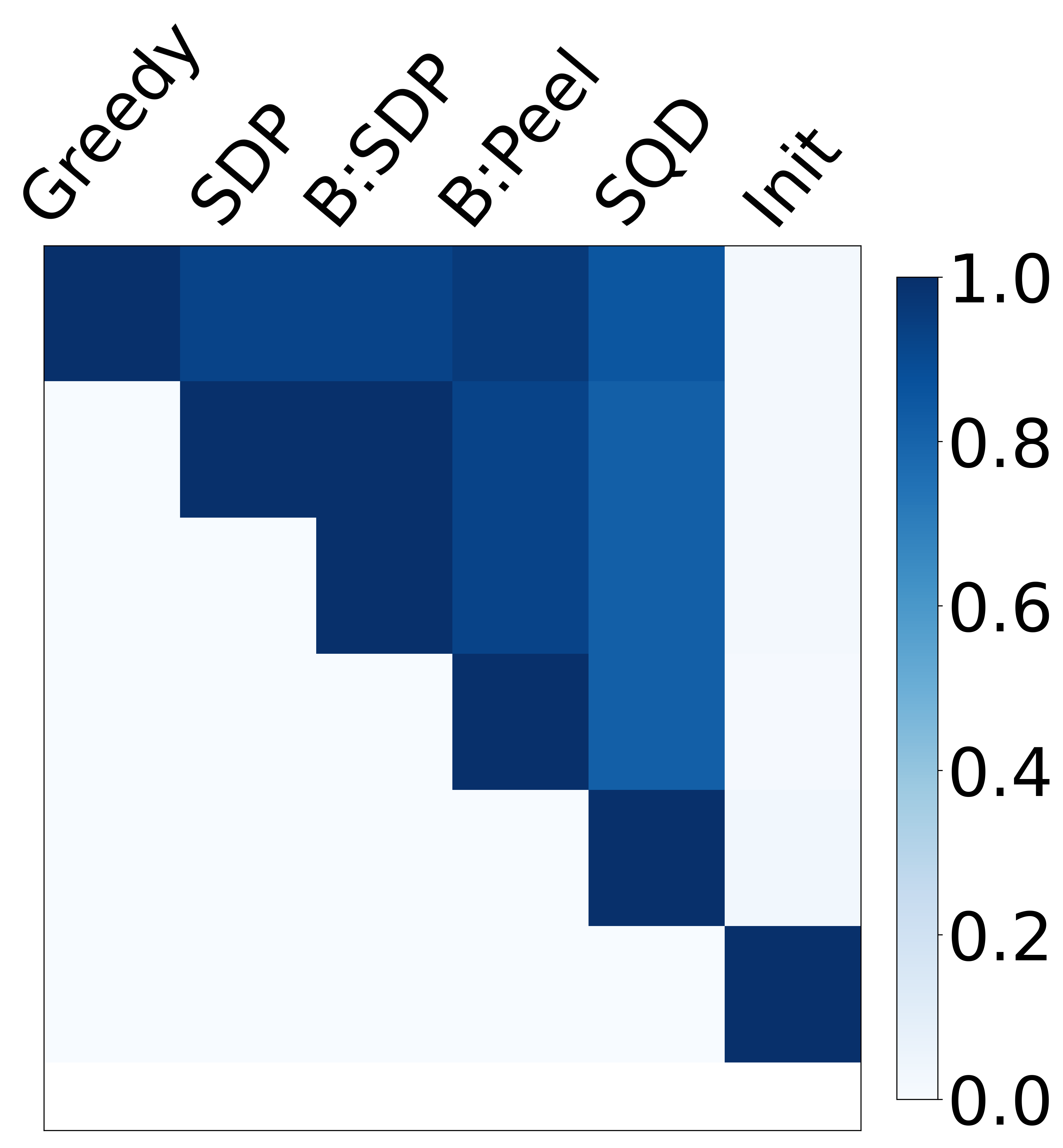}
		\\
		(a)~{\de} &
		(b)~{\gb} &
		(c)~{\us}  \\
		\includegraphics[width=0.24\columnwidth]{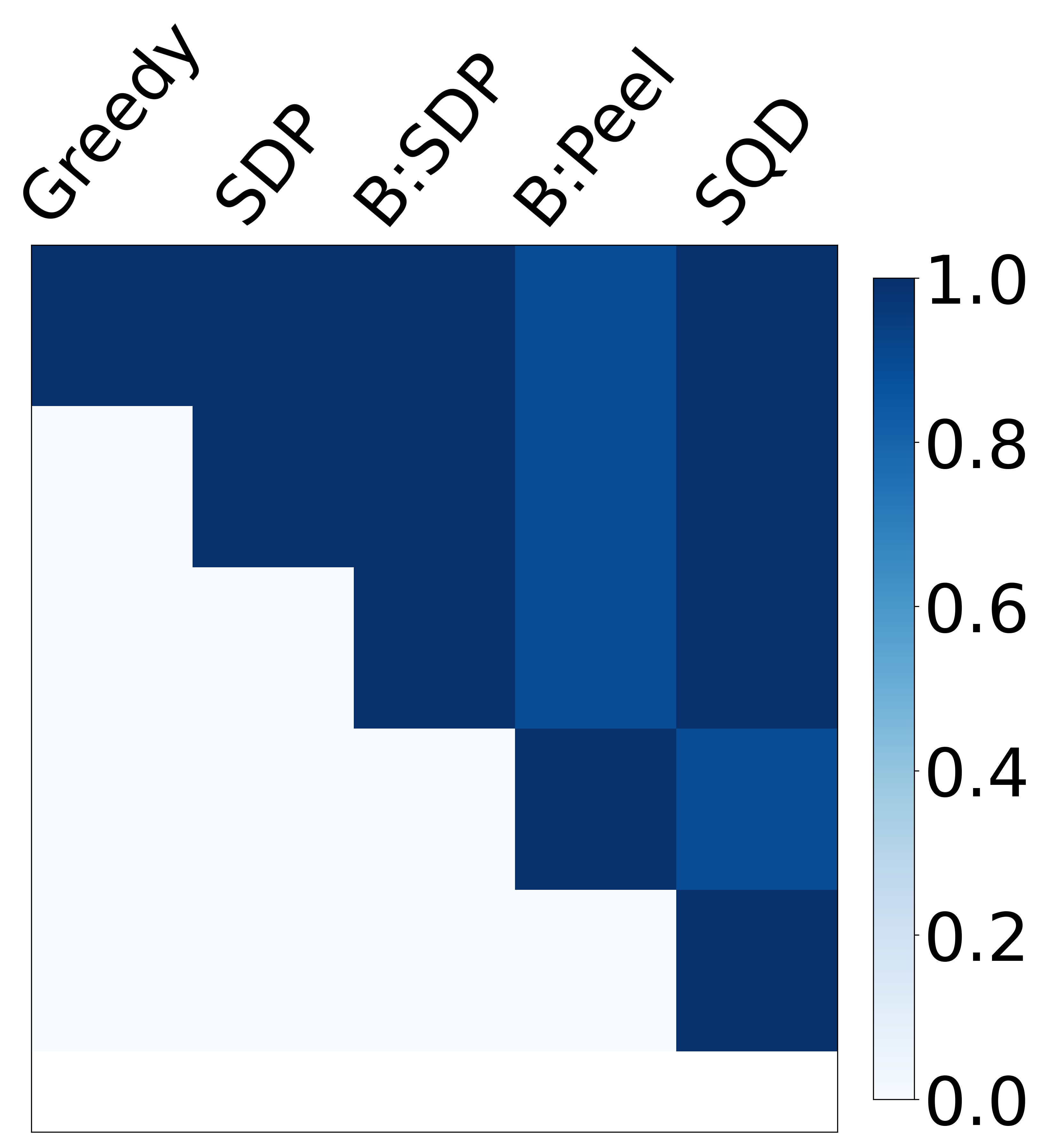}&
		\includegraphics[width=0.24\columnwidth]{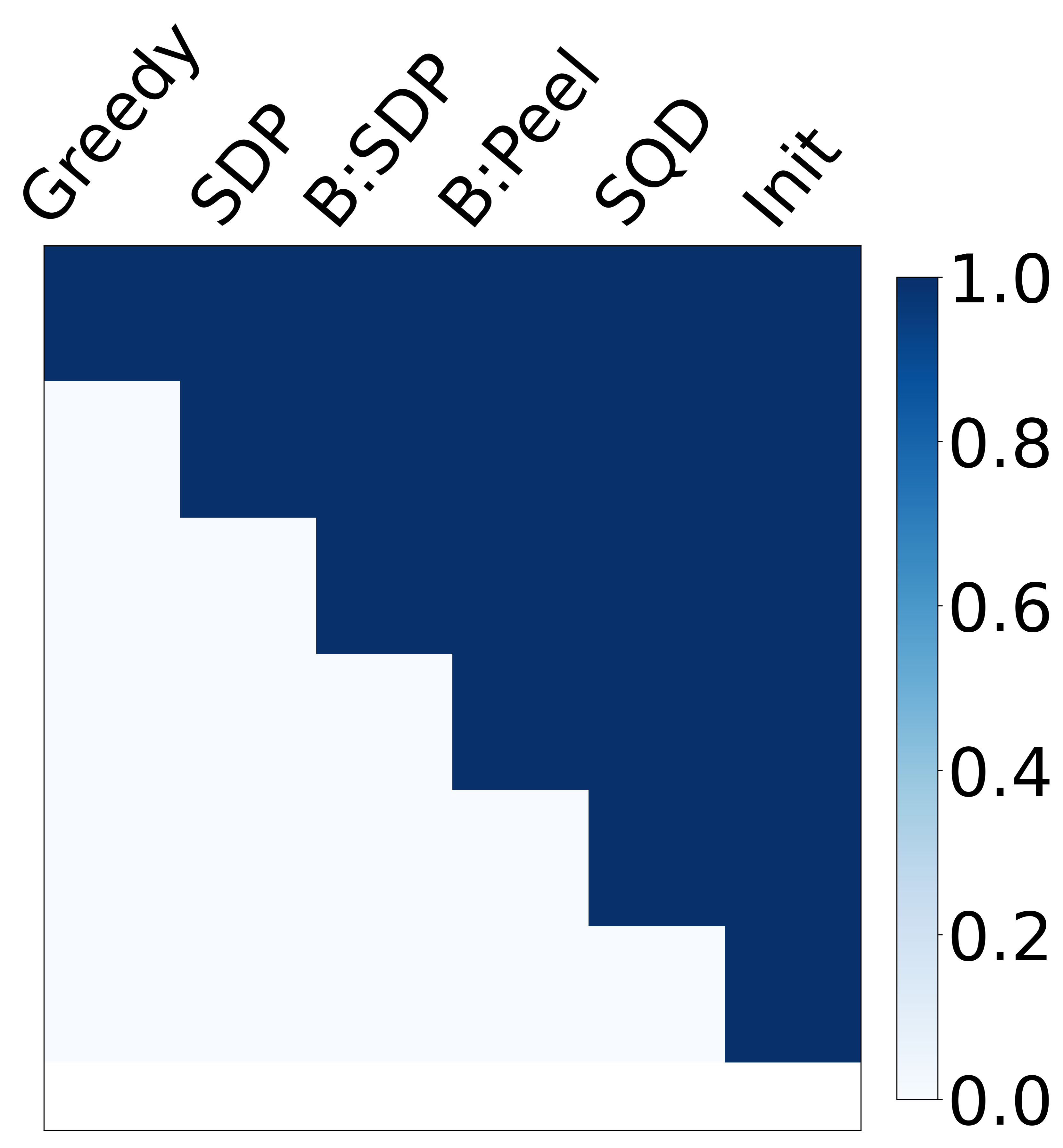}&
		\includegraphics[width=0.24\columnwidth]{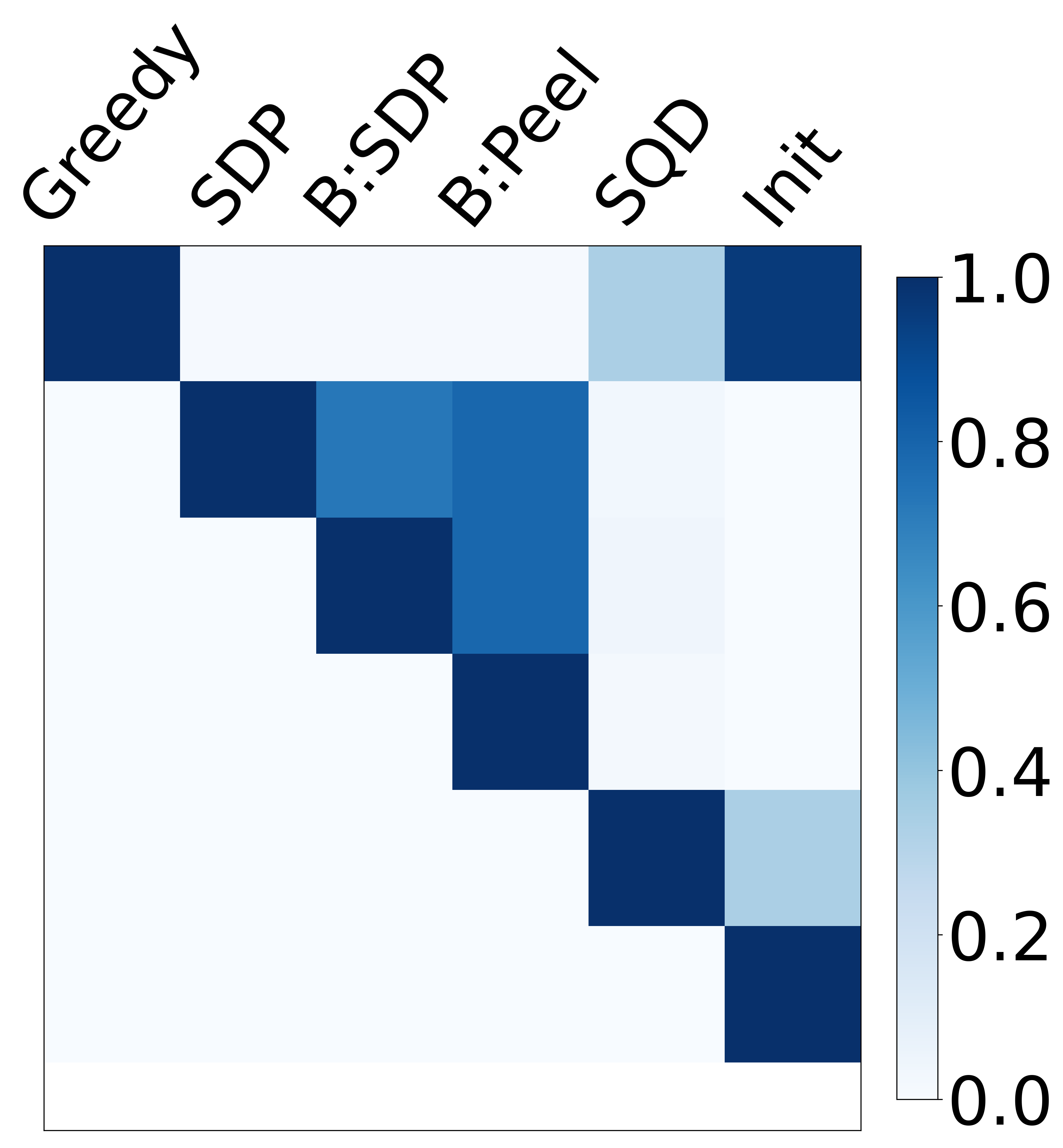}
		\\
		(d)~{\balanced}&
		(e)~{\dense}&
		(f)~{\sparse} \\
	\end{tabular}
	\caption{Jaccard similarities between 
 the nodes selected by different algorithms. 
 We initialized $U$ as a {ground-truth subgraph with $k$~nodes removed uniformly at random}.
 {\denseinit denotes the solution containing the $k$ nodes removed from the ground-truth subgraph.}
 The similarity matrix is symmetric along the diagonal; to make a clean presentation, we only show its upper triangle. Here we set $k = 50$. }
	\label{fig:similarity-country}
\end{figure}

\spara{Scalability of the algorithms.}
{
In Figure~\ref{fig:density-out-country-ratio-time}, we present the running times
of our algorithms for the four largest datasets. 
We observe that \denseSDPalgo and \denseSDPMerge are not scalable due to the 
high cost of solving the \sdp.
\denseSQD also becomes prohibitive for larger datasets, exhibiting super-linear running times as $k$ increases.
The baseline approach \denseGreedy is also not scalable on the largest dataset, \youtube, taking around $20$ minutes. 
The most scalable algorithm is \densePeelMerge, which runs in $100$ seconds on \youtube with one million nodes. 
This is because the core procedure of \densePeelMerge only requires repeatedly
removing nodes with the lowest degree, and the entire algorithm can be
implemented in time $O(n+m)$.
In conclusion, we find that \densePeelMerge is the most scalable algorithm among all algorithms, addressing {\bf RQ4}.
}

\begin{figure}[t!]
	\inputtikz{ds_plots/legend_move_out}
	\centering 
    \begin{tabular}{cccc}
		\resizebox{0.25\columnwidth}{!}{%
			\inputtikz{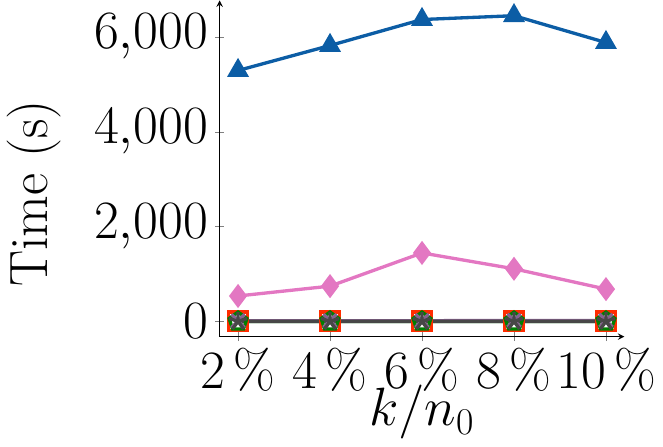}
		}&
		\resizebox{0.23\columnwidth}{!}{%
			\inputtikz{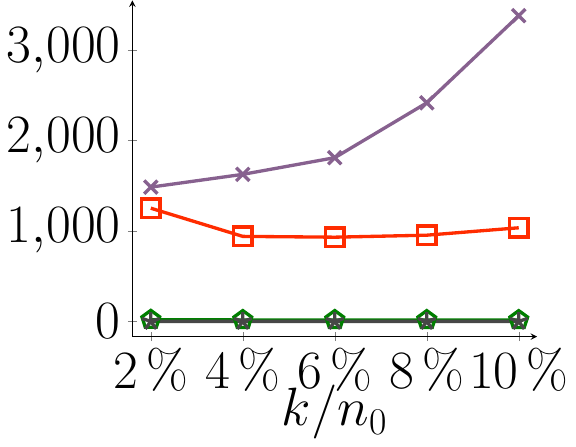}
		}&
        \resizebox{0.22\columnwidth}{!}{%
			\inputtikz{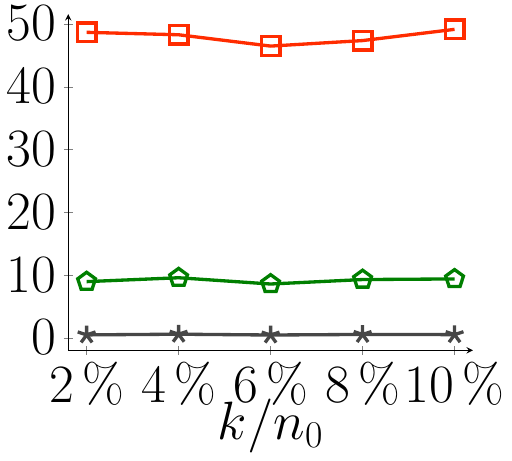}
		}&
        \resizebox{0.23\columnwidth}{!}{%
			\inputtikz{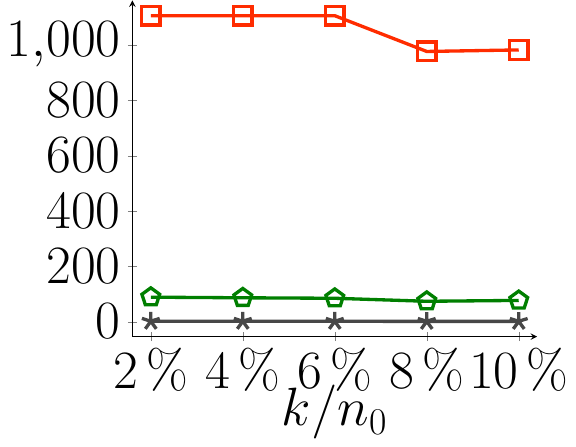}
		}\\
		(a) {\us} &
		(b) {\dblp} &
		(c) {\amazon} &
		(d) {\youtube} \\		
	\end{tabular}
	\caption{Running time for varying values of $k$. We initialize $U$ as a ground-truth subgraph and then
	remove $k$ nodes uniformly at random. We set $k = 2\%, 4\%, 6\%, 8\%, 10\% n_0$.
		}
	\label{fig:density-out-country-ratio-time}
\end{figure}

\spara{Selecting the best algorithm.} 
Last, we summarize the performance of all algorithms and select the best algorithm according to different applications.

For smaller datasets, such as Wikipedia and SBM datasets, 
the differences among the algorithms are marginal. 
However, \denseSDPalgo, \denseSDPMerge, {and \denseGreedy} show a slight advantage over \densePeelMerge, as shown in Table~\ref{tab:not-move-out-10} and Table~\ref{tab:appendix:move-out-10}.

For mid-sized datasets such as \dblp and \amazon, 
\denseSDPalgo and \denseSDPMerge fail to produce results 
due to the computational complexity of solving a semidefinite program.
{In terms of output quality, \densePeelMerge is only outperformed by \denseGreedy on \amazon, as shown in Table~\ref{tab:appendix:move-out-10}. 
However, on \dblp, \densePeelMerge's performance is 
2.5 times better than that of \denseGreedy; on \youtube, \densePeelMerge also marginally outperforms \denseGreedy.}
This is because \densePeelMerge is more efficient at identifying globally denser subgraphs. 

To conclude, {\densePeelMerge} emerges as the preferable choice for tasks with strict scalability demands. 
In scenarios involving smaller graphs, \denseSDPalgo, {\denseSDPMerge and \denseGreedy} 
output high-quality results. 
For all other cases, \densePeelMerge is recommended for 
its balance of output quality and scalability, 
thus addressing {\bf RQ5}.

\spara{Summary of the appendix.}
More details on the running times are provided in 
Appendices~\ref{sec:add-exp:running-time-by-k} and \ref{sec:add-exp-running-time-by-n}. 
Additional analysis in Appendix~\ref{sec:add-exp:dense:sparse} 
explores the nodes selected by different algorithms on \sparse.
Figure~\ref{fig:appendix:density-out-select-ratio} and 
Figure~\ref{fig:appendix:density-not-out-select-ratio} 
illustrate the quality of various algorithms across a broader range of 
$k$ values using synthetic datasets. 
Last, a comprehensive analysis of \denseSQD is presented in 
Appendix~\ref{sec:add-exp:sqd}, where \denseSQD's distinct performances 
in selecting node sets and its consistent underperformance relative 
to other algorithms are discussed. 

\section{Conclusion}
\label{sec:conclusion}
{
In this paper, we introduce the \emph{OptiRefine framework}, which characterizes problems aimed at increasing the quality of an existing solution through a small number of refinements. 
Our contributions are as follows:
First, we generalize maximum graph-partitioning problems (\maxgp), introducing the graph partitioning with $k$ refinements (\gpkc) problem. 
Second, we develop \sdp-based approaches for solving \gpkc, providing detailed analyses of maximum cut with $k$ refinements (\maxcutkc) and densest subgraph with $k$ refinements (\dskc). 
Third, we propose scalable algorithms based on a black-box reduction by leveraging the connections between our problems and \dks and \maxcut.
}

{
From a theoretical perspective, we identify two areas for future investigation. 
The first concerns whether we can apply the \sdp-based approach to problems
with more general parameters \parzero, \parone, \partwo, and \parthree, which characterize problems beyond \maxgp.
The second area involves analyzing the computational complexity of \maxcutkc when $k \in o(n)$, particularly given that the linear relaxation approach with pipage rounding for \ccmaxcut achieves a $0.5$-approximation regardless of $k$'s regime~\citep{ageev1999approximation}.
}

{
From a practical perspective, we only briefly discussed strategies for selecting
the parameter $k$, and mostly treated $k$ as a given input. We consider
exploring principled methods for determining $k$ to be a worthwhile direction for future research.
}

\section*{Acknowledgment}
We sincerely thank Per Austrin for the discussions and for providing valuable feedback on the sum-of-squares algorithm. 
We appreciate Antonis Matakos for sharing the source code for the SQD algorithm.
This research is supported by the ERC Advanced Grant REBOUND (834862), 
the EC H2020 RIA project SoBigData++ (871042), 
the Vienna Science and Technology Fund (WWTF) [Grant ID: 10.47379/VRG23013], 
the Approximability and Proof Complexity project funded by the Knut and Alice Wallenberg Foundation
and the Wallenberg AI, Autonomous Systems and Software Program (WASP) funded by the Knut and Alice Wallenberg Foundation.
The computations were enabled by resources provided by the National Academic Infrastructure for Supercomputing in Sweden (NAISS), partially funded by the Swedish Research Council through grant agreement no. 2022-06725.

\bibliographystyle{plainnat}
\bibliography{paper}
\clearpage

\appendix
\clearpage
\section{Omitted content of Section~\ref{sec:densest-subgraph}}
\label{appendix:ds-proofs}

\subsection{Proof of Lemma~\ref{lem:densest-subgraph-generalization}}
Let $G = (V, E)$ and $k$ be the input for the \dks problem. 
Suppose that we have an algorithm~$\algo$ for the \dskc problem. 
We apply $\algo$ on~$G$ with~$U = \emptyset$ and with the same parameter $k$. 
The output of $\algo$ is a set~$C$ of size~$|C|=k$, which we use as the solution~$U$ for the \dks problem. 
Since we did not make any changes to $G$, the running time of $\algo$ on~$G$ is $\bigO(T(m,n,k))$. 
We obtain our approximation guarantees by observing that 
both problems have feasible solutions consisting of exactly 
$k$ vertices, and since $U = \emptyset$, the optimal solutions 
for both problems coincide.

\subsection{Proof of Corollary~\ref{cor:densest-subgraph-hardness}}
We conclude from Lemma~\ref{lem:densest-subgraph-generalization} that if there is no 
$\beta_k$-approximation algorithm for the \dks problem, 
then there also does not exist a $\beta_k$-approximation algorithm 
for the \dskc problem with the same 
parameter $k$.

\citet{DBLP:conf/stoc/LeeG22} showed
that the \dks problem for $k = \tau \cdot n$ is 
difficult to approximate within a factor of~$O(\tau \log(1/\tau))$ under the \emph{Small Set Expansion Conjecture}. 
This implies our corollary.

\subsection{Proof of Lemma~\ref{lemma:k-densify-local-changes-connection}}
Let $C_1^*$ be the optimal set of vertices selected for the \kdense problem.
Let $C_1$ be the set of vertices chosen by an $\alpha$-approximation algorithm for the \kdense problem.
Observe that $|C_1^*| = |C_1| = k$ and ${w(E[U \cup C_1])} \geq \alpha {w(E[U \cup C_1^*])}$.

Let $C^*$ be the optimal set of vertices selected for the \dskc problem. Notice that $\abs{C^*} = k$. 
Using the condition $k \leq c\abs{U}$, we find that $|U \symm C^*|\geq (1-c)|U|$ and $|U| \geq \frac{1}{(1 + c)} \abs{U \symm C^*}$. 
Consequently, $|U \symm C^*| \geq \frac{(1 - c)}{(1+c)}\abs{U \symm C_1^*} = \frac{(1 - c)}{(1+c)}\abs{U \cup C_1^*}$.
Notice that the equality holds because $k \leq n - \abs{U}$, and hence the
optimal set of vertices~$C_1^*$ for \kdense does not intersected with $U$,
implying that $U \symm C_1^* = U \cup C_1^*$. 

Since $C_1^*$ is the optimal set of vertices for the \kdense problem, we have ${w(E[U \cup C_1^*])} \geq {w(E[U \symm C^*])}$.

Based on the above analysis, we can derive the following inequality that illustrates the performance of the $\alpha$-approximation algorithm for the \kdense problem,

\begin{equation}
\begin{aligned}
\frac{{w(E[U \symm C^*])}}{\abs{U \symm C^*}} 
&\leq \frac{(1+c)}{(1-c)} \cdot \frac{{w(E[U \cup C_1^*])}}{\abs{U \cup C_1^*}} \\
&\leq \frac{1}{\alpha} \cdot \frac{(1+c)}{(1-c)} \cdot \frac{{w(E[U \cup C_1])}}{\abs{U \cup C_1}} .\\
\end{aligned}
\end{equation}

The second inequality holds since $C_1$ is an $\alpha$-approximation solution for the \kdense problem. 
This concludes the proof, demonstrating that an $\alpha$-approximate solution for the \kdense problem implies a $\frac{1-c}{1+c}\alpha$-approximate solution for the \dskc problem.

\subsection{Pseudocode for solving \kdense and \dskc with a \dks solver (in Section~\ref{sec:densest-subgraph:black-box})}
\label{appendix:densest-subgraph:black-box:pseudocode}

The pseudocode of our black-box method on \kdense to \dks
from Section~\ref{sec:densest-subgraph:black-box}, 
is illustrated as Algorithm~\ref{algo:dense-black}.

\begin{algorithm}[H]
\caption{Solve \kdense using a solver for \dks}
\label{algo:dense-black}
\begin{algorithmic}[1]
\Require{$(G=(V,E,w),\OldSet,k)$, access to a solver for \dksplus}
\State Create a new graph $G_1 = ((V\setminus \OldSet)\cup \{i^*\}, E_1, w_1)$,
	where $i^*$ corresponds to the contracted vertices in $U$  
\State Set $E_1 \gets E[V\setminus \OldSet] \cup \{ (i^*,j) \colon (i,j)\in E \text{ for } i\in \OldSet, j\in V\setminus \OldSet\}$
\State Set $w_1(i,j) \gets w(i,j)$ for alledges $(i,j) \in E[V \setminus \OldSet]$
\ForAll{$(i^*,j) \in E_1$}
\State $w_1(i^*,j) \gets \sum_{i\in \OldSet \colon (i,j)\in E} w(i,j)$
\EndFor
\State Solve \dksplus in $G_1$ to obtain a set of vertices $U'$ of size $k+1$
\If{$i^* \in U'$}
\State $U' \gets U'\setminus \{i^*\}$
\Else
\State $j \gets \argmin_{j\in U'} \text{weighted degree of } j \text{ in } G[U \cup (U' \setminus \{i^*\})]$
\State $U' \gets U' \setminus \{j\}$
\EndIf
\State \Return{$U'$}
\end{algorithmic}
\end{algorithm}

\subsection{Proof of Lemma~\ref{lemma:k-densify-dks}}
\label{appendix:ds-proofs:kdensifydks}

We use the reduction from \kdense to \dksplus that we described in
Section~\ref{sec:relationships-densest}.  We present the pseudocode of the
reduction in Algorithm~\ref{algo:dense-black}.

We start by defining notation.
Let $G=(V,E,w)$ be the input graph for \kdense.
We let $G_1 = (V_1, E_1, w_1)$ be the graph defined in the reduction.
Formally, we set
$V_1 = (V \setminus U) \cup \{i^*\}$, where $i^*$ is the special vertex we newly
added to the graph, which represents the contracted vertices in $U$, and
$E_1 = E[V \setminus U] \cup \{(i^*, j)\colon (i,j) \in E \mbox{ for } i\in U, j \in V \setminus U\}$. 
We set $w_1(i^*,j) = \sum_{i \in U \colon j \in V \setminus U, (i,j)\in E} w(i,j)$. 

Now let $C_2^*$ be the optimal set of vertices selected for \dksplus
on $G_1$.  Let $C_2$ be the set of vertices selected for 
\dksplus on $G_1$ by an
$\alpha_{k+1}$-approximation algorithm.  Notice that $|C_2^*| = |C_2| = k+1$. 
Furthermore, we let $C_1^*$ denote the optimal set of vertices selected for
\kdense in $G$.

By definition $\frac{{w(E_1[C_2])}}{k+1} \geq \alpha_{k+1} \frac{{w(E_1[C_2^*])}}{k+1} \geq \alpha_{k+1} \frac{{w(E_1[C_1^* \cup \{i^*\}])}}{k+1}$.
The first inequality follows from the $\alpha_{k+1}$-approximation algorithm,
the second inequality follows from
${w(E_1[C_2^*])} \geq {w(E_1[C_1^* \cup \{i^*\}])}$
since $C_2^*$ is the optimal solution for \dksplus in $G_1$ and since
$C_1^*\cup\{i^*\}$ is a feasible solution for \dksplus in $G_1$.
It follows that ${w(E_1[C_2])} \geq \alpha_{k+1} {w(E_1[C_1^* \cup \{i^*\}])}$. 

Notice that if $i^* \in C_2$, then $\abs{C_2 \setminus \{i^*\}} = k$ and $C_2 \setminus \{i^*\}$ is a feasible solution for \kdense in $G$. 
If $i^* \notin C_2$, we have to pick a node $v \in C_2$ so that $C_2 \setminus \{v\}$ becomes a feasible solution for \kdense.
To obtain our approximation ratios, we we make the following case distinctions.

\emph{Case~1:} $i^* \in C_2$. The induced edge weights in the graph $G$ after adding
$C_2\setminus \{i^*\}$ to $U$ are given by ${w(E[U \cup (C_2 \setminus \{i^*\})])}$.
In addition, for the sum of edge weights induced by $C_2$ on $G_1$ we have that
$|E_1[C_2]|= {w(E[C_2 \setminus \{i^*\}])} + {w(E[U, C_2 \setminus \{i^*\}])}$. 
We then derive the following bound, 

\begin{equation}
    \begin{aligned}
        {w(E[U \cup (C_2 \setminus \{i^*\})] )}&= {w(E[U])} + {w(E[C_2 \setminus \{i^*\}])} + {w(E[U, C_2 \setminus \{i^*\}])} \\
		&= {w(E[U])} + | E_1[C_2]| \\
		&\geq {w(E[U])} + \alpha_{k+1}|E_1[C_1^* \cup \{i^*\}]| \\
		&\geq \alpha_{k+1} ({w(E[U])} + |E_1[C_1^* \cup \{i^*\}]|) \\
		&= \alpha_{k+1} {w(E[U \cup C_1^*])}. 
    \end{aligned}
\end{equation}

\emph{Case~2:} $i^* \notin C_2$. 
Let $a$ be a node in $C_2$ with smallest weighted degree in $G_1[C_2]$, i.e., $a = \arg\min_{u \in C_2} \sum_{v \in C_2}w_1(u, v)$.
The weight of induced edges after adding $C_2 \setminus \{a\}$ to $U$, i.e.,
${w(E[U \cup (C_2 \setminus \{a\})])}$, is then at least
${w(E[U])} + \frac{k-1}{k+1}{w(E_1[C_2 \cup \{i^*\})}$. 
Now we obtain that:
\begin{equation}
    \begin{aligned}
		{w(E[U \cup (C_2 \setminus \{a\})])}&= {w(E[U])} + {w(E[C_2 \setminus \{a\}])} + {w(E[U, C_2 \setminus \{a\}])} \\
		&= {w(E[U])} + |E_1[(C_2 \setminus \{a\}) \cup \{i^*\}]| \\
        &\geq {w(E[U])} + {w(E_1[C_2 \setminus \{a\}])} \\
		&\geq {w(E[U])} + \frac{k-1}{k+1}{w(E_1[C_2])} \\
		&\geq  {w(E[U])} + \frac{k-1}{k+1} \alpha_{k+1}{w(E_1[C_2^*])} \\
        &= {w(E[U])} + \frac{k-1}{k+1} \alpha_{k+1}{w(E_1[C_1^* \cup \{i^*\}])} \\
		&\geq \frac{k-1}{k+1} \alpha_{k+1} ({w(E[U])} + {w(E_1[C_1^* \cup \{i^*\}])}).\\
		&= \frac{k-1}{k+1} \alpha_{k+1} {w(E[U \cup C_1^*])}.
    \end{aligned}
\end{equation}
As ${w(E[U\cup C_1^*])}$ is the optimal objective function value for \kdense, we
have finished the proof. 

\subsection{Proof of Theorem~\ref{thm:densest-subgraph-reduction}}
\label{proof:densest-subgraph-reduction}

The theorem follows directly from combining Lemmas~\ref{lemma:k-densify-dks}
and~\ref{lemma:k-densify-local-changes-connection}.

\clearpage
\section{Omitted content of Section~\ref{sec:max-cut}}
\label{appendix:mc-proofs}
\subsection{Proof of Lemma~\ref{lem:max-cut-generalization}} 
Let $G=(V,E,w)$ and $k$ be an input for \ccmaxcut. The goal is to find $C \subseteq V$ such that $\cut{C}$ is maximized. 
Suppose we have an algorithm~$\algo$ for \maxcutkc. We apply~$\algo$ with input $(G = (V, E, w), U=\emptyset, k)$. 
Now~$\algo$ returns a set~$\SelectSet$ of size $|\SelectSet|=k$ to maximize $\cut{U \symm C}$.
Given that $U=\emptyset$, $\cut{U \symm C} = \cut{C}$.

As the constraint and the objective function for these two problems are the same, if~$\algo$ is $\alpha_{k}$-approximation algorithm for \maxcutkc, then it is $\alpha_{k}$-approximation for \ccmaxcut. 
As we directly use $C$ as the solution for \ccmaxcut, the running time is $\bigO(T(m,n, k))$. 

\subsection{Proof of Corollary~\ref{cor:max-cut-hardness}}
This is a direct result from Lemma~\ref{lem:max-cut-generalization} and the
hardness results of \cite{DBLP:conf/approx/AustrinS19}.

\subsection{Pseudocode for the black-box-based algorithm for \maxcutkc from Section~\ref{sec:max-cut:black-box}}
\label{appendix:maxcutkc:blackbox:pseudocode}

The pseudocode for black-box-based algorithm for \maxcutkc, assuming access to a solver to \maxcut. As we discussed in Section~\ref{sec:max-cut:black-box}, is illustrated in Algorithm~\ref{algo:maxcutkc-blackbox}.

\begin{algorithm}[H]
\caption{An algorithm for \maxcutkc, assuming access to a black-box \maxcut solver}
\label{algo:maxcutkc-blackbox}
\begin{algorithmic}[1]
\Require{$(G(V, E, w), \OldSet, k)$, access to a \maxcut solver}
\State Let $\SelectSet \leftarrow \{i \colon x_i \neq x^0_i\}$ be the solution
	of applying our black-box \maxcut solver on $G$
\State \textbf{Fixing $\SelectSet$:} Set $\FixSelectSet \gets \SelectSet$ 
\While{$\abs{\FixSelectSet} \neq k$}
    \If{$\abs{\FixSelectSet} < k$}
        \State $u^* \leftarrow \arg\max_{u \in V \setminus \FixSelectSet} \cut{\OldSet \symm (\FixSelectSet \cup \{u\})}$
        \State $\FixSelectSet \leftarrow \FixSelectSet \cup \{u^*\}$
    \ElsIf{$\abs{\FixSelectSet} > k$}
        \State $u^* \leftarrow \arg\max_{u \in \FixSelectSet} \cut{\OldSet \symm (\FixSelectSet \setminus \{u\})}$
        \State $\FixSelectSet \leftarrow \FixSelectSet \setminus \{u^*\}$
    \EndIf
\EndWhile
\State \Return{$\FixSelectSet$}
\end{algorithmic}
\end{algorithm}

\subsection{Proof of Theorem~\ref{thm:max-cut-black-box}}
\label{appendix:mc-proofs:blackbox}

{
Before we prove the theorem, we observe the following symmetry property:
	$\cut{\OldSet\symm (V\setminus \SelectSet)} = \cut{\OldSet\symm \SelectSet}$.
	We formally state it in Lemma~\ref{lem:max-cut-local-observation}.
This implies that in the approximation ratio analysis, we
only need to consider the case when $k \leq n/2$.
Concretely, if $k > n/2$, we simply run the algorithm with cardinality constraint equal to $n-k$.
Assume that we get $\SelectSet$ as the solution for \maxcutkc.
Then, we can return $V\setminus \SelectSet$ as the desired solution with the cardinality constraint satisfied. 
This solution gives us the same cut value as selecting $\SelectSet$ as solution. 
In Figure~\ref{fig:maxcutkr-example}, we illustrate the \maxcutkc problem on a
graph and its symmetry property. 
}

{
\begin{lemma}
\label{lem:max-cut-local-observation}
$\cut{(V \setminus \OldSet)\symm \SelectSet} = \cut{\OldSet\symm (V\setminus \SelectSet)} = \cut{\OldSet\symm \SelectSet}$ for any $\SelectSet\subseteq V$ and $\OldSet\subseteq V$.
\end{lemma}
}

\begin{proof}
{
For any set $A \subseteq V$, we denote $V \setminus A$ by $\overline{A}$.
Our goal is to prove that $\OldComSet \symm \SelectSet= \OldSet \symm \SelectComSet = \overline{\OldSet \symm \SelectSet}$, since this condition directly implies that $\cut{\OldSet \symm \SelectSet} = \cut{\OldSet \symm \SelectComSet} = \cut{\OldComSet \symm \SelectSet}$.
We expand the three formulations to show that all of them are equal. 
}

{
First, let us expand $\OldComSet \symm \SelectSet$: 
\begin{equation}
	\begin{aligned}
		\OldComSet \symm \SelectSet &= (\OldComSet \setminus \SelectSet) \cup (\SelectSet \setminus \OldComSet) \\
		&= (\OldComSet \cap \SelectComSet) \cup (\SelectSet \cap \OldSet) \\
		&= (\OldComSet \cup \SelectSet) \cap (\SelectComSet \cup \OldSet) \\
		&= (\OldSet \cup \SelectComSet) \cap (\SelectSet \cup \OldComSet)
	\end{aligned}
\end{equation}
}

{
Second, let we expand $\OldSet \symm \SelectComSet$: 
\begin{equation}
	\begin{aligned}
		\OldSet \symm \SelectComSet &= (\OldSet \setminus \SelectComSet) \cup (\SelectComSet \setminus \OldSet) \\
		&= (\OldSet \cap \SelectSet) \cup (\SelectComSet \cap \OldComSet) \\
		&= (\OldSet \cup \SelectComSet) \cap (\SelectSet \cup \OldComSet)
	\end{aligned}
\end{equation}
}

{
Finally, let us expand the right-hand side: 
\begin{equation}
	\begin{aligned}
		\overline{\OldSet \symm \SelectSet} &= \overline{(\OldSet \setminus \SelectSet) \cup (\SelectSet \setminus \OldSet)} \\
		&= \overline{\OldSet \setminus \SelectSet} \cap \overline{\SelectSet \setminus \OldSet} \\
		&= \overline{\OldSet \cap \SelectComSet} \cap \overline{\SelectSet \cap \OldComSet} \\
		&= (\OldComSet \cup \SelectSet) \cap (\SelectComSet \cup \OldSet) 
		&= (\OldSet \cup \SelectComSet) \cap (\SelectSet \cup \OldComSet)\\
	\end{aligned}
\end{equation}
}

{
We can see that $\OldComSet \symm \SelectSet = \OldSet \symm \SelectComSet = \overline{\OldSet \symm \SelectSet}$.
}
\end{proof}

Let $(G=(V, E, w), \OldSet, k = \tau n)$ be an instance of \maxcutkc, and let
$(\OldSet, \OldComSet)$ be the initial cut of the graph $G$.
{
Due to the discussion and the lemma above, 
we only consider the case when $k \leq n/2$, hence, $\tau \leq 1/2$ in the
following analysis.}

{
After running the \maxcut solver on $G$ (without the given cut), we obtain a cut $(\TempSet, \TempComSet)$. 
Set $\SelectSet$ to be one of the two sets $\TempSet \symm \OldSet$ and
$\TempSet \symm \OldComSet$ such that the cardinality of $\SelectSet$ is at most $\frac{n}{2}$.
}

Let $\mu$ be
such that $\abs{\SelectSet} = \mu n$. 
Now we make changes to $C$ until we have that $\abs{C}=k$. 
If $\abs{C}>k$ this is done by greedily removing the vertices from~$C$ which have the smallest
contribution to the cut; if $\abs{C}<k$, this is done by greedily adding the
vertices to~$C$ which increase the cut the most.
We present the pseudocode of our reduction in Algorithm~\ref{algo:maxcutkc-blackbox}.

For each node $a \in V$, let $\cutnode{a}{\OldSet}$ be the sum of edge weights of node $a$ in the cut
$(\OldSet, \OldComSet)$, i.e.,
$\cutnode{a}{\OldSet} = \sum_{(a, v)\in E[\OldSet, \OldComSet]} w(a, v)$. 

We will calculate how much the greedy algorithm reduces to the approximation ratio. To do so, we make case distinctions on $\abs{\SelectSet} > k$ or $\abs{\SelectSet} <k$, and we give the following Lemma~\ref{lem:cut-blackbox-deduction}.  

\begin{lemma}
	\label{lem:cut-blackbox-deduction}
	Let $\tau$ be such that $k = \tau n$ and assume a \maxcut solver returns the partition of the cut with $\abs{\SelectSet} = \mu n$, where $\mu \in (0, 1]$. 
	Suppose the initial cut has cut value $M_0$.
	After conducting the greedy procedure to meet the cardinality constraint, the sum of edge weights of the cut is lower bounded by $(1 - \Theta(\frac{1}{n}))\min \{\frac{\tau^2}{\mu^2}, \frac{(1-\tau)^2}{(1-\mu)^2}\} M_0$. 
\end{lemma}

\begin{proof}
	Let $M$ be the sum of edge weights in the cut at any stage of the greedy
	procedure. Note that initially $M$ is equal to $M_0$, and let the initial
	$\SelectSet$ be equal to $\SelectSet_0$. At the $i$th round of the greedy
	algorithm, $M$ becomes $M_i$, and $\SelectSet$ becomes $\SelectSet_i$. 
	
\emph{Case~1:} $\abs{\SelectSet} > \tau n$.
	We notice that $\sum_{u\in \SelectSet} \cutnode{u}{\OldSet \symm \SelectSet} \leq  2M$. 
Since the algorithm always selects a node in $C$ to maximize the cut for the
next round, i.e., $u^* = \arg \max_{u \in \SelectSet} \cut{\OldSet \symm (\SelectSet \setminus \{u\})}$,
the cut value drops at most $\frac{2M}{\abs{\SelectSet}}$.

It follows from the previous analysis that $M_i \geq M_{i-1} - \frac{2M_{i-1}}{\abs{\SelectSet_{i-1}}}$. 
Thus at the $t$-th step, 
\begin{displaymath}
	M_t \geq M_0 \prod_{i=1}^t \left(1 - \frac{2}{\abs{C_{i-1}}}\right)
	= M_0 \prod_{i=1}^t \left(\frac{\abs{C_{i-1}} - 2}{\abs{C_{i-1}}}\right).
\end{displaymath}

Note that $\abs{\SelectSet_{i -2}} = \abs{\SelectSet_{i}} + 2$, and in the end
$|\SelectSet_t| = \tau n$. Therefore, we can cancel out most of the terms and
obtain:
\begin{displaymath}
	\begin{aligned}
		M_{t} &\geq \frac{(\abs{\SelectSet_{t-2}}-2)(\abs{\SelectSet_{t-1}}-2)}{\abs{\SelectSet_0}\abs{\SelectSet_1}}M_{0} \\
		&= \frac{\tau n(\tau n - 1)}{\mu n (\mu n - 1)}M_{0} \\
		&\geq \frac{\tau^2 - \frac{\tau}{n}}{\mu^2} M_0 \\
		&= \frac{\tau^2}{\mu^2}\left(1 - \frac{1}{n\tau}\right) M_0.
	\end{aligned}
\end{displaymath}

\emph{Case~2:} $\abs{\SelectSet} < \tau n$.
Similarly, we notice that $\sum_{u\in \vertices \setminus \SelectSet} \cutnode{u}{\OldSet \symm \SelectSet} \leq  2M$. 
As the algorithm always selects a node in $\vertices \setminus \SelectSet$ to
maximize the cut for the next round, i.e., $u^* = \arg \max_{u \in \vertices \setminus \SelectSet} \cut{\OldSet \symm (\SelectSet \cup \{u\})}$, the cut
value drops at most $\frac{2M}{n - \abs{\SelectSet}}$. 

It follows from the previous analysis that $M_i \geq M_{i-1} - \frac{2M_{i-1}}{n-\abs{\SelectSet_{i-1}}}$. Thus at the $t$-th step, 
\begin{displaymath}
	M_t \geq M_0 \prod_{i=1}^t \left(1 - \frac{2}{n-\abs{C_{i-1}}}\right). 
\end{displaymath}

Note that $\abs{\SelectSet_{i-2}} = \abs{\SelectSet_{i}} - 2$, and in the end $\abs{\SelectSet_t} = \tau n$. 
Thus,
\begin{displaymath}
	\begin{aligned}
		M_{t} &\geq \frac{(n - \abs{\SelectSet_{t-2}}-2)(n - \abs{\SelectSet_{t-1}}-2)}{(n - \abs{\SelectSet_{0}})(n - \abs{\SelectSet_{1}})} M_{0} \\
		&= \frac{(n - \abs{\SelectSet_{t}})(n - \abs{\SelectSet_{t}}-1)}{(n - \abs{\SelectSet_{0}})(n - \abs{\SelectSet_{0}}+1)} M_{0} \\
		&= \frac{(n - \tau n)(n - \tau n - 1)}{(n - \mu n) (n - \mu n + 1)} M_{0}\\
		&= \frac{(1- \tau )(1 - \tau  - 1/n)}{(1 - \mu ) (1 - \mu  + 1/n)} M_{0}\\
		&= \frac{(1- \tau )^2}{(1 - \mu )^2}\frac{(1- \mu )(1 - \tau  - 1/n)}{(1 - \tau) (1 - \mu  + 1/n)} M_{0} \\
  		&= \frac{(1- \tau )^2}{(1 - \mu )^2}\frac{(1- \mu )\cdot n - \frac{(1-\mu)}{1-\tau})}{(1 - \mu) n + 1} M_{0} \\
        &\geq \frac{(1- \tau )^2}{(1 - \mu )^2} \left(1 - \frac{\frac{2 - \tau - \mu}{1 - \tau}}{(1 - \mu) n + 1}\right) M_0.
	\end{aligned}
\end{displaymath}

Combining the bounds of these two cases, we obtain that
$M_{t} \geq \min \left\{\frac{\tau^2}{\mu^2}(1 - \frac{1}{n\tau}),
		\frac{(1- \tau )^2}{(1 - \mu)^2} \left(1 - \frac{\frac{2 - \tau - \mu}{1 - \tau}}{(1 - \mu) n + 1}\right)\right\}M_0$.
Both $\frac{1}{n \tau}$ and $\frac{\frac{2 - \tau - \mu}{1 - \tau}}{(1 - \mu)n + 1}$ are in $\Theta(\frac{1}{n})$. 
We thus conclude $M_{t} \geq (1 - \Theta(\frac{1}{n}))\min \left\{\frac{\tau^2}{\mu^2}, \frac{(1-\tau)^2}{(1 - \mu)^2}\right\} M_0$. 
\end{proof}

Lemma~\ref{lem:cut-blackbox-deduction} implies that, if the \maxcut solver obtains any constant approximation ratio result, and that $\tau$ is a constant (i.e., $k = \Omega(n)$). We obtain a constant approximation algorithm for \maxcutkc. 
Specifically, let the approximation ratio for \maxcut be $\alpha$ (for instance,
$\alpha > 0.87$ for the algorithm of \citet{goemans1994approximation}),
the approximation for \maxcutkc becomes at least $(1 - \Theta(\frac{1}{n}))\min\{4\tau^2, (1-\tau)^2\} \alpha$ in the worst
case, where {we use that $\mu^2 \leq \frac{1}{4}$} and $(1-\mu)^2 \leq 1$. 
{The inequality $\mu^2 \leq \frac{1}{4}$ comes from how we pick $\SelectSet$: we always pick $\SelectSet$ with the smaller cardinality.}
The approximation
ratio in the theorem follows since $k = \tau n$.
The time complexity comes from solving the \maxcut problem using a blackbox solver and constructing and maintaining a priority queue for the greedy procedure.

\clearpage
\section{Omitted content of Section~\ref{sec:general-framework}}
\label{appendix:general}
{
In this section, we provide the content omitted from Section~\ref{sec:general-framework}. The section is structured as follows: We begin by outlining the general \sdp-based algorithm framework in Sections~\ref{appendix:general-tools}, \ref{sec:hyperplane}, and \ref{sec:fix-c}. 
Following this, we give a detailed analysis for \kdense and \maxcutkc. Specifically, the pseudocode for \dskc is presented in Section~\ref{appendix:densest-subgraph:sdp:pseudocode}, and the pseudocode for \maxcutkc is in Section~\ref{appendix:maxcutkc:sdp:pseudocode}. 
We then present the proof of Theorem~\ref{thm:densest-subgraph-sdp}, the main theorem for \dskc and \kdense, in Section~\ref{sec:proof:thm-densest-subgraph-overview}, and the proof of Theorem~\ref{thm:max-cut-sdp}, the main theorem for \maxcutkc, in Section~\ref{appendix:mc-proofs:sdp}.
}

\subsection{Useful tools for the analysis}
\label{appendix:general-tools}
Let us first introduce additional notation and lemmas
that we will repeatedly use in this section.  We define the constants
\begin{equation}
	\label{eq:const-alpha}
	\alpha = \min_{0 \leq \theta \leq \pi} \frac{2}{\pi} \frac{\theta}{1 - \cos \theta} > 0.87856,
\end{equation}
and
\begin{equation}
	\label{eq:const-beta}
	\beta = \min_{0 \leq \theta < \arccos(-1/3)} \frac{2}{\pi} \frac{2\pi - 3 \theta}{1 + 3 \cos \theta} > 0.79607.
\end{equation}
We also introduce the binary function $\sign \colon \Real \rightarrow \{-1, 1\}$ such that $\sign(x) = 1$ if $x \geq 0$, 
and $\sign(x) = -1$ if $x <0$. Additionally, we will use the following lemmas.

\begin{lemma}[Lemma 2.1 in~\citet{goemans1994approximation}]
\label{lem:alpha}	
It holds for any $y \in [-1, 1]$ that $\arccos(y) \geq \alpha \pi \frac{1}{2}(1-y)$, where $\alpha$ is defined in Equation~\eqref{eq:const-alpha}. 
\end{lemma}

\begin{lemma}[Lemma 1.2 in~\citet{goemans1994approximation}]
	\label{lem:random-projection-1}
	Let $\vecv_i$ and $\vecv_j$ be vectors on the unit sphere $\mathcal{S}_n$,
	let $\vecr$ be a unit vector that is drawn uniformly randomly from $\mathcal{S}_n$,	then $\pr(\sign(\vecv_i \cdot \vecr) \neq \sign(\vecv_j \cdot \vecr)) = \frac{1}{\pi} \arccos(\vecv_i \cdot \vecv_j)$. 
\end{lemma}

\begin{lemma}[Lemma 2.3 in~\citet{goemans1994approximation}]
\label{lem:beta}
Let $\vecv_i$, $\vecv_j$ and $\vecv_k$ be three vectors on the unit sphere $\mathcal{S}_n$, 
let $y_1 = \vecv_i \cdot \vecv_j$, $y_2 = \vecv_j \cdot \vecv_k$ and $y_3 = \vecv_i \cdot \vecv_k$. 
It follows that $1 - \frac{1}{2\pi} (\arccos(y_1) + \arccos(y_2) + \arccos(y_3)) \geq \frac{\beta}{4} (1 + y_1 + y_2 + y_3)$, where $\beta$ is defined in Equation~\eqref{eq:const-beta}. 
\end{lemma}

\subsection{Analysis of the hyperplane rounding}
\label{sec:hyperplane}
{In this subsection, we analyze the second step of the \sdp-based algorithm for \gpkc.}
In particular, we derive bounds that hold for the \gpkc after hyperplane rounding. 
Our bounds are presented in Lemma~\ref{lem:cut-bound-z}.
Similar to
\citet{DBLP:journals/algorithmica/FriezeJ97,DBLP:journals/jal/feigel01,han2002improved},
we introduce the random variable~$Z$ that controls the number of nodes that change the position and the objective value. 
Then we derive a lower bound on $Z$ with high probability using Markov's inequality. 
We present the proof at the end of this section. 

In our proof, we use the following notation.
We let $\overline{\vecx} \in \{-1, 1\}^{n}$ be the indicator vector
obtained from the hyperplane rounding;
notice that $\overline{\vecx}$ does not necessarily obey the cardinality constraint. 
We set $\overline{\SdpRnd}$ to the objective value of \gpkc induced by $\overline{\vecx}$.
Let $\OPT$ be the optimal value of \gpkc.
We let $\SdpRatio = \frac{\overline{\SdpRnd}}{\OPT}$ and $\SdpRatioAvg = \frac{\Exp{\overline{\SdpRnd}}}{\OPT}$.
Notice that $\SdpRatio$ is a random variable and $\Exp{\SdpRatio} = \SdpRatioAvg$.
We let $\tau$ be such that $k = \tau n$ and notice that $\tau=\Omega(1)$ since
$k=\Omega(n)$.
Moreover, we let $\SelectSet=\{i \colon \overline{x}_i \neq x^0_i\}$ and $\TempSet = \{i \colon \overline{x}_i = 1\}$.

\begin{lemma}
\label{lem:cut-bound-z}
Let $\gamma$ and $\eta$ be two constants such that $\gamma \in [0.1, 5]$ and $\eta \in [\frac{1-\alpha + 0.01}{\alpha}, +\infty)$.
Let $Z$ be the random variable such that 
$$Z = \SdpRatio + \gamma \eta \frac{n-\abs{\SelectSet}}{n-k} + \gamma \frac{\abs{\SelectSet}(2k - \abs{\SelectSet})}{n^2}.$$
	After repeating the hyperplane rounding at least $N = \frac{1 - p + \epsilon p}{\epsilon p} \log(\frac{1}{\epsilon})$ times, 
 where $p$ is a constant lower bounded by $0.0005$, 
 it holds that with probability at least $1 - \epsilon$, 
 $Z \geq (1-\epsilon) [\SdpRatioAvg + \gamma \eta \alpha + \gamma (\alpha (1 - \tau)^2 - 1 + 2\tau)]$.
\end{lemma}

To prove Lemma~\ref{lem:cut-bound-z}, we must introduce several bounds on $\SdpRatioAvg$ and the expected value of $\abs{\SelectSet}$. 
{
We begin with the bound on $\SdpRatioAvg$. 
For Lemma~\ref{lem:cut-bound-z} to be valid, we rely on the fact that $\SdpRatioAvg > 0$ for the parameter settings in Table~\ref{tab:case_distinctions}. 
The values of $\SdpRatioAvg$ for \kdense and \maxcutkc are presented in Corollaries~\ref{lem:dense-bound-size-3} and~\ref{lem:cut-bound-obj}, respectively. 
For the $\SdpRatioAvg$ values related to \vckc and \maxuncutkc, we refer readers to \citep{DBLP:journals/jal/feigel01}.
}

\begin{lemma}
\label{lem:boundobj}
Let $\tilde{\SdpRatioAvg}$ be the maximum value such that the inequality
\begin{equation*}
\begin{aligned}
     &(\parzero + \parone + \partwo + \parthree) - \nicefrac{2}{\pi} \cdot (\parone \arccos(\vecv_0 \cdot \vecv_i) + \partwo \arccos(\vecv_0 \cdot \vecv_j) + \parthree \arccos(\vecv_i \cdot \vecv_j)) \\
     &\geq \tilde{\SdpRatioAvg} \cdot (\parzero + \parone \vecv_0 \cdot \vecv_i + \partwo \vecv_0 \cdot \vecv_j + \parthree \vecv_i \cdot \vecv_j)
\end{aligned}
\end{equation*}
holds for every $\vecv_0, \vecv_i, \vecv_j \in \mathcal{S}_{n}$. It holds that $\SdpRatioAvg \geq \tilde{\SdpRatioAvg}$.
\end{lemma}

\begin{proof}
First, we observe that the condition implies that
\begin{equation*}
\begin{aligned}
     &\sum_{i<j} w(i,j) \left( (\parzero + \parone + \partwo + \parthree) - \nicefrac{2}{\pi} \cdot (\parone \arccos(\vecv_0 \cdot \vecv_i) + \partwo \arccos(\vecv_0 \cdot \vecv_j) + \parthree \arccos(\vecv_i \cdot \vecv_j)) \right) \\
     &\geq \tilde{\SdpRatioAvg} \cdot \sum_{i<j} w(i,j) (\parzero + \parone \vecv_0 \cdot \vecv_i + \partwo \vecv_0 \cdot \vecv_j + \parthree \vecv_i \cdot \vecv_j).
\end{aligned}
\end{equation*}

Next, we show that 
\begin{align*}
    &\Exp{(\parzero + \parone \overline{x}_i + \partwo \overline{x}_j + \parthree \overline{x}_i \overline{x}_j)}\\
    &= (\parzero + \parone + \partwo + \parthree) - \nicefrac{2}{\pi} \cdot (\parone \arccos(\vecv_0 \cdot \vecv_i) + \partwo \arccos(\vecv_0 \cdot \vecv_j) + \parthree \arccos(\vecv_i \cdot \vecv_j))
\end{align*}
by using Lemma~\ref{lem:random-projection-1} as follows, 
\begin{equation*}
    \begin{aligned}
        &\Exp{(\parzero + \parone \overline{x}_i + \partwo \overline{x}_j + \parthree \overline{x}_i \overline{x}_j)} \\
        &=  \parzero + \parone \Exp{\sign((\vecr \cdot \vecv_0)(\vecr \cdot \vecv_i))} + \partwo \Exp{\sign((\vecr \cdot \vecv_0)(\vecr \cdot \vecv_j))} + \parthree \Exp{\sign((\vecr \cdot \vecv_i)(\vecr \cdot \vecv_j))} \\
        &= \parzero + \parone (1 - \nicefrac{2}{\pi}\cdot \arccos{(\vecv_0 \cdot \vecv_i)}) + \partwo (1 - \nicefrac{2}{\pi}\cdot \arccos{(\vecv_0 \cdot \vecv_j)}) + \parthree (1 - \nicefrac{2}{\pi}\cdot \arccos{(\vecv_i \cdot \vecv_j)})\\
        &= c_0 + c_1 + c_2 + c_3 - \nicefrac{2}{\pi}\cdot (\parone \arccos{(\vecv_0 \cdot \vecv_i)} + \partwo \arccos{(\vecv_0 \cdot \vecv_j)} + \parthree \arccos{(\vecv_i \cdot \vecv_j)}).
    \end{aligned}
\end{equation*}

Next, notice that 
$\OPT \leq \sum_{i<j} w(i,j) (\parzero + \parone \vecv_0 \cdot \vecv_i + \partwo \vecv_0 \cdot \vecv_j + \parthree \vecv_i \cdot \vecv_j)$ 
and 
$\Exp{\overline{\SdpRnd}} =\sum_{i<j} w(i,j) \Exp{\parzero + \parone \overline{x}_i + \partwo \overline{x}_j + \parthree \overline{x}_i \overline{x}_j}$, and we get that
\begin{align*}
    \Exp{\overline{\SdpRnd}} \geq \tilde{\SdpRatioAvg} \OPT.
\end{align*}
By definition we get $\SdpRatioAvg \geq \tilde{\SdpRatioAvg}$, and we conclude the lemma. 

\end{proof}

\begin{corollary}
\label{lem:dense-bound-size-3}
For \kdense, it holds that $\SdpRatioAvg = \frac{\Exp{{w(E[\TempSet])}}}{\OPT} \geq \beta$. 
\end{corollary}

\begin{proof}
    First, we note that $\Exp{{w(E[\TempSet])}} = \Exp{\overline{\SdpRnd}}$ when the problem is \kdense.
    Next, recall that for \kdense we have $\parzero = \frac{1}{4}$, $\parone =
	\frac{1}{4}$, $\partwo = \frac{1}{4}$, $\parthree = \frac{1}{4}$ in
	Lemma~\ref{lem:boundobj}. Now the ratio $\beta$ is obtained by applying Lemma~\ref{lem:beta}.
\end{proof}

\begin{corollary}
	\label{lem:cut-bound-obj}
	For \maxcutkc, it holds that
 $\SdpRatioAvg = \frac{\Exp{\abs{\cut[\TempSet]}}}{\OPT} \geq \alpha$. 
\end{corollary}
\begin{proof}
    First, we note that $\Exp{\abs{\cut[\TempSet]}} = \Exp{\overline{\SdpRnd}}$ when the problem is \maxcutkc.
    Next, recall that for \maxcutkc we have $\parzero = \frac{1}{2}$, $\parone =
	0$, $\partwo = 0$, $\parthree = -\frac{1}{2}$ in Lemma~\ref{lem:boundobj}.
	Now the ratio $\alpha$ is obtained by applying Lemma~\ref{lem:alpha}.
\end{proof}

Next, we give the analysis on bounding $\Exp{\abs{\SelectSet}}$ and $\Exp{\abs{\SelectSet}(n - \abs{\SelectSet})}$ in Lemma~\ref{lem:dense-bound-size-1-short} and Lemma~\ref{lem:dense-bound-size-2-short}.
{Notice that for \maxcutkc, we can no assume that \(\abs{\SelectSet} \leq \frac{n}{2}\) in our analysis.
This is because the size of $\SelectSet$ is dependent on the input $k$.
In particular, Lemma~\ref{lem:dense-bound-size-1-short} does not hold by simply replacing $\SelectSet$ with $V\setminus \SelectSet$.}

\begin{restatable}{lemma}{denseboundsizeone}
\label{lem:dense-bound-size-1-short}
The expected size of $C$ can be bounded by $\Exp{\abs{\SelectSet}} \geq \alpha k$ and $\Exp{\abs{\SelectSet}} \leq n - n \alpha + \alpha k$.
\end{restatable}
\begin{proof}
We start by defining a sequence of Bernoulli random variables
$\{Y_i\}_{i = 1, \ldots, n}$ such that $Y_i = 1$ if $\overline{x}_i \neq x_i^0$
and, otherwise, $Y_i = 0$. 
By definition, $\pr(Y_i = 1) = \pr(\overline{x}_i \neq x_i^0)$. 
Since $x_i^0$ is a given constant that can either be $-1$ 
or $1$, we can write $\pr(Y_i = 1)$ as the summation of joint probability over $x_i^0=1$ and $x_i^0=-1$, that is  
\begin{equation*}
	\begin{aligned}
		\pr(Y_i = 1 ) &= \pr(\overline{x}_i\neq x_i^0, x_i^0 = -1) + \pr(\overline{x}_i\neq x_i^0, x_i^0 = 1). \\
	\end{aligned}
\end{equation*}

Recall that we set $\overline{x}_i = 1$ if $(\vecv_0 \cdot \vecr)(\vecv_i \cdot \vecr) \geq 0$ and 
$\overline{x}_i = -1$ otherwise. 
Thus, $\overline{x}_i = \sign(\vecv_0 \cdot \vecr)\sign(\vecv_i \cdot \vecr)$, 
and $x_i^0 \neq \overline{x}_i$ is equivalent to
$x_i^0 \neq \sign(\vecv_0 \cdot \vecr)\sign(\vecv_i \cdot \vecr)$. 
We make the case distinctions for $x_i^0 = -1$ and $x_i^0 = 1$. 

Case 1: $x_i^0 = -1$. Then $\sign(\vecv_i \cdot \vecr) = -\sign(\vecv_0 x_i^0 \cdot \vecr)$, and 
\begin{equation*}
	\displaystyle
	\begin{aligned}
		\pr(\overline{x}_i\neq x_i^0, x_i^0 = -1) 
		& = \pr(\sign(\vecv_i \cdot \vecr) \sign(\vecv_0 x_i^0 \cdot \vecr) = -1, x_i^0 = -1). \\
	\end{aligned}
\end{equation*}

Case 2: $x_i^0 = 1$. Then $\sign(\vecv_i \cdot \vecr) = \sign(\vecv_0 x_i^0 \cdot \vecr)$, and
\begin{equation*}
	\begin{aligned}
		\pr(\overline{x}_i\neq x_i^0, x_i^0 = 1) 
		& = \pr(\sign(\vecv_i \cdot \vecr) \sign(\vecv_0 x_i^0 \cdot \vecr) = -1, x_i^0 = 1). \\
	\end{aligned}
\end{equation*}

Adding up the above two formulas,
\begin{equation*}
	\begin{aligned}
		\pr(Y_i = 1 ) &= \pr(\sign(\vecv_i \cdot \vecr) \sign(\vecv_0 x_i^0 \cdot \vecr) = -1) 
		&\overset{(1)}{=}  \frac{1}{\pi} \arccos(x_i^0\vecv_0 \cdot \vecv_i),
	\end{aligned}
\end{equation*}
where $(1)$ holds according to Lemma~\ref{lem:random-projection-1}.

Note that $\abs{\SelectSet} = \sum_{i=1}^n Y_i$, and by the linearity of expectation, $\Exp{\abs{\SelectSet}} = \sum_{i=1}^n \Exp{Y_i}$. 
From now on, the analysis is the same as the one that appears in \citet[Lemmas~2.3 and 3.2]{DBLP:journals/jal/feigel01} to obtain lower and upper bounds on $\Exp{\abs{\SelectSet}}$. 
In detail, the lower bound on $\Exp{\abs{\SelectSet}}$ is obtained from
\begin{equation}
	\begin{aligned}
		\Exp{\abs{\SelectSet}} &\overset{(a)}{=} \frac{1}{\pi}\sum_{i=1}^n \arccos(x_i^0 \vecv_0 \cdot \vecv_i) &\overset{(b)}{\geq} \alpha \sum_{i=1}^n \frac{1 - x_i^0 \vecv_0 \cdot \vecv_i}{2} 
		\overset{(c)}{=} \alpha k. 
	\end{aligned}
\end{equation}

Similarly, the upper bound is obtained from 
$\Exp{\abs{\SelectSet}}$:
\begin{equation}
	\begin{aligned}
		\Exp{\abs{\SelectSet}} &\overset{(a)}{=} \sum_{i=1}^n \left(1 - \frac{1}{\pi}\arccos(-x_i^0 \vecv_0 \cdot \vecv_i)\right) \\
		&\overset{(b)}{\leq} n - \alpha \sum_{i=1}^n \frac{1 + x_i^0 \vecv_0 \cdot \vecv_i}{2} 
		\overset{(c)}{=} (1- \alpha) n + \alpha k. 
	\end{aligned}
\end{equation}

In the above formulas, $(a)$ follows from the definition of $\abs{\SelectSet}$, $(b)$ follows from Lemma~\ref{lem:alpha}, and $(c)$ follows from the constraint on the \sdp~\eqref{sdp:densest-subgraph}.

\end{proof}

\begin{restatable}{lemma}{denseboundsizetwo}
\label{lem:dense-bound-size-2-short}
 It holds that $\Exp{\abs{\SelectSet} (n - \abs{\SelectSet})} \geq \alpha (n - k)k$.
\end{restatable}

\begin{proof}
	Let $R_{ij}$ denote the Bernoulli random variable 
	indicating whether the relative locations of nodes~$i$ and~$j$ has changed, that is, 
	\begin{equation*}
	R_{ij} =
	\begin{cases}
	1, & \text{if } (\overline{x}_i \neq \overline{x}_j \text{ and } x_i^0 = x_j^0) \text{ or } (\overline{x}_i = \overline{x}_j \text{ and } x_i^0 \neq x_j^0), \\
	0, & \text{otherwise} .
	\end{cases}
	\end{equation*}
	By definition,  
	$\pr(R_{ij} = 1) = \pr(\overline{x}_i \neq \overline{x}_j, x_i^0 = x_j^0) + \pr(\overline{x}_i = \overline{x}_j, x_i^0 \neq x_j^0)$.

	Notice that $\overline{x}_i = \sign(\vecv_0 \cdot \vecr)\sign(\vecv_i \cdot \vecr)$.
	Now $\overline{x}_i \neq \overline{x}_j$ can be equivalently written as $\overline{x}_i \overline{x}_j = -1$ and, based on the statement before, is also equivalent to  
	$\sign(\vecv_i \cdot \vecr)\sign(\vecv_j \cdot \vecr)\sign(\vecv_0 \cdot \vecr)^2 = -1$. 
	Note that the term $\sign(\vecv_0 \cdot \vecr)^2$ is always equal to $1$ and
	can thus be dropped. 
	The first term of $\pr(R_{ij} = 1)$ can thus be formulated as 
	\begin{equation*}
		\begin{aligned}
			\pr(\overline{x}_i \neq \overline{x}_j, x_i^0x_j^0 = 1)
			&= \pr(\sign(\vecv_i \cdot \vecr) \sign(\vecv_j \cdot \vecr) = -1, x_i^0 x_j^0 = 1) \\
			&= \pr(\sign(x_i^0\vecv_i \cdot \vecr) \sign(x_j^0\vecv_j \cdot \vecr) = -1, x_i^0 x_j^0 = 1). \\
		\end{aligned}
	\end{equation*}

	Similarly, $\overline{x}_i \overline{x}_j = 1$ can be equivalently formulated as $\sign(\vecv_i \cdot \vecr)\sign(\vecv_j \cdot \vecr) = 1$. The second term of $\pr(R_{ij} = 1)$ can thus be formulated as
	\begin{equation*}
		\begin{aligned}
			\pr(\overline{x}_i = \overline{x}_j, x_i^0x_j^0 = -1) 
			&= \pr(\sign(\vecv_i \cdot \vecr) \sign(\vecv_j \cdot \vecr) = 1, x_i^0 x_j^0 = -1) \\
			&= \pr(\sign(x_i^0\vecv_i \cdot \vecr) \sign(x_j^0\vecv_j \cdot \vecr) = -1, x_i^0 x_j^0 = -1). \\
		\end{aligned}
	\end{equation*}

	After summing up the two terms and applying
	Lemma~\ref{lem:random-projection-1} in the last equality, we get that
	\begin{equation*}
		\begin{aligned}
			\pr(R_{ij} = 1) &= \pr(\overline{x}_i \neq \overline{x}_j, x_i^0x_j^0 = 1) + \pr(\overline{x}_i = \overline{x}_j, x_i^0x_j^0 = -1)\\
			&= \pr(\sign(\overline{x}_i\vecv_i \cdot \vecr) \sign(\overline{x}_j\vecv_j \cdot \vecr) = -1)  \\
			&= \frac{\arccos(x_i^0x_j^0 \vecv_i \cdot \vecv_j)}{\pi} .
		\end{aligned}
	\end{equation*}

	Recall that $\SelectSet = \{i \colon \overline{x}_i \neq x^0_i\}$ is the
	set of vertices that changed sides after the hyperplane rounding; thus,
	$\SelectSet$ is a random set and also $\abs{\SelectSet}$ is a random variable.
	Note that by definition, $\abs{\SelectSet}(n - \abs{\SelectSet}) = \sum_{i<j} R_{ij}$. Hence, by the linearity of expectation,  
	$\Exp{\abs{\SelectSet}(n - \abs{\SelectSet})} = \sum_{i<j} \Exp{R_{ij}}$.
	Now we obtain our lower bound as follows:

	\begin{equation}
		\begin{aligned}
			\Exp{\abs{\SelectSet}(n - \abs{\SelectSet})} 
			&\overset{(a)}{=} \sum_{i<j} \frac{\arccos(x_i^0x_j^0 \vecv_i \cdot \vecv_j)}{\pi} \\
			&\overset{(b)}{\geq} \alpha \sum_{i<j} \frac{1 - x_i^0x_j^0 \vecv_i \cdot \vecv_j}{2} \\
			&\overset{(c)}{\geq} \alpha(nk - k^2),
		\end{aligned}
	\end{equation}
 where holds by our derivation above, $(b)$ follows by Lemma~\ref{lem:alpha}, and $(c)$ holds by the constraints in the \sdp~\eqref{sdp:densest-subgraph}.
\end{proof}

\para{Proof of Lemma~\ref{lem:cut-bound-z}.} Now, we are ready to prove our main
result for the hyperplane rounding step of our algorithm.
\begin{proof}[Proof of Lemma~\ref{lem:cut-bound-z}]
We start the proof with a lower bound on $\Exp{Z}$.
By applying Lemma~\ref{lem:dense-bound-size-1-short}, Lemma~\ref{lem:dense-bound-size-2-short} and Lemma~\ref{lem:boundobj}, 
we get
$\Exp{Z} \geq \SdpRatioAvg + \gamma \eta \alpha + \gamma (\alpha (1 - \tau)^2 - 1 + 2\tau)$. 

Next, observe that $Z < 1 + \gamma \eta + \gamma$, as $\SdpRatioAvg < 1$, $\alpha <1$ and $\alpha (1 - \tau)^2 - 1 + 2\tau <1$. 
By Markov's inequality, we get $\pr(Z \leq (1-\epsilon)(\SdpRatioAvg + \gamma \eta \alpha + \gamma (\alpha (1 - \tau)^2 - 1 + 2\tau))) \leq \frac{1-p}{1-(1-\epsilon)p},$ for 
$p = \frac{\SdpRatioAvg + \gamma \eta \alpha + \gamma(\alpha(1-\tau)^2 - 1 + 2\tau)}{1 + \gamma \eta + \gamma}$.
More precisely, this inequality can obtained from
\begin{equation*}
	\begin{aligned}
	&\pr(Z \leq (1-\epsilon)(\SdpRatioAvg + \gamma \eta \alpha + \gamma(\alpha(1-\tau)^2 - 1 + 2\tau))) \\
	&\leq \frac{1 + \gamma \eta + \gamma - \Exp{Z}}{1 + \gamma \eta + \gamma - (1-\epsilon)(\SdpRatioAvg + \gamma \eta \alpha + \gamma(\alpha(1-\tau)^2 - 1 + 2\tau))} \\
	&\leq \frac{1 + \gamma \eta + \gamma - [\SdpRatioAvg + \gamma \eta \alpha + \gamma(\alpha(1-\tau)^2 - 1 + 2\tau)]}{1 + \gamma \eta + \gamma - (1-\epsilon)(\SdpRatioAvg + \gamma \eta \alpha + \gamma(\alpha(1-\tau)^2 - 1 + 2\tau))} \\
	&\leq \frac{1 - p}{1 - (1 - \epsilon)p}.
	\end{aligned}
\end{equation*}

Now observe that $p < 1$. Additionally, we derive a lower bound on $p$, and we
notice that $p>0.0005$ since
\begin{equation*}
	\begin{aligned}
p &= \frac{\SdpRatioAvg + \gamma \eta \alpha + \gamma(\alpha(1-\tau)^2 - 1 + 2\tau)}{1 + \gamma \eta + \gamma} \\
&\overset{(a)}{>} \frac{ \gamma \eta \alpha + \gamma(\alpha(1-\tau)^2 - 1 + 2\tau)}{1 + \gamma \eta + \gamma} \\
&\overset{(b)}{>} \frac{ \gamma \eta \alpha + \gamma(\alpha -1)}{1 + \gamma \eta + \gamma} \\
&\overset{(c)}{\geq} \frac{0.01\gamma}{1 + 1.01\gamma /\alpha} \\
&\overset{(d)}{>} 0.0005.
	\end{aligned}
\end{equation*}
Above, $(a)$ holds as $\SdpRatioAvg>0$, $(b)$ holds as $2\alpha \tau < 2\tau$, $(c)$ holds as the formula obtains the minimum value by setting $\eta = \frac{1 - \alpha + 0.01}{\alpha}$, $(d)$ holds as the denominator is at most $2$ and the numerator is at least $0.001$ by setting $\gamma = 0.1$.

As we repeat the hyperplane rounding $N$ times, and pick the largest value of $Z$ for fixed $\gamma$ and $\eta$, 
the probability that $Z \leq (1-\epsilon)(\beta + \gamma \eta \alpha + \gamma(\alpha(1-\tau)^2 - 1 + 2\tau))$ is bounded from above by 
\begin{equation*}
	\begin{aligned}
		\left[ \frac{1-p}{1-(1-\epsilon)p} \right]^N &\leq \left[1 - \frac{1}{1+ \frac{1-p}{\epsilon p}} \right]^{(1+\frac{1-p}{\epsilon p}) \frac{\epsilon p}{1- p + \epsilon p} N} \\
	     &\leq \exp\left[-\frac{\epsilon p}{1 - p + \epsilon p} N\right] .
	\end{aligned}
\end{equation*}

Thus if we choose $N \geq \frac{1 - p + \epsilon p}{\epsilon p} \log(\frac{1}{\epsilon})$,
 we can guarantee the above probability is upper bounded by $\epsilon$, 
 that is, we can guarantee $Z \geq (1 - \epsilon) (\SdpRatioAvg + \gamma \eta \alpha + \gamma (\alpha (1 - \tau)^2 - 1 + 2\tau))$ for a 
small positive value $\epsilon$, with probability at least $1 - \epsilon$.
\end{proof}

\subsection{Analysis of fixing $\SelectSet$}
\label{sec:fix-c}
{Next, we analyze the third step of the \sdp-based algorithm, which ensures that the cardinality constraint for $\SelectSet$ is met.}
Let the size of the selected set of nodes after the hyperplane rounding be
$\mu n$, i.e., we let $\mu$ be such that $\abs{\SelectSet} = \mu n$. 
Notice that $\mu$ is a random variable and $1 \geq \mu > 0$. 
Recall that $k\in \Omega(n)$ and $k = \tau n$, where $\tau$ is an input constant. 
We analyze the cost of fixing $\SelectSet$ to ensure that the cardinality
constraint is met and illustrate how the approximation ratios of \gpkc are
computed. 

Let us start from a direct implication from Lemma~\ref{lem:cut-bound-z}, where we set $\epsilon = \frac{1}{\Theta(n)}$, which is small enough relative to other variables. 
\begin{corollary}
    \label{cor:lambda-tau}
    Let $\OPT$ be the optimal value of the \gpkc problem. 
    After repeating the hyperplane rounding $N = \Theta(n \log n)$ times, with probability at least $1 - \frac{1}{\Theta(n)}$, there exists a solution $\overline{\SdpRnd}$ after the hyperplane rounding such that
    \begin{equation}\label{eq:lambda-tau}
        \frac{\overline{\SdpRnd}}{\OPT} \geq (\SdpRatioAvg + \gamma \eta \alpha + \gamma (\alpha (1 - \tau)^2 - 1 + 2\tau) - \gamma \eta (1-\mu)\frac{1}{1-\tau} - \gamma \mu (2\tau - \mu)) - \frac{1}{\Theta(n)}. 
    \end{equation}
\end{corollary}

Recall that $\SdpRatio = \frac{\overline{\SdpRnd}}{\OPT}$.
Based on Corollary~\ref{cor:lambda-tau}, we can write $\lambda$ as a function of
$\mu$, $\gamma$ and $\eta$ with given $\tau$, and we denote this function by
$\SdpRatio_{\tau}(\mu, \gamma, \eta)$. Then we obtain that
\begin{equation}
	\begin{aligned}
	\frac{\overline{\SdpRnd}}{\OPT}&=\SdpRatio_{\tau}(\mu, \gamma, \eta). 
	\end{aligned}
\end{equation}

{
Let $\SdpRnd$ be the objective value of \gpkc after fixing $\SelectSet$ through the greedy procedure. 
Next, we outline a general \sdp-based algorithm framework for analyzing the approximation ratio of \gpkc.
}
For all the \gpkc problems, we note that the ratio between $\SdpRnd$ and $\overline{\SdpRnd}$ is a function of 
{$\tau = \frac{k}{n}$ and $\mu = \frac{\abs{\SelectSet}}{n}$. 
}
Let us introduce a function $\kappa_{\tau}(\mu)$ such that 
\begin{equation}
    \frac{\SdpRnd}{\overline{\SdpRnd}} = \kappa_{\tau}(\mu).
\end{equation}

Then the approximation ratio of \gpkc is given by
\begin{equation*}
     \frac{\SdpRnd}{\OPT} = \frac{\overline{\SdpRnd}}{\OPT} \cdot \frac{\SdpRnd}{\overline{\SdpRnd}} = \kappa_{\tau}(\mu) \SdpRatio_{\tau}(\mu, \gamma, \eta),
\end{equation*}
and the approximation ratio depends on the specific settings of $\gamma$ and $\eta$.  
For the worst case analysis, we can bound it by selecting the most appropriate $\gamma$ and $\eta$ by solving 
\begin{equation}
	\max_{\gamma \in [0.1, 5], \eta \in [\frac{1-\alpha + 0.01}{\alpha}, +\infty)} \min_{\mu \in [0 ,1]} \kappa_{\tau}(\mu) \SdpRatio_{\tau}(\mu, \gamma, \eta).
\end{equation}

\subsection{Pseudocode for {\sdp}-based algorithm for \kdense and \dskc (in Section~\ref{sec:general:sdp})}
\label{appendix:densest-subgraph:sdp:pseudocode}

The pseudocode for the {\sdp}-based algorithm for \kdense and \dskc, as discussed in Section~\ref{sec:general:sdp}, 
is illustrated in Algorithm~\ref{algo:sdp-k-dense}.

\begin{algorithm}[H]
\caption{An \sdp-based algorithm for \kdense and \dskc}
\label{algo:sdp-k-dense}
\begin{algorithmic}[1]
\Require{$(G(V, E, w), \OldSet, k)$, number of iterations $N$}
\State \textbf{Step 1: Solve the SDP.} Solve the \sdp~\eqref{sdp:densest-subgraph} to obtain $\vecv_0, \ldots, \vecv_n$
\For{$i = 1$ to $N$}
\State \textbf{Step 2: Hyperplane Rounding.} Obtain the indicator vector $\overline{\vecx}$
\State Sample a unit vector $\vecr$ uniformly randomly from $\mathcal{S}_n$
\For{$i = 1$ to $n$}
\If{$(\vecr \cdot \vecv_0)(\vecr \cdot \vecv_i) \geq 0$}
\State $\overline{x}_i \gets 1$
\Else
\State $\overline{x}_i \gets -1$
\EndIf
\EndFor
\State Set $\SelectSet \gets \{i \colon \overline{x}_i \neq x^0_i\}$ and
$\TempSet \gets \{i \colon \overline{x}_i = 1\}$
\State \textbf{Step 3: Fixing $\SelectSet$.}  Set $\FixSelectSet \gets \SelectSet$
\If{$\abs{\FixSelectSet} < k$}
\State Add $k - \abs{\FixSelectSet}$ arbitrary vertices from $V \setminus (\FixSelectSet \cup \TempSet)$ to $\FixSelectSet$
\ElsIf{$\abs{\FixSelectSet} > k$}
\If{$\abs{\FixSelectSet \setminus \TempSet} \geq \abs{\FixSelectSet} - k$}
\State Remove $\abs{\FixSelectSet} - k$ arbitrary vertices in $\FixSelectSet \setminus \TempSet$ from $\FixSelectSet$
\Else
\State Remove all vertices in $\FixSelectSet \setminus \TempSet$ from $C$
\State Greedily choose $\abs{\FixSelectSet \cap \TempSet} - k$ minimum weighted
	degree nodes in subgraph $G[\FixSelectSet \cap \TempSet]$ and them remove from $\FixSelectSet$
\EndIf
\EndIf
\EndFor
\State \Return{$\FixSelectSet$}
\end{algorithmic}
\end{algorithm}

\subsection{Proof of Theorem~\ref{thm:densest-subgraph-sdp}}
\label{sec:proof:thm-densest-subgraph-overview}

The pseudocode of our algorithm can be found in
Algorithm~\ref{algo:sdp-k-dense}. In the present section, we demonstrate that
Algorithm~\ref{algo:sdp-k-dense} provides a constant factor approximation to the \kdense problem. As implied by Lemma~\ref{lemma:k-densify-local-changes-connection}, this solution also offers a constant approximation for the \dskc problem under the theorem's specified condition.

\subsubsection{Fixing $\SelectSet$}
\label{sec:proof:subgraph-sdp:fix-c}
Notice that at the end of a round of hyperplane rounding, we get $\SelectSet = \{i \colon x_i \neq x_i^0\}$, and $\TempSet = \OldSet \symm \SelectSet = \{i \colon x_i = 1\}$. 
In addition, we use Corollary~\ref{lem:dense-bound-size-3} to bound $\Exp{{w(E[\TempSet])}} \geq \beta W_{\sdp}^*$. 
Since $\abs{\SelectSet}$ is not necessarily equal to $k$, the algorithm needs to alter $\SelectSet$. 
This is done by considering a new variable $\FixSelectSet=\SelectSet$ and
making changes to $\FixSelectSet$ until $\abs{\FixSelectSet} = k$. 

To give a lower bound on ${w(E[\OldSet \symm \FixSelectSet])}$, we notice that only the operation that greedily removes nodes in $\FixSelectSet \cap \TempSet$ from $\FixSelectSet$ decreases ${w(E[\OldSet \symm \FixSelectSet])}$. In the worst case scenario, all $\abs{\SelectSet} - k$ nodes are removed.
Lemma~\ref{lem:dense-move-all} bounds the decrease of ${w(E[\TempSet])}$ after this greedy procedure. 

\begin{restatable}{lemma}{densemoveall}
	\label{lem:dense-move-all}
 For \kdense, if $\mu > \tau$, $\SdpRnd \geq \frac{\tau^2}{\mu^2}(1 - \frac{1}{\Theta(n)}) \overline{\SdpRnd}$, in other words, the sum of edge weights of the subgraph is lower bounded by ${w(E[\TempSet])}\frac{\tau^2}{\mu^2}(1 - \frac{1}{\Theta(n)})$. 
 Otherwise, $\SdpRnd \geq \overline{\SdpRnd}$. 
\end{restatable}

\begin{proof}
Notice that if $\tau \geq \mu$, to satisfy the cardinality constraint, $\TempSet$ adds more nodes, thus edges, and the sum of edge weights increases. In this case, $\SdpRnd \geq \overline{\SdpRnd}$.

When $\tau < \mu$, we let $\FixTempSet = \OldSet \symm \FixSelectSet$. We allow $\FixTempSet$ to change correspondingly with $\FixSelectSet$. Initially, $\FixTempSet = \TempSet$, and we will demonstrate that, in the end, ${w(E[\FixTempSet])} \geq {w(E[\TempSet])}\frac{k^2}{\abs{\SelectSet}^2}(1 - \frac{1}{k})$.

We analyze three operations that modify $\FixSelectSet$. We first note that the
operation which adds nodes from $\vertices \setminus (\FixSelectSet \cup \TempSet)$ to $\FixSelectSet$ only occurs when $\abs{\SelectSet} < k$. This operation increases ${w(E[\FixTempSet])}$.

For scenarios where $\abs{\SelectSet}>k$, two operations exist. The operation
that removes nodes from $\FixSelectSet \setminus \TempSet$ further adds nodes to $\FixTempSet$, thereby increasing ${w(E[\FixTempSet])}$.
It remains to consider the operation, which greedily removes nodes from
$\FixSelectSet \cap \TempSet$ in $\FixSelectSet$ and which decreases
${w(E[\FixTempSet])}$.

We now define $\FixTempSet_0$ as the initial $\FixTempSet$ of this greedy procedure. As the previous two operations contribute additional nodes to $\FixTempSet$, we conclude that ${w(E[\FixTempSet_{0}])} \geq {w(E[\TempSet])}$. Similarly, we assign $\FixSelectSet_0$ as the initial $\FixSelectSet$ of the greedy procedure.

Let $\FixTempSet_t$ represent $\FixTempSet$ and $\FixSelectSet_t$ represent $\FixSelectSet$ at the $t$-th step of the greedy procedure. We observe that at the $t$-th step, the decrease in the sum of edge weights is at most $\frac{2{w(E[\FixTempSet_t])}}{\abs{\FixSelectSet_t}}$, as suggested by the pigeonhole principle. In the worst-case scenario, the sum of edge weights decreases $\abs{\FixSelectSet_0} - k$ times, until $\abs{\FixSelectSet} = k$. The value of ${w(E[\FixTempSet])}$ at this point is lower bounded by:

\begin{displaymath}
\begin{split}
&{w(E[\FixTempSet_{0}])} \left(1 - \frac{2}{\abs{\FixSelectSet_0}}\right)\left(1 - \frac{2}{\abs{\FixSelectSet_1}}\right) \ldots \left(1 - \frac{2}{k+1}\right) \\
&= {w(E[\FixTempSet_{0}])} \frac{(\abs{\FixSelectSet_0} -2)}{\abs{\FixSelectSet_0}} \frac{(\abs{\FixSelectSet_1} -2)}{\abs{\FixSelectSet_1}}\ldots \frac{k-1}{k+1} \\
&= {w(E[\FixTempSet_{0}])} \frac{k(k-1)}{\abs{\FixSelectSet_0}(\abs{\FixSelectSet_0} -1)}.
\end{split}
\end{displaymath}

Given that $\abs{\FixSelectSet_0} \leq \abs{\SelectSet}$ and ${w(E[\FixTempSet_{0}])} \geq {w(E[\TempSet])}$, the above expression is at least ${w(E[\TempSet])}\frac{k(k-1)}{\abs{\SelectSet}(\abs{\SelectSet} -1)} \geq {w(E[\TempSet])}\frac{k(k-1)}{\abs{\SelectSet}^2}$.
Notice that we assume $k \in \Omega(n)$, and thus $1 - \frac{1}{k} = 1 -
\frac{1}{\Theta(n)}$. This concludes our proof. 
\end{proof}

Recall that our goal is to compute the approximation ratio, which can be formulated as 
\begin{equation}
	\max_{\gamma \in [0.1, 5], \eta \in \left[\frac{1-\alpha + 0.01}{\alpha}, +\infty\right)} \min_{\mu \in [0 ,1]} \kappa_{\tau}(\mu) \SdpRatio_{\tau}(\mu, \gamma, \eta).
\end{equation}
Let us define $f_{\tau}(\mu, \gamma, \eta) = \kappa_{\tau}(\mu) \SdpRatio_{\tau}(\mu, \gamma, \eta)$.
Notice that for \kdense, according to the above lemma, when $\mu \leq \tau$, $\kappa_{\tau}(\mu) = 1$.
Hence, for the worst-case analysis, we only need to discuss the scenario when $\mu \in (\tau, 1]$. Below, we give more details on computing the approximation ratio.
We ignore the factor $(1 - \frac{1}{\Theta(n)})$ in our computation as it only introduces a small multiplicative error in our result.

\subsubsection{Computing the approximation ratio}
\label{sec:proof:subgraph-sdp:fix-c-lowerbound-compute}
We first find the stationary point of $f_{\tau}(\mu, \gamma, \eta)$ w.r.t.\ $\mu$ by solving 
\begin{equation}
	\begin{aligned}
		\frac{\partial f_{\tau}(\mu, \gamma, \eta)}{\partial \mu} = 0.
	\end{aligned}
\end{equation}

Let this stationary point be $\mu_0$. We notice that $\mu_0$ can be written as a function of $\gamma$, $\eta$ and $\tau$. 
Since $\mu_0$ should be in the domain $[\tau, 1]$, we introduce the following step function $g_{\tau}(\gamma, \eta)$ to 
tackle the situation when $\mu_0$ is invalid (chosen outside of the domain):
\begin{equation}
	\begin{aligned}
		g_{\tau}(\gamma, \eta) =
		\begin{cases}
			f_{\tau}(\mu_0, \gamma, \eta) & \text{for } \tau \leq \mu_0 \leq 1, \\ 
			1 & \text{otherwise}.
		\end{cases}
	\end{aligned}
\end{equation}

For any given $\gamma$ and $\eta$, the minimization of the function
$f_{\tau}(\mu, \gamma, \eta)$ w.r.t.\ $\mu$ can be written as the minimum of three constant functions. 
The two boundary cases when $\mu = \tau$ or $\mu = 1$ are given by 
$f_{\tau}(\mu_0, \gamma, \eta)$ and $f_{\tau}(1, \gamma, \eta)$, whereas
$f_{\tau}(\mu_0, \gamma, \eta)$ represents the case when $\mu \in [\tau, 1]$. 
Thus, 
$\min_{\mu \in [\tau, 1]}f_{\tau}(\mu, \gamma, \eta) = \min \left\{f_{\tau}(\tau, \gamma, \eta), f_{\tau}(1, \gamma, \eta), g_{\tau}(\gamma, \eta)\right\}$. 
Now we obtain:
\begin{equation*}
	\begin{aligned}
f(\tau) &= \max_{\gamma, \eta} \min_{\mu \in [\tau, 1]}f_{\tau}(\mu, \gamma, \eta)  \\
&= \max_{\gamma, \eta} \min \left\{f_{\tau}(\tau, \gamma, \eta), f_{\tau}(1, \gamma, \eta), g_{\tau}(\gamma, \eta)\right\}.
	\end{aligned}
\end{equation*}

Then for any $\tau \in (0, \frac{1}{2}]$, we numerically select $\gamma$ and $\eta$ to maximize $f_{\tau}(\cdot)$ and $g_{\tau}(\cdot)$, and we obtain Figure~\ref{fig:dense-ratio}. 

\begin{figure}
\centering
\resizebox{0.6\columnwidth}{!}{%
\inputtikz{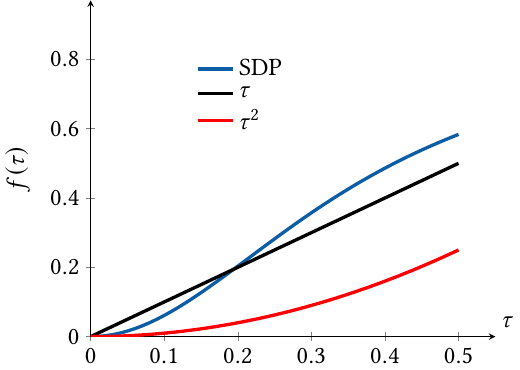}
}
\caption{Our approximation ratios for $\kdense$ where $\tau = \frac{k}{n}$. We compare the approximation ratio $f(\tau)$ with $\tau$ and $\tau^2$.
Here, $f(\tau)$ is the worst case approximation ratio of our \sdp-based
algorithm for the \kdense problem. The approximation ratio for \dskc is $f(\tau) \frac{1-c}{1+c}$.}
\label{fig:dense-ratio}
\end{figure}

In particular, if we set $\gamma = 0.920$ and $\eta = 1.65$, we get $f(\tau = 0.5) > 0.58$. 
Note if we set the same parameter as~\cite{DBLP:journals/jal/feigel01} ($\gamma
	= 3.87$ and $\eta = \frac{0.65}{3.87}$), we obtain the same approximation
ratio as \cite{DBLP:journals/jal/feigel01} obtained for \dks with $f(\tau = 0.5) = 0.517$. 

\subsection{Pseudocode for {\sdp}-based algorithm for \maxcutkc (in Section~\ref{sec:general:sdp})}
\label{appendix:maxcutkc:sdp:pseudocode}

The pseudocode for the {\sdp}-based algorithm for \maxcutkc, as discussed in Section~\ref{sec:general:sdp}, 
is illustrated in Algorithm~\ref{algo:maxcutkc-sdp}.

\begin{algorithm}[H]
\caption{An \sdp-based algorithm for \maxcutkc}
\label{algo:maxcutkc-sdp}
\begin{algorithmic}[1]
\Require{$(G(V, E, w), \OldSet, k)$, number of iterations $N$}
\State \textbf{Step 1: Solve SDP.} Solve the \sdp~\eqref{sdp:densest-subgraph} to obtain the solution $\vecv_0, \ldots, \vecv_n$
\For{$i = 1$ to $N$}
\State \textbf{Step 2: Hyperplane Rounding.} Obtain the indicator vector $\overline{\vecx}$
\State Sample a unit vector $\vecr$ uniformly randomly from $\mathcal{S}_n$
\ForAll{$i \in \{0, \ldots, n\}$}
    \If{$(\vecr \cdot \vecv_0)(\vecr \cdot \vecv_i) \geq 0$}
        \State $\overline{x}_i \leftarrow 1$
    \Else
        \State $\overline{x}_i \leftarrow -1$
    \EndIf
\EndFor
\State Set $\SelectSet \leftarrow \{i \colon \overline{x}_i \neq x^0_i\}$
\State \textbf{Step three, Fix $\SelectSet$:} Set $\FixSelectSet \gets \SelectSet$ 
\While{$\abs{\FixSelectSet} \neq k$}
    \If{$\abs{\FixSelectSet} < k$}
        \State $u^* \leftarrow \arg\max_{u \in V \setminus \FixSelectSet} \cut{\OldSet \symm (\FixSelectSet \cup \{u\})}$
        \State $\FixSelectSet \leftarrow \FixSelectSet \cup \{u^*\}$
    \ElsIf{$\abs{\FixSelectSet} > k$}
        \State $u^* \leftarrow \arg\max_{u \in \FixSelectSet} \cut{\OldSet \symm (\FixSelectSet \setminus \{u\})}$
        \State $\FixSelectSet \leftarrow \FixSelectSet \setminus \{u^*\}$
    \EndIf
\EndWhile
\EndFor 
\State \Return{$\FixSelectSet$}
\end{algorithmic}
\end{algorithm}

\subsection{Proof of Theorem~\ref{thm:max-cut-sdp}}
\label{appendix:mc-proofs:sdp}

The pseudocode for our algorithm can be found in
Algorithm~\ref{algo:maxcutkc-sdp}. In the present section, we demonstrate that
Algorithm~\ref{algo:maxcutkc-sdp} provides a constant factor approximation to
the \maxcutkc problem.
Note that the first two steps (i.e., solving the SDP and hyperplane rounding) are identical to those of Algorithm~\ref{algo:sdp-k-dense}. 
The third step (i.e., fixing $\SelectSet$) is identical to the third step of Algorithm~\ref{algo:maxcutkc-blackbox}.

\subsubsection{Fixing $\SelectSet$}
\label{appendix:mc-proofs:sdp:fix-c}
Again we let $\abs{\SelectSet} = \mu n$ and $k = \tau n$. 
We notice the greedy procedure (i.e., fixing $\SelectSet$) that appears in Algorithm~\ref{algo:maxcutkc-sdp} is the same as the counterpart that appears in Algorithm~\ref{algo:maxcutkc-blackbox}.
Thus, we can directly apply Lemma~\ref{lem:cut-blackbox-deduction} to derive a lower bound on the cut for fixing $\SelectSet$, as in Lemma~\ref{lemma:cut-move-all}.

\begin{restatable}{lemma}{cutmoveall}
	\label{lemma:cut-move-all}
 For \maxcutkc, $\SdpRnd \geq \min \{\frac{\tau^2}{\mu^2}(1 - \frac{1}{\Theta(n)}),  \frac{(1- \tau )^2}{(1 - \mu )^2} (1 - \frac{1}{\Theta(n)})\} \overline{\SdpRnd}$.
	In particular, when $\mu > \tau$, the cut is lower bounded by 
	$\frac{\tau^2}{\mu^2}(1 - \frac{1}{\Theta(n)})\cut{\TempSet}$;
	when $\mu < \tau$, the cut is lower bounded by 
	$\frac{(1- \tau )^2}{(1 - \mu )^2} (1 - \frac{1}{\Theta(n)})\cut{\TempSet}$.
\end{restatable}

\begin{proof}
	The proof is identical to Lemma~\ref{lem:cut-blackbox-deduction}. 
\end{proof}

\subsubsection{Approximation ratio}
\label{appendix:mc-proofs:sdp:approx}
Let $k = \tau n$ and let $\abs{\SelectSet} = \mu n$. 
{
Again, we analyze the approximation ratio of \maxcutkc when $k \leq n/2$.
By symmetry, we can obtain the approximation ratio of \maxcutkc when $k > n/2$.
}

We start with the following corollary of Lemma~\ref{lem:cut-bound-z}.
Notice that the corollary slightly differs from Lemma~\ref{lem:cut-bound-z} in
terms of introducing different random variables $Z_1$ and $Z_2$ {for the two cases $\abs{\SelectSet} \geq k$ and $\abs{\SelectSet} < k$, respectively.} 
{This design helps} to obtain better approximation ratios for \maxcutkc {by selecting better parameters for the two cases separately}. 

\begin{corollary}
	\label{cor:cut-bound-z-two-case}
	{Consider the constants $\gamma_1, \gamma_2 \in [0.1, 5]$ and $\eta_1, \eta_2 \in \left[\frac{1-\alpha+0.01}{\alpha}, +\infty\right)$.} 
	Define the random variables $Z_1$ and $Z_2$ as follows:
	\begin{displaymath}
		\begin{cases}
			Z_1 = \frac{\cut{\TempSet}}{W^*_{SDP}} + \gamma_1 \eta_1 \frac{n-\abs{\SelectSet}}{n-k} + \gamma_1 \frac{\abs{\SelectSet}(2k - \abs{\SelectSet})}{n^2}, & \text{if } \abs{\SelectSet} > k, \\ 
			Z_2 = \frac{\cut{\TempSet}}{W^*_{SDP}} + \gamma_2 \eta_2 \frac{\abs{\SelectSet}}{k} + \gamma_2 \frac{\abs{\SelectSet}(2k - \abs{\SelectSet})}{n^2}, & \text{if } \abs{\SelectSet} < k.
		\end{cases}
	\end{displaymath}
	{If the hyperplane rounding is repeated} at least $N = 2\cdot\frac{1 - p + \epsilon p}{\epsilon p} \log(\frac{1}{\epsilon})$ times, where $p$ is a constant with a lower bound of $0.0005$, then one of the following holds: either $\abs{\SelectSet} \geq k$ {occurs} at least $\frac{N}{2}$ times, and with probability at least $1 - \epsilon$, $Z_1 \geq (1-\epsilon) [\alpha + \gamma_1 \eta_1 \alpha + \gamma_1 (\alpha (1 - \tau)^2 - 1 + 2\tau)]$; or $\abs{\SelectSet} < k$ {occurs} at least $\frac{N}{2}$ times, and with probability at least $1 - \epsilon$, $Z_2 \geq (1-\epsilon) [\alpha + \gamma_2 \eta_2 \alpha + \gamma_2 (\alpha (1 - \tau)^2 - 1 + 2\tau)]$.
\end{corollary}

Next, we use Corollary~\ref{cor:cut-bound-z-two-case} to bound our approximation
ratio by solving
\begin{equation}
	\max_{\gamma \in [0.1, 5], \eta \in \left[\frac{1-\alpha + 0.01}{\alpha}, +\infty\right)} \min_{\mu \in [0 , 1]} \kappa_{\tau}(\mu) \SdpRatio_{\tau}(\mu, \gamma, \eta).
\end{equation}
 

Notice that according to Lemma~\ref{lemma:cut-move-all}, $\kappa_{\tau}(\mu) = \frac{\tau^2}{\mu^2} (1 - \frac{1}{\Theta(n)})$ when {$\abs{\SelectSet} \geq k$ i.e., $\mu \geq \tau$}, and $\kappa_{\tau}(\mu) = \frac{(1- \tau)^2}{(1-\mu)^2}(1 - \frac{1}{\Theta(n)})$ when {$\abs{\SelectSet} < k$ i.e., $\mu < \tau$}. 
According to Corollary~\ref{cor:cut-bound-z-two-case}, 
\begin{equation}
\lambda_{\tau}(\mu, \gamma_1, \eta_1) \geq \alpha + \gamma_1 \eta_1 \alpha + \gamma_1 (\alpha (1 - \tau)^2 - 1 + 2\tau) - \gamma_1 \eta_1 (1-\mu)\frac{1}{1-\tau} - \gamma_1 \mu (2\tau - \mu) - \frac{1}{\Theta(n)},
\end{equation}
when $\mu \geq \tau$;
and 
\begin{equation}
\lambda_{\tau}(\mu, \gamma_2, \eta_2) \geq \alpha + \gamma_2 \eta_2 \alpha + \gamma_2 (\alpha (1 - \tau)^2 - 1 + 2\tau) - \gamma_2 \eta_2 \frac{\mu}{\tau} - \gamma_2 \mu (2\tau - \mu) - \frac{1}{\Theta(n)},
\end{equation}
when $\mu < \tau$. 

We ignore the factor $(1 - \frac{1}{\Theta(n)})$ in our computation as it only introduces a small multiplicative error in our result.

\emph{Case 1:} $\abs{\SelectSet} \geq k$. This implies that $1 \geq \mu \geq
\tau$. In this domain of $\mu$, let us define the function $f^1_{\tau}(\mu, \gamma_1, \eta_1) = \kappa_{\tau}(\mu)\lambda_{\tau}(\mu, \gamma_1, \eta_1)$. 
The approximation ratio we get from this case is 
$$f^1(\tau) = \max_{\gamma_1 \in [0.1, 5], \eta_1 \in \left[\frac{1-\alpha+0.01}{\alpha}, +\infty\right)} \min_{\mu \in [\tau, 1]} f^1_{\tau}(\mu, \gamma_1, \eta_1).$$

\emph{Case 2:} $\abs{\SelectSet} <k$. This implies $\tau > \mu > 0$. In this
domain of $\mu$, let us define the function $f^2_{\tau}(\mu, \gamma_2, \eta_2) = \kappa_{\tau}(\mu)\lambda_{\tau}(\mu, \gamma_2, \eta_2)$. 
The approximation ratio we get from this case is 
$$f^2(\tau) =\max_{\gamma_2 \in [0.1, 5], \eta_2 \in \left[\frac{1-\alpha+0.01}{\alpha}, +\infty\right)} \min_{\mu \in (0, \tau]} f^2_{\tau}(\mu, \gamma_2, \eta_2).$$ 

\emph{Combining Case 1 and Case 2}, consider $1 \geq \mu > 0$. 
We define the function 
\begin{equation*}
	\begin{aligned}
&h(\tau) = \min\left\{f^1(\tau), f^2(\tau) \right\},
	\end{aligned}
\end{equation*}
and note that $h(\tau)$ is the approximation ratio we get in the end. We plot
$h(\tau)$ in Figure~\ref{fig:cut-ratio}.

\begin{figure}
\centering
\resizebox{0.6\columnwidth}{!}{%
\inputtikz{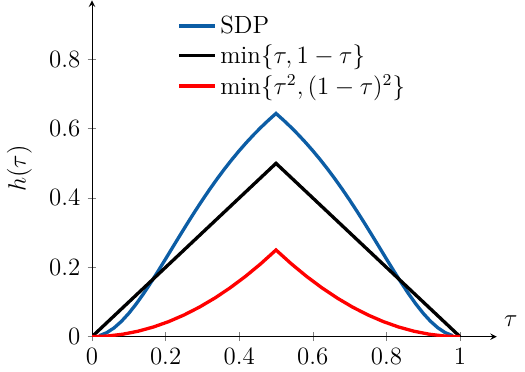}
}
\caption{Our approximation ratios for \maxcutkc, where $\tau = \frac{k}{n}$. We
compare the approximation ratio $h(\tau)$ of our \sdp-based algorithm with $\min\{\tau, 1-\tau\}$ and $\min\{\tau^2, (1-\tau)^2\}$.}
\label{fig:cut-ratio}
\end{figure}

\clearpage 

\clearpage

\section{Omitted content and proofs of Section~\ref{sec:max-cut:sos}}
\label{sec:sos_detailed}
In this section, we elaborate on the sum-of-squares (SoS) algorithm from Section \ref{sec:max-cut:sos}. 
In particular, we apply the method proposed by ~\cite{DBLP:conf/soda/RaghavendraT12}.
In detail, we will explain how to relax \maxcutkc into the Lasserre \sdp hierarchy (or, Lasserre hierarchy), how to obtain globally uncorrelated solutions to the Lasserre hierarchy, and how to apply a rounding scheme to the solutions so that the cardinality constraint is almost satisfied. 

For a detailed description and properties of the Lasserre \sdp hierarchy, we refer the readers to~\citep{DBLP:journals/siamjo/Lasserre02,rothvoss2013lasserre}.
For more technical details in this section, we refer the readers to  
\citep{DBLP:journals/talg/AustrinBG16,DBLP:conf/soda/RaghavendraT12,Chlamtac2012}.

The remainder of this section is structured as follows.
Section~\ref{subsection:lassere_hierarchy:preliminary} briefly introduces properties of the Lasserre hierarchy and mutual information to be applied in the following analysis. 
Section~\ref{appendix:lasserre:program} formulates a degree $d$ Lasserre hierarchy program for \maxcutkc.
Section~\ref{appendix:lasserre:decorrelate} describes how to obtain a \emph{decorrelated} solution via conditioning.
Section~\ref{appendix:lesserre:rounding}
formulates the degree $2$ Lasserre hierarchy and the hyperplane rounding algorithm we apply.
Section~\ref{appendix:lasserre:roundingcombined} proves
Lemma~\ref{lem:bound-variance} that we presented in the main content.

\subsection{Preliminaries}
\label{subsection:lassere_hierarchy:preliminary}

\subsubsection{Notation}
\label{subsection:lassere_hierarchy:notations}
Let $G=(V,E,W)$ be an undirected weighted graph,
and let $\vecx^0 \in \left \{-1,1\right\}^{V}$ 
encode the initial partition, i.e., $(\OldSet, \OldComSet)$. 
The $i$-th entry of $\vecx^0$ is denoted by $x^0_i$.
Let $S \subseteq V$, and $\mu_{S}$ be a (pseudo)-distribution on $\{-1, 1\}^S$. We will elaborate on its properties in Section~\ref{subsection:lassere_hierarchy_properties}.
Let $\vecx \in \left \{-1, 1\right\}^{V}$ be a random vector that encodes a random partition.
We use $\vecx_{S}$ to indicate the partition constrained on $S$.
Moreover, for $\xi \in \{-1, 1\}^S$, we denote by $\Pr_{\mu_{S}}(\vecx_{S} = \xi)$ the probability of assigning $\vecx_{S}$ with $\xi$ under the distribution $\mu_{S}$.

Let $k \leq \abs{V}$ be the number of refinements and $k = \tau n$.
We let $\varsigma=(n-2k)/n, \varsigma\in (-1,1)$. 
Let us fix $\varepsilon>0$ to be a small constant, and let $\delta=(\varepsilon/3)^{96}$. 
Let $d'= \lceil 1+ \frac{1}{\delta} \rceil $. We will construct a degree $d=d'+2$ Lasserre hierarchy. 

\subsubsection{Properties of the Lasserre hierarchy} 
\label{subsection:lassere_hierarchy_properties}
\para{Consistent local distribution.} 
Any Lasserre hierarchy configuration can be written in terms of local
distributions \citep{DBLP:journals/siamjo/Lasserre02}. In particular, in the
degree $d$ Lasserre hierarchy, for each set $S \subseteq V, |S| \leq d$, there
is a distribution $\mu_S$ on $\left \{-1,1\right\}^S$. Furthermore, Lasserre
hierarchy ensures that for any two such
probability distributions, one on $\{-1, 1\}^S$, denoted as $\mu_S$, and one on
$\{-1, 1\}^T$, denoted as $\mu'_T$, their marginal distribution on $\{-1, 1\}^{S \cap T}$ is the same: $\mu_{S \cap T} = \mu'_{S \cap T}$. 
In other words, local probability distributions are consistent. 

\para{Semi-definiteness property.}
Local distributions of the Lasserre hierarchy also satisfy positive semi-definiteness properties.
In the degree $d$ Lasserre hierarchy, for each $S, T \subseteq V, |S \cup T| \leq d,$ and $\xi_1 \in \left \{-1,1\right\}^S,\xi_2 \in \left \{-1,1\right\}^T$, there exist vectors $\vecv_{S,\xi_1},\vecv_{T,\xi_2}$ such that 
\begin{equation*}
\Pr_{\mu_{S \cup T}} (\vecx_S = \xi_1, \vecx_T = \xi_2) =  \vecv_{S,\xi_1} \cdot \vecv_{T,\xi_2}, 
\end{equation*}
where $\vecx_T,\vecx_S$ are random variables taking values in $\left
\{-1,1\right\}^S, \left \{-1,1\right\}^T$, respectively, and $\cdot$ denotes the
dot product of two vectors. 

\para{Conditioning of assignments.}
One property that we will use in our algorithm is that the local probability distributions can be conditioned.
In particular, any event $\vecx_S = \zeta$
can be conditioned on an 
an event $\vecx_U = \xi$ given that $\abs{S \cup U} \leq d$ and $\Pr_{\mu_{S \cup U}}(\vecx_U = \xi) \neq 0$. 

After the conditioning, we obtain a new 
Lasserre hierarchy of degree $d-|U|$ with local distributions $\left \{\bar{\mu}_S\right\}_{S \subseteq V,|S| \leq d-|U|}$.
In particular, 
\begin{equation*}
\begin{split}
	\Pr_{\bar{\mu}_{S}}\left( \vecx_S = \zeta \right) 
 = \Pr_{\mu_{S \cup U }}  \left( \vecx_S = \zeta \mid \vecx_U = \xi\right)
 = \frac{\vecv_{S, \zeta} \cdot \vecv_{U, \xi}}{\vecv_{U, \xi} \cdot \vecv_{U, \xi}}.
\end{split}
\end{equation*}

\subsubsection{Mutual information}
\label{subsection:information_theory}
Next, we illustrate how to use conditioning to reduce the correlation between fractional solutions of the vertices, allowing us to round values of different vertices \emph{almost independently on average}. 
We need to use mutual information to measure the correlation between vertices.
\begin{definition}
	Let $x_i,x_j$ be two random variables taking values in $\left \{-1,1\right\}$, which indicates the partition of node $i$ and $j$.
    Let $S \supseteq \left \{i,j\right\},U \ni i,T \ni j$ in place of $\left \{i,j\right\}, \left \{i\right\}, \left \{j\right\}$. The \emph{mutual} information between $x_i$ and $x_j$ is then defined as 
	\begin{equation*}
		I(x_i,x_j) =  \sum_{a,b \in \{-1,1\}} \Pr_{\mu_S}(x_i=a,x_j=b) \cdot \log\left(\frac{\Pr_{\mu_S}(x_i=a,x_j=b)}{\Pr_{\mu_U}(x_i=a)\Pr_{\mu_T}(x_j=b)}\right) .
	\end{equation*}
\end{definition}
This follows from the fact that probability distributions are consistent under intersections, and hence we can take $S \cap \left \{i,j\right\}$  to obtain the 
distribution on $\left \{i,j\right\}$, etc. \par
We will aim to get 
\begin{equation*}
	\E_{i,j \sim V} \left[ I(x_i, x_j)\right] \leq \delta.
\end{equation*}

This goal will be achieved by conditioning and utilizing the observation from Lemma~\ref{lemma:low_mutual_information}.
In particular, Lemma~\ref{lemma:low_mutual_information} implies that the
following equation holds:
\begin{equation}
\label{eq:multral-info-goal}
\min_{\substack{\{i_1, \ldots, i_t\} \subseteq V, t \leq d' \\ \{a_{i_1}, \ldots, a_{i_t}\} \in \{-1, 1\}^{\{i_1, \ldots, i_t\}}}} \E _{i,j \sim V}\left[ I(x_i,x_j \mid x_{i_1} = a_{i_1},\ldots,x_{i_t} = a_{i_t}) \right] \leq 1/(d'-1).
\end{equation}

\begin{lemma}[Lemma 4.5 in~\citet{DBLP:conf/soda/RaghavendraT12}]
\label{lemma:low_mutual_information}
		There exists $t \leq d'$ such that 
		\begin{equation*}
			\E_{i_1,\hdots,i_t \sim V} \E _{i,j \sim V}\left[ I(x_i,x_j \mid x_{i_1},\hdots,x_{i_t}) \right] \leq 1/(d'-1).
		\end{equation*}
\end{lemma}
This lemma holds true for any Lasserre hierarchy over a binary\footnote{The original lemma was shown for nonbinary alphabets (alphabets of size $q$), with the right-hand side containing an additional $\log_2 q$ factor. }  alphabet. 

\subsection{Lasserre hierarchy program for \maxcutkc}
\label{appendix:lasserre:program}
The Lasserre hierarchy relaxation of degree $d$ for \maxcutkc is as follows:
\begin{equation}\label{eq:lassere_formulation}
	\begin{aligned}
		\max\, & \sum_{ i < j } W_{ij} \Pr_{\mu_{\{i,j\}}} (x_i \neq x_j)  \\
		\textrm{s.t.} \, & \sum_{j=1}^n\Pr_{\mu_{S \cup \{j\}}} (x_j = x_j^0 \cap \vecx_{S} = \xi) = (n-k) \Pr_{\mu_{S}}(\vecx_{S} = \xi), \\
  & \forall S \subseteq V, \xi \in \{-1, 1\}^S, |S \cup \{j\}| < d.
	\end{aligned}
\end{equation}

We can show that Equation~\eqref{eq:lassere_formulation} is indeed a relaxation of \maxcutkc. 
Let $\hat{\vecx}$ be any feasible solution of \maxcutkc, and notice that it satisfies the constraint $\sum_{i} x_i^0 \hat{x}_i = n-2k$. 
Now we independently assign that $\Pr_{\mu_{S \cup \{j\}}}(x_j = \hat{x}_j) = 1$, and check whether this assignment satisfies the constraint. 
Notice that for any $\vecx_{S}$ and $\xi$, if $\Pr_{\mu_{S \cup \{j\}}}(\vecx_{S} = \xi) = 0$, the constraint trivially holds. 
Otherwise, we condition the probability on $\vecx_{S} = \xi$, and we get 
\begin{equation*}
\begin{aligned}
    \sum_{j=1}^n\Pr_{\mu_{S \cup \{j\}}} (x_j = x_j^0 \mid \vecx_{S} = \xi) 
    &= \sum_{j=1}^n\Pr_{\mu_{S \cup \{j\}} } (\hat{x}_j = x_j^0 \mid \vecx_{S} = \xi) \\
    &\stackrel{(a)}{=} \sum_{j=1}^n\Pr_{\mu_{S \cup \{j\}}} (\hat{x}_j = x_j^0) \\
    &\stackrel{(b)}{=} n \frac{(n-k)}{n} = n-k. \\
\end{aligned}
\end{equation*}
In the above formulation, $(a)$ holds as the assignment is independent, and $(b)$ holds by the constraint from Equation~\ref{ip:densest-subgraph} that $\sum_{i} x_i^0 \hat{x}_i = n-2k$.

Notice that an equivalent \sdp can be formulated as follows. For every $S \subseteq V, |S \cup \{j\}| \leq d$, and $\xi \in \{-1, 1\}^S$, we introduce a 
variable $\vecv_{S, \xi}$.
Then we obtain:
\begin{equation}\label{eq:lassere_formulation_matrix}
	\begin{aligned}
		\max\, & \sum_{ i < j } W_{ij} (\vecv_{i,1} \cdot \vecv_{j,-1} + \vecv_{i,-1} \cdot \vecv_{j,1})  \\
		\textrm{s.t.} \, &  \sum_{j=1}^n x_j^0 \vecv_{j,1} \cdot \vecv_{S,\xi} - \sum_{j=1}^n x_j^0 \vecv_{j,-1} \cdot \vecv_{S,\xi} = (n-2k) \vecv_{S,\xi} \cdot \vecv_{S,\xi}, \\
  & \forall S \subseteq V, \xi \in \{-1, 1\}^S, |S \cup \{j\}| < d.\\
	\end{aligned}
\end{equation}

Since this is a semidefinite relaxation with a polynomial ($n^{O(d)}$) number of
constraints, we can obtain the optimal value of this program in polynomial time
up to precision $\delta>0$. Let us write $\Sol$ to denote the value of this solution and $\OPT$ to denote the optimal value \maxcutkc. We also use $\left \{\vecv_S\right\}_{S \subseteq V, |S| \leq d}$ to denote the vectors $\vecv_S$  under which the value $\Sol$ is attained, and hence we have 
\begin{equation*}
	(1-\delta)\OPT \leq  \sum_{ i < j } W_{ij} (\vecv_{i,1} \cdot \vecv_{j,-1} + \vecv_{i,-1} \cdot \vecv_{j,1}) = \Sol.
\end{equation*}

\subsection{Obtaining a ``decorrelated'' solution via conditioning}
\label{appendix:lasserre:decorrelate}

After we solve Equation~\eqref{eq:lassere_formulation}, we obtain the local distributions of $\vecx$, which are consistent on the marginal distributions up to $d$ entries. 
We will apply a conditioning process in order to obtain a Lasserre hierarchy
such that the mutual information of each pair of entries is bounded, i.e.,
$\E_{i,j \sim V} \left[ I(x_i, x_j)\right] \leq \delta$. In particular, a step in the conditioning 
the procedure works as follows: 
\begin{itemize}
	\item[] Pick a vertex $i$ uniformly at random, and fix its value to either $1$ or $-1$, according to its marginal distribution $\mu_{\{i\}}$. Update the probability distributions $\{\mu_{S}\}_{S \subseteq V}$ according to this fixing, as discussed in Subsection \ref{subsection:lassere_hierarchy_properties}.
\end{itemize}
Let us observe that the objective value of Equation \eqref{eq:lassere_formulation} does not change in expectation. 
Furthermore, if we denote $\tilde{\vecv}_{j, -1}, \tilde{\vecv}_{j, 1}, \tilde{\vecv}_{S, \xi}$ the vectors obtained after the conditioning,
then in expectation we obtain that
\begin{equation*}
\sum_{j=1}^n x_j^0 \tilde{\vecv}_{j,1} \cdot \tilde{\vecv}_{S,\xi} - \sum_{j=1}^n x_j^0 \tilde{\vecv}_{j,-1} \cdot \tilde{\vecv}_{S,\xi} = (n-2k) \tilde{\vecv}_{S,\xi} \cdot \tilde{\vecv}_{S,\xi}
\end{equation*}
is satisfied.

Next, we apply Lemma~\ref{lemma:low_mutual_information} to obtain Lemma~\ref{lemma:conditioning}, which is analogous to Lemma 4.6 from \citet{DBLP:conf/soda/RaghavendraT12}. The difference is that our problem \maxcutkc has a different constraint; hence, a different analysis on the lower bound on the objective value is needed in order to apply the Markov inequality.

\begin{lemma}\label{lemma:conditioning}
There exists an algorithm running in time 
\begin{equation*}
	O\left(n^{2}d' \left(\frac{en}{d'} \right)^{d'}2^{d'}\right),
\end{equation*}
which, given a Lasserre solution outlined in Equation~\eqref{eq:lassere_formulation} of degree $d= d'+2$ with value $\Sol$, finds 
a degree $2$  Lasserre solution with objective value 
$\Sol (1-\delta)$, and which satisfies $\E[I(x_i, x_j)] \leq \delta$.
\end{lemma}
\begin{proof}
	Using the probabilistic method, we first show that with some probability, conditioning of variables does not reduce the 
	objective value by more than $\delta \Sol$. Then, we show that with high probability, conditioning satisfies the inequality 
	\begin{equation*}
		\E[I(x_i, x_j)] \leq \delta.
	\end{equation*}
	The randomness in the expression above comes from the random choices in our conditioning algorithm. 
	Then, by union bound, we can conclude that there is a conditioning with value at least $\Sol - \delta \Sol$ and with low average mutual information. We then conclude the lemma by discussing how this conditioning can be found efficiently in the time described in the lemma statement.

	Let us now provide details. Let us denote with $W \in \mathbb{R}_+$ the
	value $W=\sum_{i<j} W_{ij}$. First, consider the conditioning procedure
	outlined in this section. Since we are sampling the variables according to
	their marginal distributions, the expected value of the program
	\eqref{eq:lassere_formulation} remains $\Sol$. Let us use $X$  to denote a
	random variable whose value is the value of the Lasserre hierarchy after $t$
	conditioning steps, where $t$ is the same number as the number from Lemma~\ref{lemma:low_mutual_information}. We can then calculate
	\begin{equation*}
		\pr\left( X < (1-\delta) \E[X] \right) = \pr\left(W-X \geq  W- (1-\delta) \E[X] \right).  
		\end{equation*}
	Now, since $X \in [0,W]$, by Markov's inequality we have 
	\begin{equation}\label{appendix_aleksa:eq:markov_x}
		\pr\left( X < (1-\delta) \E[X] \right) \leq \frac{W- \E[X]}{W-(1-\delta) \E[X]}  = \frac{1- \E[X]/W}{1-(1-\delta)\E[X]/W}.
	\end{equation}
	Let us use $\zeta=\E[X]/W$. In order to find a suitable upper bound on the expression above, we first prove a lower bound 
	on $\zeta$.

	For this, we consider a feasible solution to 
 Equation~\eqref{eq:lassere_formulation}: for each 
	$i \in V$, we independently assign $x_i^0$  to vertex $i$  according to probability $\Pr_{\mu_{i}}(x_i = x_i^0)= \frac{n-k}{n} =(1+\varsigma)/2$, 
	and $-x_i^0$ otherwise. 
 Let us now give a lower bound for $\E[X]$,  by giving a lower bound of a term 
	\begin{equation*}
		\Pr_{\mu_{\{i,j\}}} (x_i \neq x_j) = \frac{1-\varsigma^2}{2},
	\end{equation*}
	according to the assignment.
	Going back to Equation~\eqref{appendix_aleksa:eq:markov_x}, this shows that $\zeta=\E[X]/W$  satisfies $\zeta \geq (1-\varsigma^2)/2$.
	Hence, we can write
	\begin{equation*}
		\pr\left( X < (1-\delta) \E[X] \right) \leq \frac{1- \zeta}{1-(1-\delta)\zeta}  
		\leq \frac{1}{1+\delta\frac{1-\varsigma^{2}}{1+\varsigma^{2}}}.
	\end{equation*}
 
	\par
	Let us denote with $I$ the random variable $I=I(X_i,X_j \mid X_{i_1},\hdots,X_{i_t})$.  By Markov's inequality
	\begin{equation*}
		\pr \left( I \geq \sqrt{\frac{1}{d'-1}} \right) \leq  \frac{\E[I]}{\sqrt{\frac{1}{d'-1}} } \leq \sqrt{\frac{1}{d'-1}},
	\end{equation*}
	and therefore the probability of the randomized algorithm not finding a fixing for which $\E[I_{i,j}] \leq \sqrt{\frac{1}{d'-1}}$ is at most $\sqrt{\frac{1}{d'-1}}$. Hence, the probability of the algorithm failing is at most 
	\begin{equation*}
		\begin{split}
		\frac{1}{1+\delta \cdot (1-\varsigma^{2})/(1+\varsigma^{2})} + \sqrt{\frac{1}{d'-1}}
		= \frac{1}{1+\delta\cdot (1-\varsigma^{2})/(1+\varsigma^{2})}  + \frac{1}{4}\cdot \delta(1-\varsigma^2) <1,
		\end{split}
	\end{equation*}
	and therefore by union bound there is a fixing which satisfies $\E[I_{i,j}] \leq \delta$ and with value at least 
	$\Sol(1-\delta)$.

	\par 
Such a fixing, i.e., computing Equation~\eqref{eq:multral-info-goal}, can be found using brute-force search. The size of the space of all possible fixings can be calculated
	as follows. First, we observe that there\footnote{Since we do not know the value $t$  before hand, we need to test all $i=1,\hdots,d'$.} are
	$\sum_{i=0}^{d'} {n \choose i} $ possible variables. Furthermore, once we fix $i$  variables, the space of possible 
	values these variables can be assigned to is of size $2^{i}$. Hence, using the upper bound ${n \choose i} \leq (n e / i ) ^ i$, the size of this space is at most 
	\begin{equation*}
		d' \left( \frac{en}{d'}\right)^{d'} 2^{d'}.
	\end{equation*}
	Since calculating mutual information takes time $O(n^{2})$, this gives us the final running time of our algorithm.
\end{proof}

\subsection{Description of the hyperplane rounding algorithm}
\label{appendix:lesserre:rounding}
In this section, we will describe the hyperplane rounding algorithm to the solution of the degree $2$ Lasserre hierarchy, which is obtained by using the conditioning outlined in the proof of Lemma \ref{lemma:conditioning}. 
Let us first write the formulation of the degree $2$ Lasserre hierarchy in
Equation~\eqref{eq:lassere_formulation_2_degree}:
\begin{equation}\label{eq:lassere_formulation_2_degree}
	\begin{aligned}
		& \sum_{ i < j } W_{ij} (\vecv_{i,1} \cdot \vecv_{j,-1} + \vecv_{i,-1} \cdot \vecv_{j,1}) \geq \Sol (1-\delta), \quad \textrm{and}  \\
		 &  \sum_{j=1}^n x_j^0 \vecv_{j,1} \cdot \vecv_{i,1} - \sum_{j=1}^n x_j^0 \vecv_{j,-1} \cdot \vecv_{i,1} = (n-2k) \vecv_{i,1} \cdot \vecv_{i,1}, \quad \textrm{and}\\
  &  \sum_{j=1}^n x_j^0 \vecv_{j,1} \cdot \vecv_{i,-1} - \sum_{j=1}^n x_j^0 \vecv_{j,-1} \cdot \vecv_{i,-1} = (n-2k) \vecv_{i,-1} \cdot \vecv_{i,-1}.\\
	\end{aligned}
\end{equation}

If we define $\vecv_i \coloneqq \vecv_{i, -1} - \vecv_{i, 1}$, 
and $\vecv_0 \coloneqq \vecv_{i, -1} + \vecv_{i, 1}$, 
the solution presented in Equation~\eqref{eq:lassere_formulation_2_degree} implies a solution that satisfies the following constraints:

\begin{equation*}
	\begin{split}
		\frac{1}{2}	& \sum_{i<j}W_{ij} (1-\vecv_i \cdot \vecv_j) \geq \Sol \cdot (1- \delta), \quad \textrm{and}\\
		 &\sum_{i \in V} x_i^0 \vecv_i\cdot \vecv_0   = \tau n.
	\end{split}
\end{equation*}

Next, we describe the hyperplane rounding procedure, which differs from the ones
in the other parts of our paper. 
The procedure is essentially the same as used by \citet[Section 5.4]{DBLP:conf/soda/RaghavendraT12}, with a mild difference in how we set the bias $\kappa_i$. Here we set $\kappa_i \coloneqq \vecv_i \cdot \vecv_0$. 
The goal is to ensure that the cardinality constraint in \maxcutkc holds in expectation.

Let $\bar{\vecv}_i$ be defined as
\begin{equation*}
	\bar{\vecv}_i = \begin{cases}
		\frac{\vecv_i - \kappa_i \vecv_0 }{ \| \vecv_i - \kappa_i \vecv_0\| }, & \textrm{if } |\kappa_i| \neq 1,\\
		\textrm{a unit vector orthogonal to all other vectors}, 			   & \textrm{if } |\kappa_i| = 1.
	\end{cases} 
\end{equation*}
 Let us use $\varPhi$ to denote the cumulative distribution function of 
 a single variable Gaussian distribution, and let us define the assignment $\{\bar{x}_i\}_{i=1}^n$ as 
\begin{equation*}
	\begin{split}
	g \sim 		\textrm{standard $n$-dimensional Gaussian vector},\\
	\bar{x}_i = \begin{cases}
		-1, & \textrm{if } \bar{\vecv}_i \cdot g \leq \varPhi^{-1}\left( \frac{1-\kappa_i}{2} \right),\\
		1,  & \textrm{otherwise.}
	\end{cases}
	\end{split}
\end{equation*}

Notice that $\bar{\vecx}$ encodes the partition after the hyperplane rounding. 
We remark that that $\Exp{\bar{x_i}} = \kappa_i$. A computer-assisted proof shows that rounding the vectors $\vecv_i$  incurs a loss of at most $0.858$, as claimed in \citet[Section 5.4]{DBLP:conf/soda/RaghavendraT12}. 

We present our main result of this hyperplane rounding procedure in Lemma~\ref{lem:rounding-procedure}.

\begin{lemma}
\label{lem:rounding-procedure}
After the hyperplane rounding procedure, it holds that:
\begin{enumerate}
	\item[(a)] $\frac{1}{2} \sum_{i<j} W_{ij}(1 - \bar{x}_i \bar{x}_j) \geq 0.858 \frac{1}{2}\sum_{i<j}W_{ij} (1-\vecv_i \cdot \vecv_j)$, and
	\item[(b)] $\Exp{\sum_{i\in V} x_i^0 \bar{x}_i} = \tau n$.
\end{enumerate}
\end{lemma}

Denote by $\overline{\Rnd}$ the cut induced by the partition $\bar{\vecx}$
encodes. Then Part~(a) of the lemma implies that $\overline{\Rnd} \geq 0.858
\Sol (1-\delta)$, and Part~(b) implies that, in expectation, the cardinality
constraint of \maxcutkc is satisfied.

Lemma~\ref{lem:rounding-procedure} provides us with a cut that approximates the optimal solution well. However, the number of refinements satisfies the cardinality constraints only in expectation. 
In particular, $\sum_{i \in V} x_i^0 \cdot \bar{x}_i = (n-2k)$ might do not hold.

In the last step we show that $\sum_{i \in V} x_i^0 \cdot \bar{x}_i $ is very close to $(n-2k)$. Hence, 
we can change the signs of a very small number of $\bar{x}_i$  to obtain the final assignment $x_i$ for which the value of the objective function does not change much.

\subsection{Obtaining a rounded solution which almost satisfies constraints}
\label{appendix:lasserre:roundingcombined}
In the previous section, we proved that for our rounding procedure it holds that $\Exp{\sum_{i\in V} x_i^0 \bar{x}_i} = \tau n$.
In this section, we will show that $\sum_{i \in V} x_i^0 \cdot \bar{x}_i $ is indeed very close to $n-2k$, by bounding its variance. 

\begin{lemma}
\label{lem:varaince-bound}
    $\E_{i,j \in V}[I(x_i, x_j) ] \leq \delta$ implies that 
\begin{equation*}
  \Var\left[ \sum_{i \in V} x_i^0 \bar{x}_i  \right]	\leq C |V|^2 \delta^{1/12},
\end{equation*}
after we apply the hyperplane rounding procedure according to the above section, where $C$ is some constant. 
\end{lemma}
\begin{proof}[Partial proof]
    We only highlight the parts different from~\citet[Theorem 5.6]{DBLP:conf/soda/RaghavendraT12}.
    The differences arise as we consider the initial partition encoded as $\vecx_0$. 
    The remaining part of the analysis is the same.
	We obtain:
    \begin{equation*}
		\begin{split}
	\Var\left[ \sum_{i \in V} x_i^0 \bar{x}_i  \right]	
	&= \sum_{i, j \in V} (\mathbb{E}[\bar{x}_i x_i^0 \bar{x}_j x_j^0] - \mathbb{E}[\bar{x}_i x_i^0] \mathbb{E}[\bar{x}_j x_j^0]) \\
	&\leq \sum_{i, j \in V} 4 \max_{a \in \{-1,1\}, b \in \{-1, 1\}} \Pr(\bar{x}_i=a, \bar{x}_j = b) - \Pr(\bar{x}_i=a) \Pr(\bar{x}_i=b) \\
	&=\sum_{i, j \in V} \bigO \left(\sqrt{I(\bar{x}_i, \bar{x}_j)}\right) \\
        &\leq  C |V|^2 \delta^{1/12}.
		\end{split}	
    \end{equation*}
    The last inequality follows directly from the analysis of \citet[Theorem 5.6]{DBLP:conf/soda/RaghavendraT12}.  
\end{proof}

Hence, by Chebyshev's inequality, we have that 
\begin{equation*}
  \pr\left(  \left|\sum_{i \in V} x_i^0 \bar{x}_i - (n-2\cdot k) \right| \geq \delta^{1/48} |V|\right) \leq  \frac{|V|^2 C(\delta^{1/12})}{\delta^{1/24}|V|^2} = C\delta^{1/24}.
\end{equation*}

In other words, with probability $1- C\delta^{1/24}$ the values $\bar{x}_i$ satisfy 
\begin{equation}
	\left|\sum_{i \in V} x_i^0 \bar{x}_i - (n-2k) \right| \leq \delta^{1/48} |V|,
\end{equation}

Now we are ready to prove the main Lemma that we leave in the main content.
\boundvariance*

\begin{proof}
	The coefficient $0.858 (1-\delta)$ comes from the rounding algorithm described in the previous subsection and obtains a rounded solution of value at least $0.858$ from 
	the SDP value, in expectation. We need to repeat rounding $\Omega(\log(1/\delta))$
	times to ensure that with a high probability we obtain a rounded solution which incurs a loss of at most $0.858-\delta $. 
 The other $(1 - \delta)$-factor comes from the choice of precision when we solve the SDP formulation. 
Furthermore,
	for a single rounding step we have that 
\begin{equation}\label{eq:no_vertices_to_move}
	\left|\sum_{i \in V} x_i^0 \bar{x}_i - (n-2k) \right| \leq \delta^{1/48} |V|,
\end{equation}
with probability $1- C\delta^{1/24}$. Hence, for sufficiently small $\delta = (\varepsilon/3)^{96}$ by the union bound 
we can ensure that the properties stated in the lemma hold.
\end{proof}

\clearpage
\section{Additional experiments for \dskc}
\label{sec:add-exp:dense}

In this section, we present further details of our experiments. 
We run the experiments on a machine with an Intel Xeon Processor E5 2630 v4 and 
64GB RAM. Our algorithms are implemented in Python. 
{We use networkX for generating synthetic graphs, and for reading the real-world datasets.} 
We use Mosek to solve the semidefinite programs.

\subsection{Datasets}
\label{sec:add-exp:data}

The first four datasets comprise Wikipedia politician page networks, and to
obtain them we follow the method of~\citet{neumann2022sublinear} with minor
changes: 
instead of constructing edges according to the hyperlinks in the whole Wikipedia page, 
we only constructing edges using the hyperlinks that appear in the articles.
By doing so, we avoid constructing edges simply because the two politicians belong to the same category. 
These networks represent politicians and party activities from different countries, 
annotated according to their respective parties. 
We select each graph based on the country and choose all the nodes belong to one party 
as the vertices of the initial subgraph. 

The second type of datasets we utilize are obtained from SNAP networks~\citep{snapnets} with ground-truth communities. 
Specifically, we select \dblp\footnote{com-DBLP: \url{https://snap.stanford.edu/data/com-DBLP.html}}, 
\amazon\footnote{com-Amazon: \url{https://snap.stanford.edu/data/com-Amazon.html}}, and 
\youtube\footnote{com-Youtube: \url{https://snap.stanford.edu/data/com-Youtube.html}}. 
For these datasets, we always choose the largest community as the initial subgraph.

We also use synthetic datasets generated from the Stochastic Block Model (SBM) with planted communities, 
namely \balanced, \sparse, and \dense. Each graph has four communities. 
We generate \sparse and \dense using the same parameters for generating the SBM,
and they consist of three sparse communities and one dense community (see below
for the concrete parameter values we picked). They only differ in their initial subgraphs.
In \sparse, we use a sparse community as the initial subgraph,
while in \dense, we use the dense community as the initial subgraph. 
For \balanced, all four communities have equal density.
With this configuration, the densest subgraph for \balanced is the whole graph,
while for \sparse and \dense, the denset subgraph is the planted dense community.

\label{sec:add-exp:data:parameter}
We elaborate on the parameters of the three SBM graphs. 
In each dataset, we establish four communities of equal size. 
For \balanced, all four communities are created with an edge probability of $0.3$, 
and the inter-community edge probabilities are set to $0.1$.

In contrast, for the \sparse and \dense datasets, we create one dense community with edge probabilities $0.8$, 
alongside three sparse communities with edge probabilities $0.2$. The edge probability between communities for these two datasets remains at $0.1$.

Concerning the initial subgraph, we adopt a different community for each dataset. For \balanced, we choose any one of the planted communities. 
We select a sparse community for the \sparse dataset. Conversely, for \dense, the dense community serves as our initial subgraph.

\subsection{Additional experimental results for \dskc}
\subsubsection{Performance when initializing with ground-truth subgraphs with $k$ nodes removed}

In Figure~\ref{fig:appendix:density-out-select-ratio}, we present additional
experimental results when $k$~nodes were removed from a ground-truth subgraph.

\begin{figure}[t!]
	\centering
	\inputtikz{ds_plots/legend_move_out}
	\begin{tabular}{cccc}
		\resizebox{0.3\columnwidth}{!}{%
			\inputtikz{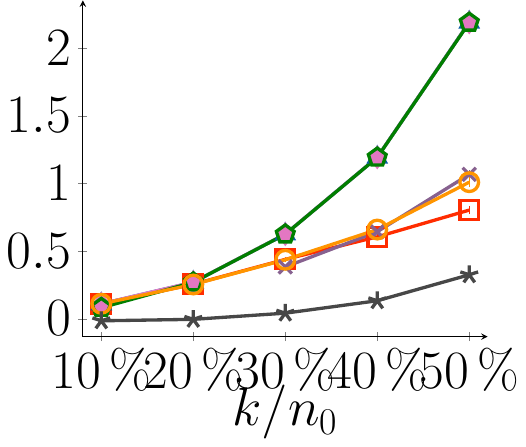}
		}&
		\hspace{-1.3em}
		\resizebox{0.29\columnwidth}{!}{%
			\inputtikz{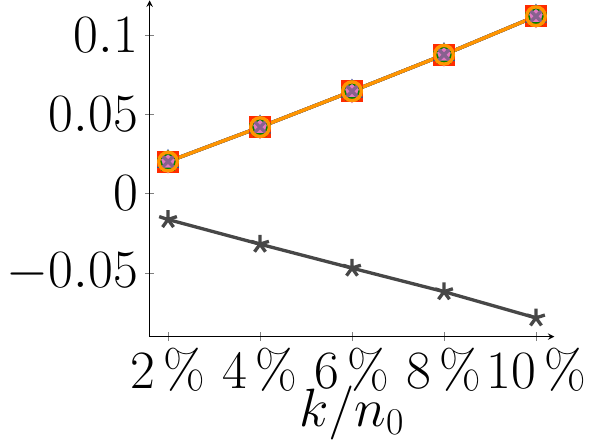}
		}&
		\hspace{-1.3em}
		\resizebox{0.29\columnwidth}{!}{%
			\inputtikz{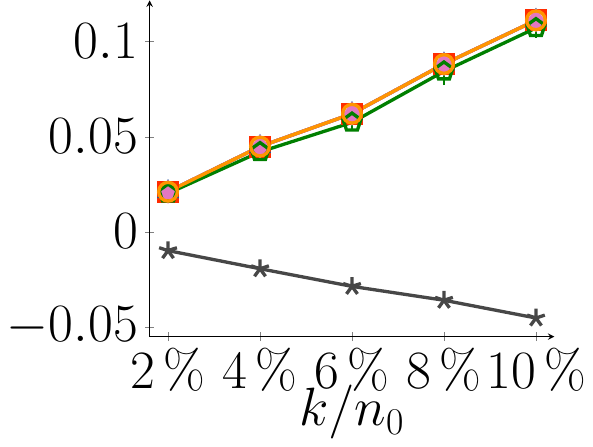}
		}\\
		{{\footnotesize \sparse}} &
		{{\footnotesize \dense}} &
		{{\footnotesize \balanced}} \\
	\end{tabular}
	\caption{Relative increase of density for varying values of $k$. We initialize $U$ as a ground-truth subgraph and then
	remove $k$ nodes uniformly at random.}
	\label{fig:appendix:density-out-select-ratio}
\end{figure}

\subsubsection{Performance when initializing with ground-truth subgraphs}

In Table~\ref{tab:appendix:move-out-10} and Figure~\ref{fig:appendix:density-not-out-select-ratio}, 
we present additional experimental results when $U$~is chosen a ground-truth
subgraph.

\begin{figure}[t!]
	\centering
	\inputtikz{ds_plots/legend_not_move_out}
	\begin{tabular}{cccc}
		\resizebox{0.3\columnwidth}{!}{%
			\inputtikz{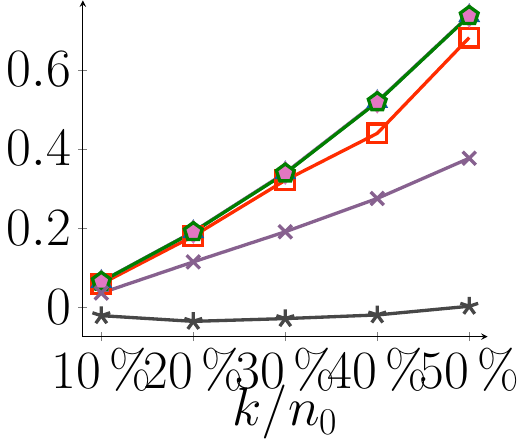}
		}&
		\hspace{-1.3em}
		\resizebox{0.29\columnwidth}{!}{%
			\inputtikz{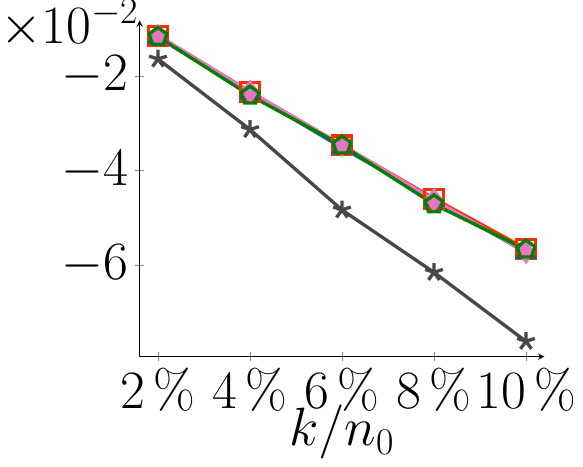}
		}&
		\hspace{-1.3em}
		\resizebox{0.29\columnwidth}{!}{%
			\inputtikz{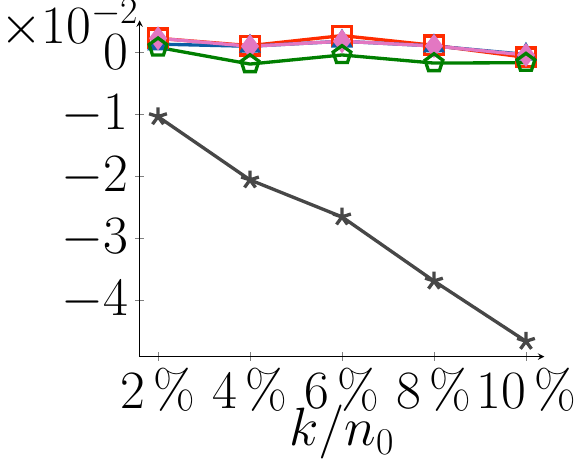}
		}\\
		{{\footnotesize \sparse}} &
		{{\footnotesize \dense}} &
		{{\footnotesize \balanced}} \\
	\end{tabular}
	\caption{Relative increase of density for varying values of {$k$}. 
	We initialized $U$ as a ground-truth subgraph.}
	\label{fig:appendix:density-not-out-select-ratio}
\end{figure}

\subsubsection{Running time dependency on $k$}
\label{sec:add-exp:running-time-by-k}

In Figure~\ref{fig:density-out-synthetic-k-time} we present the running time on 
synthetic datasets. 
Notice that for those synthetic datasets, we set a different range on $k$, i.e., 
$k \in [10, 50]$. This is because it gives us a larger range on $k$ 
compared to $k \in [2\%n_0, 10\%n_0]$, since $n_0$ for those synthetic 
datasets is $250$. 
We notice that \denseSDPalgo, \denseSDPMerge and \denseSQD are slower on 
\sparse than on \dense. We do not get results for \denseSQD on \balanced.  
In addition, \denseSDPalgo is always slower than \denseSDPMerge, 
since \denseSDPalgo contains more constraints.  

\begin{figure}[t!]
	\inputtikz{ds_plots/legend_move_out}
	\centering 
    \begin{tabular}{ccc}
        \resizebox{0.30\columnwidth}{!}{%
			\inputtikz{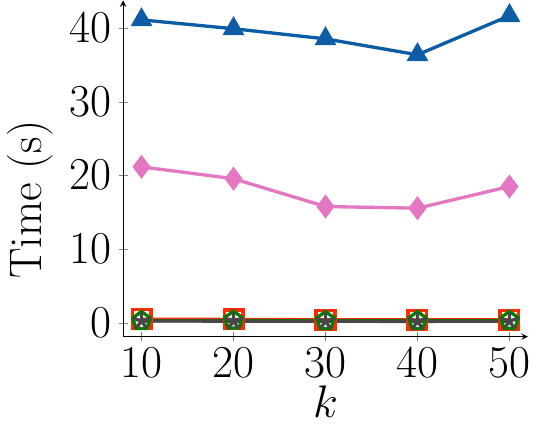}
		}&
        \resizebox{0.30\columnwidth}{!}{%
			\inputtikz{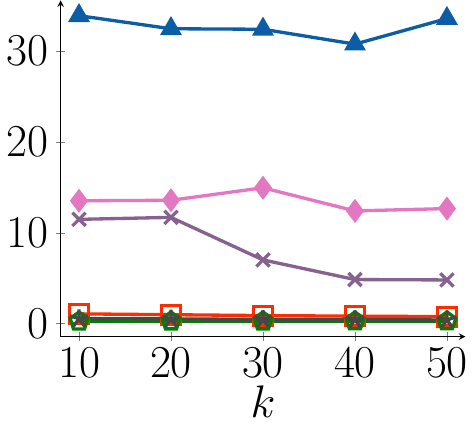}
		}&
		\resizebox{0.30\columnwidth}{!}{
			\inputtikz{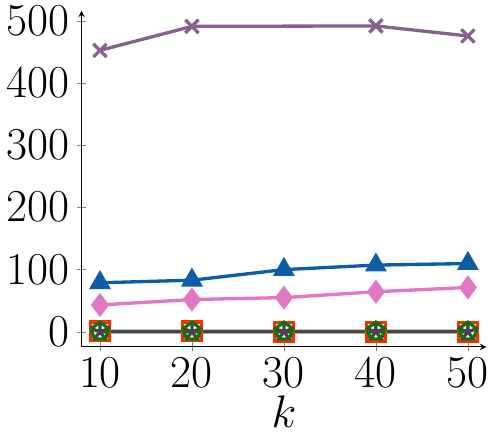}
		}
		\\
		{\balanced} &
		{\dense} &
		{\sparse} \\
	\end{tabular}
	\caption{Running time for varying values of $k$. We initialize $U$ as a ground-truth subgraph and then
	remove $k$ nodes uniformly at random. We set $k = 10, 20, 30, 40, 50$.
		}
	\label{fig:density-out-synthetic-k-time}
\end{figure}

\subsubsection{Running time dependency on $n$}
\label{sec:add-exp-running-time-by-n}
{
We vary the number of nodes in the synthetic datasets \dense and \sparse.
Figure~\ref{fig:density-out-country-ratio-n-time} shows the 
algorithms' running times for different values of $n$. 
The running times of \denseSDPalgo and \denseSDPMerge increase quadratically with $n$ 
due to the high cost of solving the \sdp, making these algorithms less scalable.
}

\begin{figure}[t!]
	\centering 
	\inputtikz{ds_plots/legend_move_out}
    \begin{tabular}{cc}
        \resizebox{0.30\columnwidth}{!}{%
			\inputtikz{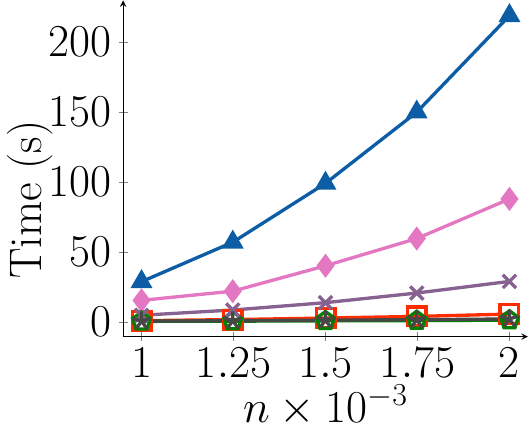}
		}&
        \resizebox{0.30\columnwidth}{!}{%
			\inputtikz{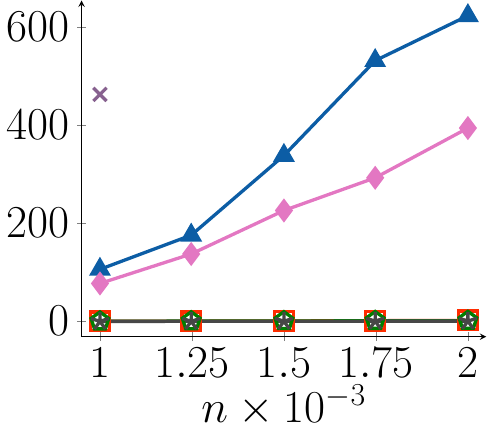}
		}
		\\
		(a)~{\dense} &
		(b)~{\sparse} \\
	\end{tabular}
	\caption{
		Running time for varying values of $n$. We initialize $U$ as a ground-truth subgraph and then
	remove $k= 10\%n_0$ nodes uniformly at random.
	We set $n = 1\,000, 1\,250, 1\,500, 1\,750, 2\,000$. }
	\label{fig:density-out-country-ratio-n-time}
\end{figure}

\subsection{An analysis on \sparse} 
\label{sec:add-exp:dense:sparse}
To better understand our experimental results on the \sparse dataset, we present
an analysis of two na\"ive baselines.
In particular, we consider \sparse where we 
initialize $U$ as a ground-truth subgraph and then remove $k$ nodes uniformly at random. 
In Figure~\ref{fig:sparse-analysis}
we plot the expected increase in density by two na\"ive baselines:
We present the first baseline in blue, which directly recovers $k$ removed nodes
(i.e., if $k=\abs{U}$ it does exactly the same as \denseinit).
We present the second baseline in orange, which selects $k$ nodes from a dense community uniformly at random.

We observe that when $k\leq 84$, the baseline which picks nodes from the sparse
community exhibits better performance. 
However, for larger values of $k$, the algorithm that select nodes from the dense community performs better. 
This figure can serve as a reference for us in understanding the selection of nodes by the actual algorithms.

\begin{figure}[t!]
	\centering 
	\includegraphics[width=1\columnwidth]{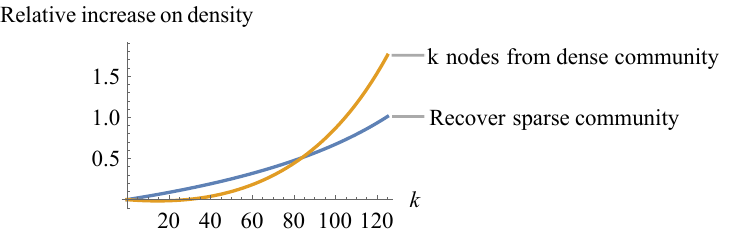}
	\caption{On dataset \sparse. 
	We first uniformly randomly remove $k$ nodes from the sparse community. 
	We compute the expected relative increase on density by recovering $k$ nodes from the sparse community, 
	or by select $k$ nodes from the dense community. 
		When $k = 84$ those two are equal.}
	\label{fig:sparse-analysis}
\end{figure}

In particular, a connection may be drawn between \sparse in Figure~\ref{fig:appendix:density-out-select-ratio} and Figure~\ref{fig:sparse-analysis}. 
Upon analyzing the performances of \denseSDPalgo, \denseSDPMerge, and \densePeelMerge, we find that when $k$ is small, 
their performances mirror that of the \emph{recovery of the sparse community}, in the sense that their performances are closely aligned with \denseinit. 
However, as $k$ increases, the performance of these three algorithms tends to resemble the \emph{selection of $k$ nodes from the dense community}.

In contrast, \denseGreedy always recovers the sparse community. 
Consequently, this algorithm does not exhibit a better performance when $k$ is large.

\subsection{A discussion of \denseSQD's performance}
\label{sec:add-exp:sqd}
Next, we provide a more detailed discussion of the performance of the \denseSQD
algorithm.

The \denseSQD algorithm contains two phases. 
In Phase 1, the algorithm constructs a dense graph with at least $\frac{k}{2}$ nodes. 
In Phase 2, the algorithm tries to satisfy the cardinality constraint. 
If the dense graph found from the phase 1 contains more than $k$ nodes, 
the algorithm in phase $2$ removes nodes according to the $\sigma$-elimination order; 
otherwise, the algorithm randomly adds nodes. 

We observe that \denseSQD performs better than \denseGreedy on the \dblp dataset
(see Table~\ref{tab:not-move-out-10}). 
This is because that on the \dblp dataset, Phase 1 of \denseSQD constructs a
dense while small subgraph,
and in the Phase 2 it only randomly adds extra nodes while still gaining larger density than \denseGreedy.

We also observe that on some datasets \denseSQD does not produce results. 
This is because the graph constructed after Phase 1 either does not have a
$\sigma$-quasi-elimination ordering,
or due to the high time and memory complexity of computing a
$\sigma$-quasi-elimination ordering.
For example, on small graphs, \balanced and \dense in Table~\ref{tab:not-move-out-10}, and \dense in Table~\ref{tab:appendix:move-out-10}, 
\denseSQD constructs a graph with significantly larger than $10\%\cdot n_0$
nodes in Phase 1.
In fact, the size of the graoh is equal to the size of the whole remaining graph. 

Even for those datasets where \denseSQD produces results, if computing the
$\sigma$-quasi-elimination order is necessary, it is very costly. 
In Figure~\ref{fig:density-out-synthetic-k-time} on \sparse, we see that running
\denseSQD is even more costly than running \denseSDPalgo. 
We should also note that for the largest datasets (which are relatively sparse),
such as \dblp, \denseSQD does compute the $\sigma$-elimination ordering and
obtains results.

\section{Experiments for Max-Cut with $k$~refinements}
\label{sec:exp:cut}

In this section, we conduct an extensive evaluation of our algorithms on a
variety of datasets to solve \maxcutkc. Our focus in this section will be two primary questions:

\begin{description}
\item[RQ1:] Do our algorithms consistently produce a larger cut when initialized with a random partition?
\item[RQ2:] Which algorithm is the best among our proposed algorithms?
\end{description}

\subsection{Algorithms} 
We use our SDP and Greedy algorithms and denote them as \cutSDPalgo and
\cutGreedy, respectively.

Additionally, we utilize the following black-box solvers for \maxcut to obtain
algorithms via the reduction in Theorem~\ref{thm:max-cut-black-box}:
\begin{enumerate}
\item A Greedy algorithm, where each node's side is chosen in a way that
	maximizes the cut at each step. We denote it by \cutBlackGreedy.
\item A Local Search algorithm which makes refinements to each node's side if
	the cut can be improved.  The algorithm either starts with a randomly
	initialized partition or uses the result of the Greedy algorithm as the
	initial partition.  We denote them
	as \cutBlackLocalOne and \cutBlackLocalTwo, respectively.
\item The \sdp algorithm of~\citet{goemans1995improved} which we denote by \cutBlackSDP.
\end{enumerate}

\subsection{Datasets} 
We use the same datasets as for \dskc. 
Since \dense and \sparse essentially have the same graph structure and since in
\maxcutkc the ground-truth communities are not relevant for us, we use \dense to represent them. 
The networks statistics are presented in Table~\ref{tab:statistics-cut}.

\begin{table*}[t]
	\centering
	\caption{\small{Network statistics, 
	and average relative increase of cut in $5$ runs with $k= 50$. 
	Here, $n$ and $m$ are the number of nodes and edges of the graph; $\mathsf{cut}_0$ is the average of $5$ random cuts.
	Next, \_ denotes that an algorithm does not finish in time (2 hours for SNAP datasets, and 30 minutes for other datasets).}}
	\label{tab:statistics-cut}
	\resizebox{0.85\textwidth}{!}{
		\begin{tabular}{@{}rrrrrrrrrrrrrr}
\toprule
\multirow{2}{*}{\textsf{dataset}} & \multicolumn{3}{c}{Network Property} & \multicolumn{2}{c}{Algorithms} & \multicolumn{4}{c}{Blackbox Methods}\\
\cmidrule(lr){2-4} \cmidrule(lr){5-6} \cmidrule(lr){7-10}
& $n$ & $m$ & $\mathsf{cut}_0$ & Greedy & SDP & B:Greedy & B:SDP & B:Local:I & B:Local:II \\
\midrule
\balanced & 1\,000 & 74\,940 &  37\,695.6 & 0.030 & \textbf{0.031} & \emph{0.021} & 0.020 & \emph{0.021} & \emph{0.021}\\
\dense & 1\,000    & 81\,531 &  40\,761.8 & 0.027 & \textbf{0.028} & 0.019 & 0.018 & \emph{0.020} & 0.019\\ 
\midrule
\es & 205 & 372 & 185.2               & 0.438 & \textbf{0.533} & 0.421 & \emph{0.462} & 0.386 & 0.426\\
\de & 768 & 3\,059 & 1\,531.6         & 0.157 & \textbf{0.182} & 0.118 & \emph{0.142} & 0.126 & 0.126\\
\gb & 2\,168 & 18\,617 &  9\,314.6    & 0.070 & \textbf{0.074} & 0.035 & \emph{0.049} & 0.043 & 0.038\\
\us & 3\,912 & 18\,359 & 9\,157.0     & 0.058 & \textbf{0.061} & \emph{0.035} & 0.031 & 0.034 & 0.034\\
\midrule 
\dblp & 317\,080   & 1\,049\,866 & 525\,028.4 & \textbf{0.002} & \_ & \emph{0.001} & \_ & \emph{0.001} & \emph{0.001}\\
\amazon & 334\,863    & 925\,872 & 462\,570.4 & \textbf{0.002} & \_ & \emph{0.001} & \_ & \emph{0.001} & \emph{0.001}\\
\youtube & 1\,134\,890 & 2\,987\,624 & 1\,493\,444.6 & \textbf{0.002} & \_ & \emph{0.000} & \_ & \_ & -0.000\\
\bottomrule
\end{tabular}

	}
\end{table*}

\subsection{Evaluation}
In our experiments, we proceeds as follows.
We uniformly randomly partition each node in the graph to either side, and
repeat this procedure $5$ times.  At each time, we run our algorithms and
calculate the cut.  We compute the relative increase on the initial given cut at
each run, and take the average of these $5$ runs. We present the results in
Table~\ref{tab:statistics-cut}. 

Here, it is worth noting that for \maxcut, a random partition already yields a
2-approximate solution compared to the maximum cut in expectation.
Thus, since we initialize the original cut randomly, this means that the given cut
will already be large (with a relatively large probability) and therefore
improving the cut significantly will be difficult for our algorithms.

\subsection{Performance}
In Table~\ref{tab:statistics-cut}, we report the average of the relative
increases of the cut value by different algorithms. 
In Figure~\ref{fig:cut-country}, we show the relative increase of the cut values for 
our Wikipedia politician networks for varying values of $k$. 

We notice from the plots and the table that, for all the settings of $k$, \cutSDPalgo is always 
the best, as long as \cutSDPalgo outputs a result.
As $k$ increases, there is a trend that \cutSDPalgo performs better. 
Even though \cutGreedy does not perform better than \cutSDPalgo, in many cases, 
it performs better than the other competitors. 
This addresses \textbf{RQ2}.

Additionally, the local-search based algorithm \cutBlackLocalTwo only very
slightly improves upon \cutBlackGreedy
(which it uses to obtain its initial partition)
and does not perform significantly better than \cutBlackLocalOne
(which uses a random initialization).
This suggests that the initial solution returned by \cutBlackGreedy is not very
good, as is also suggested by comparison with \cutGreedy and \cutSDPalgo.
However, we will see below that \cutBlackGreedy has a significantly lower
running time than \cutBlackLocalOne.
This addresses \textbf{RQ1}.

\begin{figure}[t!]
	\centering 
	\inputtikz{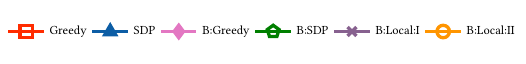}
    \begin{tabular}{cccc}
        \resizebox{0.25\columnwidth}{!}{%
			\inputtikz{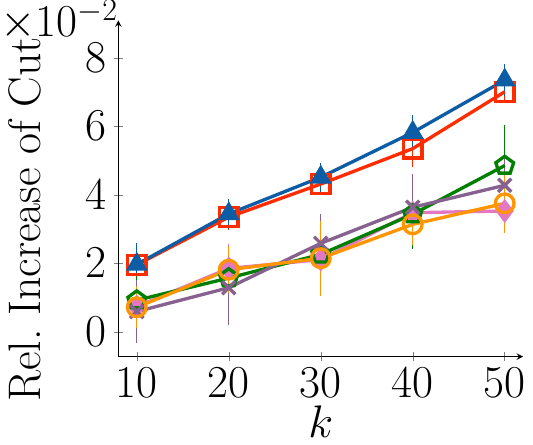}
		}&
		\hspace{-1.3em}
        \resizebox{0.23\columnwidth}{!}{%
			\inputtikz{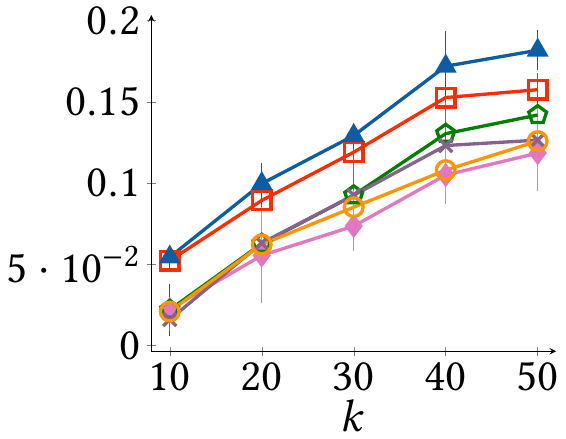}
		}&
		\hspace{-1.3em}
        \resizebox{0.23\columnwidth}{!}{%
			\inputtikz{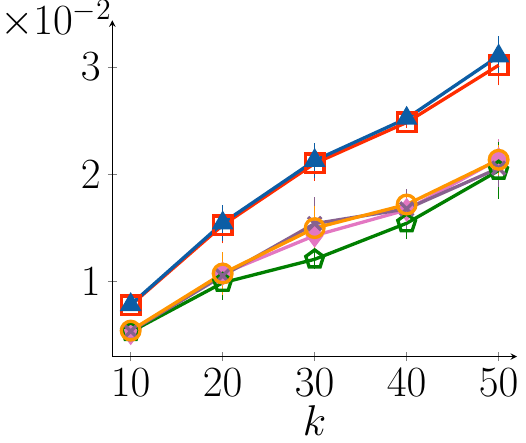}
		}&
		\hspace{-1.3em}
        \resizebox{0.23\columnwidth}{!}{%
			\inputtikz{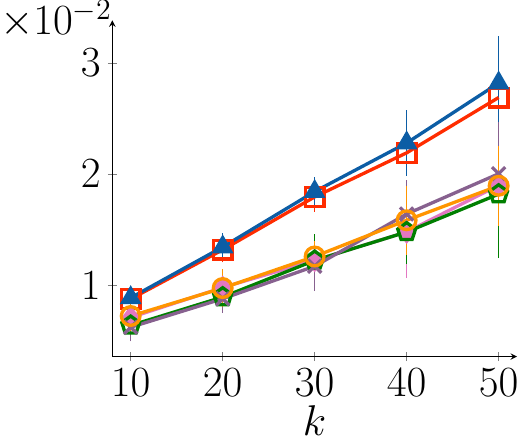}
		}
		\\
		\gb &
		\de &
		\balanced &
		\dense \\
        \resizebox{0.25\columnwidth}{!}{%
			\inputtikz{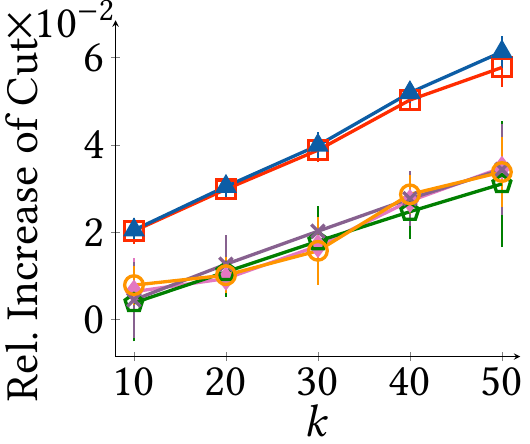}
		}&
		\hspace{-1.3em}
        \resizebox{0.23\columnwidth}{!}{%
			\inputtikz{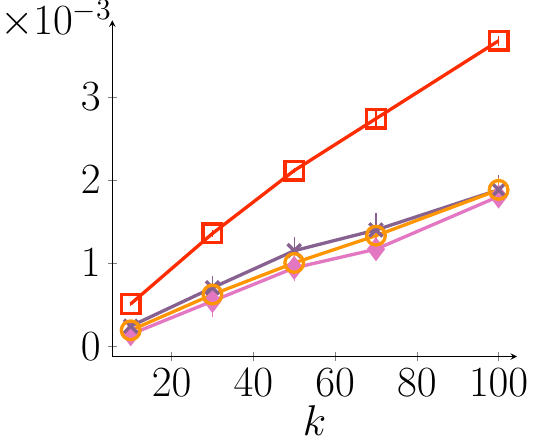}
		}&
		\hspace{-1.3em}
        \resizebox{0.23\columnwidth}{!}{%
			\inputtikz{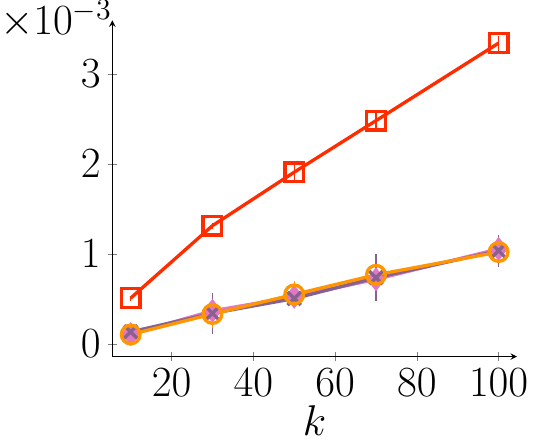}
		}&
		\hspace{-1.3em}
        \resizebox{0.23\columnwidth}{!}{%
			\inputtikz{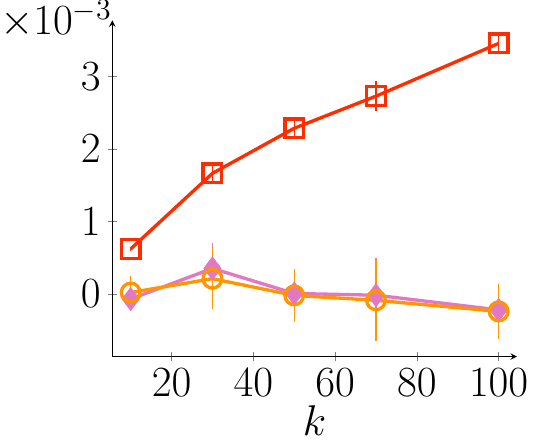}
		}
		\\
		\us &
		\dblp &
		\amazon &
		\youtube \\
	\end{tabular}
	\caption{Relative increase of the cut value. We used $k \in [10, 50]$ on smaller
		graphs, and $k \in [10, 100]$ for larger graphs.}
	\label{fig:cut-country}
\end{figure}

\sbpara{Running time dependency on $k$.}
In Figure~\ref{fig:cut-country-k-time} we present the running time on 
three datasets with varying values of $k$.  
We notice that \cutSDPalgo requires much longer time than \cutBlackSDP.
This is because \cutBlackSDP only needs to solve the standard Max-Cut 
\sdp, while \cutSDPalgo adds additional constraints, resulting in
longer running time. \cutBlackLocalOne requires more time than \cutBlackLocalTwo, this 
is because the local search algorithm with random initialization takes much
longer time to converge.

\begin{figure}[t!]
	\centering 
	\inputtikz{cut_plots/legend}
    \begin{tabular}{cccc}
        \resizebox{0.25\columnwidth}{!}{%
			\inputtikz{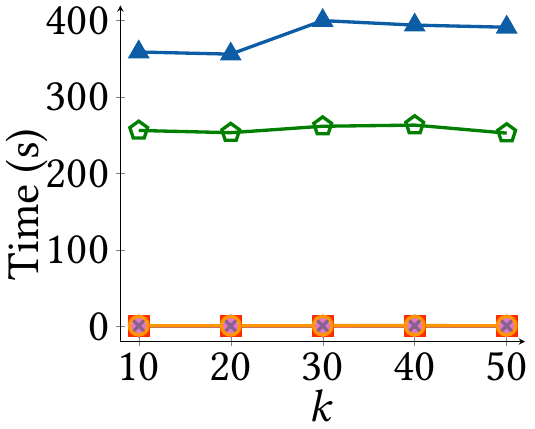}
		}&
		\hspace{-1.3em}
        \resizebox{0.23\columnwidth}{!}{%
			\inputtikz{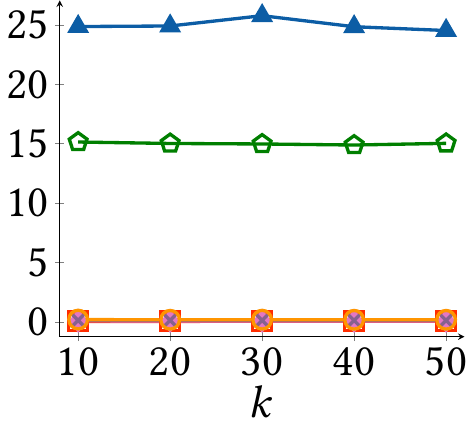}
		}&
		\hspace{-1.3em}
        \resizebox{0.23\columnwidth}{!}{%
			\inputtikz{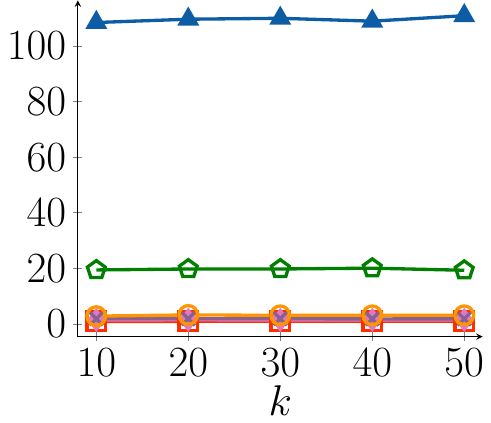}
		}&
		\hspace{-1.3em}
        \resizebox{0.23\columnwidth}{!}{%
			\inputtikz{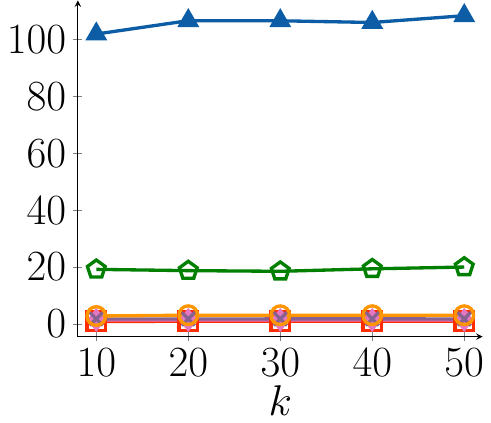}
		}
		\\
		\gb &
		\de &
		\balanced &
		\dense \\
        \resizebox{0.25\columnwidth}{!}{%
			\inputtikz{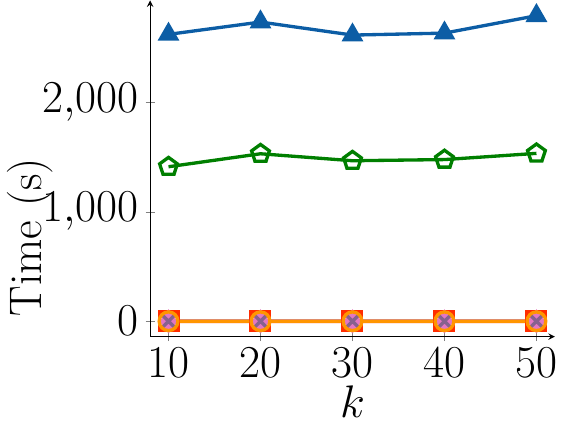}
		}&
		\hspace{-1.3em}
        \resizebox{0.23\columnwidth}{!}{%
			\inputtikz{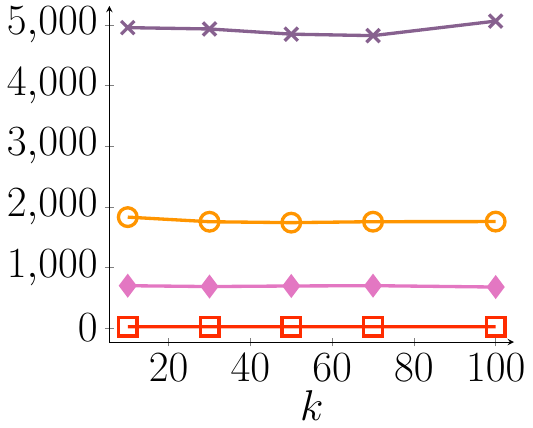}
		}&
		\hspace{-1.3em}
        \resizebox{0.23\columnwidth}{!}{%
			\inputtikz{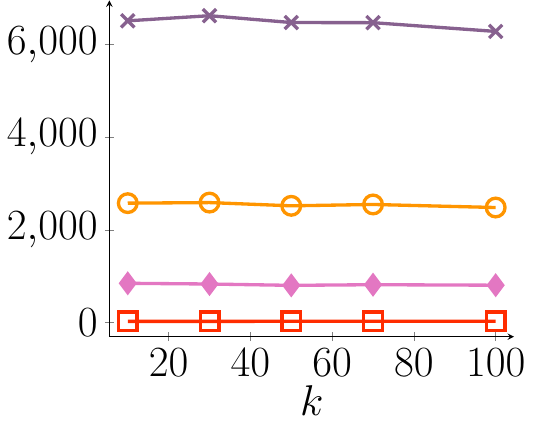}
		}&
		\hspace{-1.3em}
        \resizebox{0.23\columnwidth}{!}{%
			\inputtikz{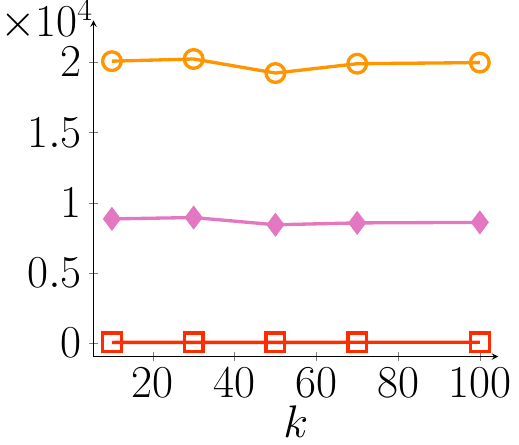}
		}
		\\
		\us &
		\dblp &
		\amazon &
		\youtube \\
	\end{tabular}
	\caption{Running time of our algorithms. We set $k \in [10,  50]$ on smaller graphs, and $k \in [10, 100]$ for larger graphs.}
	\label{fig:cut-country-k-time}
\end{figure}

\sbpara{Running time dependency on $n$.} 
\label{sec:cut-running-time-by-n}
We increase the size of synthetic datasets \dense and \balanced, 
and in Figure~\ref{fig:cut-n-time} we give the 
running time in terms of increasing $n$. 
\denseSDPalgo and \denseSDPMerge increase largely as $n$ increases, as we expect, indicating these two algorithms are not scalable.

\begin{figure}[t!]
	\centering 
	\inputtikz{cut_plots/legend}
    \begin{tabular}{cc}
        \resizebox{0.3\columnwidth}{!}{%
			\inputtikz{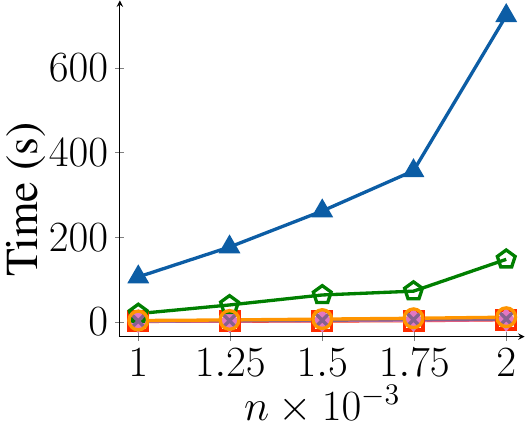}
		}&
		\hspace{-1.3em}
        \resizebox{0.3\columnwidth}{!}{%
			\inputtikz{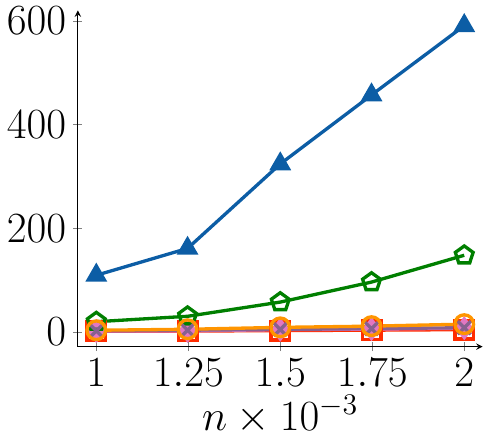}
		}
		\\
		(a)~\dense &
		(b)~\balanced \\
	\end{tabular}
	\caption{Running time of the algorithms for fixed $k = 50$ and varying $n = 1\,000, 1\,250, 1\,500, 1\,750, 2\,000$.}
	\label{fig:cut-n-time}
\end{figure}

\subsection{A case study}
\cite{matakos2020tell} propose a method to maximize the diversity of a social
network by flipping $k$ nodes' opinions, where each opinion is either $1$ or
$-1$.  In our problem setting, the term \emph{diversity} is mathematically equivalent to \emph{cut} 
and the nodes of the graph can be partitioned based on their opinions (where the
		cut is given by all nodes with opinion~$1$ on one side and all other
		nodes on the other side).
It is important to note that our problem setting differs from theirs in several ways: 
they operate under a budget constraint, whereas we operate under a cardinality constraint; 
they require an inequality constraint, while our problem necessitates an equality constraint. 
We executed their algorithm on our datasets by setting the cost of all nodes to $1$, 
allowing for a comparison of our methods.
The performance, running time, and the average number of changed nodes are presented in 
Table~\ref{table:diversity-comparison}. We note that we do not report results for~\us,
since the algorithm of \citet{matakos2020tell} runs out of memory on this dataset.

We observe that, somewhat surprisingly, the \cutSDPalgo algorithm proposed by \citet{matakos2020tell} 
takes significantly more time than the \cutSDPalgo algorithm proposed in our paper (as shown in Figure~\ref{fig:cut-country-k-time}). 
In addition, despite having the \emph{at most} constraint instead of the equality constraint, 
their algorithm does not yield better performance. 
Specifically, on the dataset \de, our \cutSDPalgo algorithm performs twice as well.

\begin{table}[t]
\centering
\caption{
	The results of the \cutSDPalgo algorithm proposed by \cite{matakos2020tell},
	where we set $k = 50$ on several datasets. 
	We use \emph{Matakos et al.} to indicate the average relative increase of the cut
	value with their approach in $5$ runs, and \emph{SDP} to indicate the result of our \cutSDPalgo-based algorithm.
	}
\label{table:diversity-comparison}
\begin{tabular}{@{}lllll}
\toprule
\textsf{dataset} & SDP & Matakos et al. & \textsf{Time (s)} & \textsf{\# changed nodes} \\
\midrule
\es & 0.533 & 0.408 & 22.1848 &  48.3\\
\de & 0.182 & 0.105 & 784.401 &  46.4\\
\gb & 0.074 & 0.031 & 5780.306 &  48.5\\
\bottomrule
\end{tabular}
\end{table}

\end{document}